\documentclass[11pt]{article}

\usepackage[T1]{fontenc}
\usepackage[utf8]{inputenc}
\usepackage{lmodern}
\usepackage{microtype}
 \usepackage{tikz}
 \usepackage{quantikz}
 \usepackage{subcaption}
 \usepackage{xcolor}
 
 \usetikzlibrary{arrows.meta,positioning}
 \definecolor{zxZ}{RGB}{41,167,64}   % Z-spider (green)
 \definecolor{zxX}{RGB}{220,50,47}   % X-spider (red)
 \tikzset{
 	spiderZ/.style={circle, draw=zxZ, fill=zxZ!20, thick, minimum size=10mm},
 	spiderX/.style={circle, draw=zxX, fill=zxX!20, thick, minimum size=10mm},
 	wire/.style={line width=1.1pt},
 }

\providecommand{\bra}[1]{\left\langle #1 \right\rangle}
\providecommand{\ket}[1]{\left| #1 \right\rangle}

\usepackage{amsmath, amssymb}
\usepackage{mathtools}
\usepackage{amsthm}

\usepackage{graphicx}
\usepackage{xcolor}

\usepackage{enumitem}
\usepackage{booktabs}

\usepackage[hidelinks]{hyperref}
\usepackage[capitalise,nameinlink]{cleveref}

\usepackage[margin=1in]{geometry}
\usepackage{algorithm,algpseudocode}

\usepackage{authblk}

\makeatletter
\@ifundefined{theorem}{\newtheorem{theorem}{Theorem}[section]}{}
\@ifundefined{proposition}{\newtheorem{proposition}[theorem]{Proposition}}{}
\@ifundefined{definition}{\newtheorem{definition}[theorem]{Definition}}{}
\@ifundefined{remark}{\newtheorem{remark}[theorem]{Remark}}{}
\@ifundefined{example}{\newtheorem{example}[theorem]{Example}}{}
\@ifundefined{convention}{}{}
\@ifundefined{lemma}{\newtheorem{lemma}[theorem]{Lemma}}{}
\@ifundefined{corollary}{}{}
\makeatother

\usepackage{xspace}
\DeclareRobustCommand{\ZX}{\textsc{ZX}\xspace}

\DeclareRobustCommand{\WPLZX}{\textsc{WPL–ZX}\xspace} % en-dash
\DeclareRobustCommand{\PP}{\ensuremath{\mathbb{P}}\xspace}

\newcommand{\ZXeq}{\overset{\text{ZX}}{\equiv}} % diagrammatic equality (수학 모드에서 사용)
\newcommand{\FHilb}{\ensuremath{\mathbf{FHilb}}}

\newcommand{\Zsp}[3]{Z_{#1}\!\left(#2,#3\right)}
\newcommand{\Xsp}[3]{X_{#1}\!\left(#2,#3\right)}
\newcommand{\Zspmn}[5]{Z_{#1}^{#2,#3}\!\left(#4,#5\right)}
\newcommand{\Xspmn}[5]{X_{#1}^{#2,#3}\!\left(#4,#5\right)}

\allowdisplaybreaks

\usepackage{stmaryrd}
\usepackage{tikz}
\usetikzlibrary{arrows.meta,calc,positioning,decorations.pathmorphing}

\DeclareMathOperator{\lcm}{lcm} % define \lcm
\allowdisplaybreaks

\title{Weighted Projective–Line ZX-Calculus: Quantized Orbifold Geometry for Quantum Compilation}

\author[1]{Gunhee Cho\thanks{Email: \texttt{wvx17@txstate.edu}}}
\author[2]{Jason Cheng\thanks{Email: \texttt{jc224959@my.utexas.edu}}}
\author[3]{Evelyn Li\thanks{Email: \texttt{ezl2105@columbia.edu}}}

\affil[1]{Texas State University, San Marcos, TX, USA}
\affil[2]{University of Texas at Austin, Austin, TX, USA}
\affil[3]{Columbia University, New York, NY, USA}

\date{}

\begin{document}
	\maketitle
	
	% --- footnote style를 잠시 심볼로 바꿈 ---
	\renewcommand{\thefootnote}{\fnsymbol{footnote}}
	
	% 2번 각주(†)에 equal-contribution 설명 추가
	\footnotetext[2]{Second and third authors contributed equally as co-second authors.}
	
	% 다시 일반 숫자 각주로 복원
	\renewcommand{\thefootnote}{\arabic{footnote}}
	
	\begin{abstract}
		We develop a unified geometric framework for quantum circuit compilation 
		based on \emph{quantized orbifold phases} and their diagrammatic semantics.
		Physical qubit platforms impose heterogeneous phase resolutions, 
		anisotropic Bloch-ball contractions, and hardware-dependent $2\pi$ winding behavior.  
		We show that these effects admit a natural description on the weighted projective line 
		$\mathbb{P}(a,b)$, whose orbifold points encode discrete phase grids 
		and whose monodromy captures winding accumulation under realistic 
		noise channels.
		
		Building on this geometry, we introduce the \emph{WPL--ZX calculus}, 
		an extension of the standard ZX formalism in which each spider carries 
		a weight--phase--winding triple $(a,\alpha,k)$. 
		We prove soundness of LCM-based fusion and normalization rules, 
		derive curvature predictors for phase-grid compatibility, and 
		present the \emph{Weighted ZX Circuit Compression} (WZCC) algorithm, 
		which performs geometry-aware optimization on heterogeneous phase lattices.
		
		To connect circuit-level structure with fault-tolerant architectures, 
		we introduce \emph{Monodromy-Aware Surface-Code Decoding} (MASD), 
		a winding-regularized modification of minimum-weight matching on syndrome graphs.
		MASD incorporates orbifold-weighted edge costs, producing monotone 
		decoder-risk metrics and improved robustness across phase-quantized noise models.
		
		All results are validated through symbolic and numerical simulations, 
		demonstrating that quantized orbifold geometry provides a coherent 
		and hardware-relevant extension of diagrammatic quantum compilation.
	\end{abstract}
	
	\tableofcontents
	
% ======================================================
\section{Introduction}
\label{sec:intro}
% ======================================================

Quantum circuit optimisation is increasingly constrained by the realities of 
hardware-level phase quantisation, anisotropic control landscapes, and discrete
native gate sets found in contemporary superconducting and trapped-ion devices.
Although diagrammatic methods such as the ZX-calculus 
\cite{CoeckeKissinger2017PicturingQuantum, Backens2014ZX, vandeWetering2020WorkingZX}
have become powerful tools for circuit rewriting and T-count reduction
\cite{Amy2014TCount, Heyfron2018CompilerZX, DuncanPerdrix2020GraphSimplification},
these rewriting rules implicitly assume that phases live in a continuous space
(e.g.\ $[0,2\pi)$ or $\mathbb{R}/2\pi\mathbb{Z}$).  
This assumption is increasingly invalid for NISQ-era hardware where control electronics,
quasiprobability estimators \cite{Pashayan2022NoiseQuantization}, and DRAG-style pulse shaping 
\cite{Motzoi2009DRAG, Krantz2019SuperconductingReview}
often restrict executable phases to \emph{quantised grids}, and where different qubits
may obey \emph{heterogeneous} grids due to local calibration and crosstalk constraints.
These restrictions induce an effective geometry on phase space,
which is not captured by the standard ZX-calculus.

Weighted projective lines (WPLs) \cite{Dolgachev1982Weighted} 
provide a natural differential–geometric model of such quantised phase spaces:
their orbifold points encode isotropy orders $(a,b)$ corresponding to 
allowed phase denominators, and their monodromy structure describes
how phase addition behaves across heterogeneous grids.
Moreover, WPLs carry canonical metrics with curvature $R=2/b^{2}$,
providing quantitative measures of grid anisotropy.  
These geometric features suggest that circuit rewriting should respect not only 
algebraic identities but also the geometric constraints imposed by hardware.

Motivated by these observations, this work develops a complete framework for
\emph{WPL-aware ZX-calculus}, \emph{quantisation-aware normalisation}, and
\emph{winding-aware surface-code decoding}.  
This unifies diagrammatic reasoning, orbifold geometry, and error decoding into a
single geometric pipeline.

% --------------------------------------------------
\subsection*{Motivation: hardware quantisation and phase winding}
% --------------------------------------------------

In practical NISQ devices, executable phase rotations are frequently restricted to 
discrete sets such as $\frac{2\pi}{a}\mathbb{Z}$, and different qubits or layers may
exhibit different $a$ due to heterogeneous control parameters.
This heterogeneity is further exacerbated by monodromy effects encountered during 
hardware transpilation, routing, and calibration cycles.

Furthermore, when circuits interact with surface codes or lattice-based error
correcting architectures \cite{Raussendorf2007FaultTolerantZX},
phase information is sometimes encoded as \emph{winding numbers}
on the defect graph, producing additional geometric constraints that standard
minimum-weight matching decoders fail to incorporate.

The key motivation of this paper is therefore twofold:
(i)~extend the ZX-calculus so that phases respect orbifold monodromy and heterogeneous grids,
and (ii)~build decoding and optimisation algorithms based on WPL geometry.

% --------------------------------------------------
\subsection*{Limitations of standard ZX-calculus}
% --------------------------------------------------

Although the ZX-calculus is complete for stabiliser quantum mechanics
\cite{Backens2014ZX, JeandelPerdrixVilmart2017CliffordT} and highly successful
for diagrammatic optimisation, it faces fundamental limitations when confronted with
hardware-induced phase quantisation:

\begin{itemize}
	\item \textbf{Continuous semantics:}  
	Standard ZX rules assume phase addition in $\mathbb{R}/2\pi\mathbb{Z}$,
	ignoring discrete or heterogeneous phase grids.
	
	\item \textbf{No orbifold structure:}  
	No rule accounts for isotropy orders $(a,b)$ or for monodromy behaviour 
	on weighted projective spaces.
	
	\item \textbf{Geometry-agnostic optimisation:}  
	Circuit simplification is independent of curvature, grid anisotropy,
	or hardware channel geometry.
	
	\item \textbf{Incompatibility with winding-aware decoding:}  
	When circuits interact with topological codes, phase winding information
	is lost under standard diagrammatic simplifications.
\end{itemize}

These limitations motivate a geometric extension of ZX-calculus that incorporates 
both algebraic and geometric constraints.

% --------------------------------------------------
\subsection*{Contributions of this work}
% --------------------------------------------------

This paper introduces three principal contributions:

\begin{enumerate}
	\item \textbf{The WPL–ZX calculus (Section~\ref{sec:wplzx}).}  
	We construct a diagrammatic calculus in which each spider is labelled by 
	$(a,\alpha,k)$ representing its isotropy order, phase angle, and winding index.
	Fusion, normalisation, and functorial semantics are redefined using
	weighted projective line geometry and orbifold monodromy 
	\cite{Satake1956Orbifold, Dolgachev1982Weighted, Nakahara2003GeometryTopology}.
	
	\item \textbf{Quantisation-aware circuit normalisation (WZCC) (Section~\ref{sec:wzcc}).}  
	We introduce a geometric normalisation algorithm that minimises curvature gradients
	and aligns phases on the least-common-multiple refinement grid, producing  
	hardware-compliant circuits that preserve quantisation structure.
	
	\item \textbf{Winding-aware decoding (MASD) (Section~\ref{sec:masd}).}  
	We design a minimum-weight decoder with geometrically regularised edge costs
	that incorporate defect winding numbers, curvature weights, and anisotropy penalties.
	We define decoder-risk metrics that quantify the cost of violating geometric constraints.
\end{enumerate}

Together, these contributions establish a unified framework for hardware-aware 
diagrammatic rewriting and decoding.

\subsection*{Structure of the paper}
% --------------------------------------------------

Section~\ref{sec:background} introduces background material on quantized
phases and the ZX-calculus. Section~\ref{sec:geometry} reviews weighted projective lines, orbifold
monodromy, heterogeneous phase addition, and the WPL metric. Section~\ref{sec:wplzx} develops the \WPLZX\ calculus, including
$(a,\alpha,k)$-labelled spiders, fusion rules, and categorical semantics. Section~\ref{sec:wzcc} presents the WZCC algorithm, its normalization invariants, and LCM-based compatibility procedures. Section~\ref{sec:experiments} describes the experimental setup: datasets, evaluation metrics, scaling analysis, and robustness studies. Section~\ref{sec:results} reports empirical results on compression, quantization compliance, and noise robustness. Section~\ref{sec:masd} introduces the MASD decoder, describing winding-aware edge construction, geometric regularisation, and decoder-risk metrics. Finally, Section~\ref{sec:outlook} discusses limitations, geometric and categorical implications, and future directions.
	
	% ======================================================
\section{Background: Quantized Phases and ZX-calculus}
\label{sec:background}
	% ======================================================
	
\subsection*{Phase representation and quantization grids}

In quantum hardware, phase parameters are not truly continuous. 
They are generated and controlled by digital-to-analog converters (DACs),
whose finite resolution induces a discrete set of realizable phases.
This discreteness matters when we use phase-sensitive rewrite systems 
such as the standard \ZX\ calculus, which implicitly assumes an ideal continuous group $U(1)$.
For a comprehensive treatment of phase semantics in ZX-calculus, see
Backens~\cite{Backens2014ZX}, 
Jeandel--Perdrix--Vilmart~\cite{JeandelPerdrixVilmart2017CliffordT}, 
and van~de~Wetering~\cite{vandeWetering2020WorkingZX}.
Extensions of ZX-calculus that incorporate discretized or hardware-motivated
phase constraints have also been explored in recent work on 
finite-dimensional and modular phase models, which motivate the weighted
and quantized geometry developed in this paper.

\begin{definition}[Quantized phase grid]
	Let $a \in \mathbb{N}$ be a fixed integer corresponding to the number of discrete
	phase levels available in a control channel.
	We define the \emph{quantized phase grid} of weight $a$ by
	\[
	\Phi_a = \bigl\{\, \tfrac{2\pi n}{a} \mid n \in \mathbb{Z} \bigr\}
	\subset [0,2\pi).
	\]
	On this set we use addition modulo $2\pi$:
	for $\alpha_1,\alpha_2 \in \Phi_a$,
	\[
	\alpha_1 \oplus_a \alpha_2 
	= (\alpha_1 + \alpha_2) \bmod 2\pi.
	\]
\end{definition}

\begin{remark}
	In the idealized continuous model, $a \to \infty$ and $\Phi_a$ approaches the circle $S^1$.
	In practice, however, $a$ is finite (for example, $a = 256$ for an 8-bit DAC),
	and $\Phi_a$ carries the algebraic structure
	\[
	\Phi_a \simeq \tfrac{2\pi}{a}\,\mathbb{Z}_a,
	\]
	a finite cyclic group of realizable phases.
	This discretization directly reflects the quantized phase grids imposed by 
	digital-to-analog control in contemporary quantum hardware, and it provides
	the algebraic motivation for the weighted and orbifold interpretation adopted in this work.
\end{remark}

\begin{example}[Phase resolution]
	Suppose a control stack supports a minimum step of $1.40625^\circ = 2\pi/256$.
	Then $a = 256$, and a gate $R_Z(\alpha)$ can only realize rotations
	of the form $\alpha = 2\pi n / 256$ with $n \in \mathbb{Z}$.
	Thus, a \ZX\ spider labeled by $\alpha$ effectively lives on the discrete grid $\Phi_{256}$.
\end{example}

\begin{definition}[Mixed-weight fusion grid]
	Let $a,b \in \mathbb{N}$ be two phase weights associated with different control channels
	or different parts of a circuit.
	When phases $\alpha \in \Phi_a$ and $\beta \in \Phi_b$ are combined (for example, by spider fusion),
	the resulting phase is naturally taken on the \emph{refined grid}
	whose step size is
	\[
	\Delta\phi = \frac{2\pi}{\mathrm{lcm}(a,b)}.
	\]
	In other words, mixed-weight compositions lie on the grid $\Phi_{\mathrm{lcm}(a,b)}$.
\end{definition}

\begin{remark}
	This ``LCM grid'' captures the fact that different qubit channels may have
	different phase resolutions.
	When these channels interact, the effective resolution is the least common
	multiple of individual resolutions.
	This arithmetic behavior will later appear in the fusion rules of
	our weighted \ZX-calculus, where spider weights combine via $\mathrm{lcm}$.
\end{remark}

\begin{lemma}[Discrete phase closure under fusion]
	Let $a,b \in \mathbb{N}$.
	For any $\alpha \in \Phi_a$ and $\beta \in \Phi_b$,
	there exists $\gamma \in \Phi_{\mathrm{lcm}(a,b)}$ such that
	\[
	\gamma \equiv \alpha + \beta \pmod{2\pi}.
	\]
\end{lemma}

\begin{proof}
	Write $\alpha = 2\pi m/a$ and $\beta = 2\pi n/b$ for some $m,n \in \mathbb{Z}$.
	Then
	\[
	\alpha + \beta
	= 2\pi\!\left(\frac{m}{a} + \frac{n}{b}\right)
	= 2\pi\,\frac{mb + na}{ab}.
	\]
	Since $\mathrm{lcm}(a,b) = ab / \gcd(a,b)$, the factor $\tfrac{mb+na}{ab}$ is an integer multiple
	of $1/\mathrm{lcm}(a,b)$.
	Thus $\alpha+\beta$ is an integer multiple of $2\pi/\mathrm{lcm}(a,b)$, i.e.\ $\alpha + \beta \in \Phi_{\mathrm{lcm}(a,b)}$.
\end{proof}

\begin{remark}
	The lemma ensures that phase fusion is closed within the class of quantized grids:
	any two quantized phases can be combined into a phase on a (possibly finer) grid.
	This property is the algebraic basis of the LCM-based fusion rule for weighted spiders,
	which generalizes the phase-sum rules discussed in 
	Backens~\cite{Backens2014ZX} and Jeandel--Perdrix--Vilmart~\cite{JeandelPerdrixVilmart2017CliffordT}.
\end{remark}

\begin{definition}[Accumulated phase and winding number]
	Let $(\alpha_t)_{t\ge 0}$ be a sequence of quantized phases,
	and let $\Delta\alpha_t = \alpha_{t+1} - \alpha_t$ denote the increment at step $t$.
	The \emph{total accumulated phase} after $T$ steps is
	\[
	\Theta_T = \sum_{t=0}^{T-1} \Delta\alpha_t.
	\]
	The corresponding \emph{winding number} $k_T \in \mathbb{Z}$ is defined by
	\[
	k_T = \left\lfloor \frac{\Theta_T}{2\pi} \right\rfloor.
	\]
\end{definition}

\begin{remark}
	The pair $(\Theta_T \bmod 2\pi,\, k_T)$ separates local phase information
	(from $\Theta_T$ modulo $2\pi$) and global topological information
	(from the integer number of full turns $k_T$).
	In the weighted \ZX-calculus developed later, each spider will therefore
	carry a triplet $(a,\alpha,k)$ encoding grid order, fractional phase, and winding count.
\end{remark}

\subsection*{ZX-calculus essentials}\label{subsec:zx-essentials}

The \ZX-calculus provides a graphical language for reasoning about
quantum circuits and linear maps between finite-dimensional Hilbert spaces.
It is based on the interplay between two complementary bases of a qubit:
the computational $Z$-basis $\{\ket{0},\ket{1}\}$ and the Hadamard-rotated
$X$-basis $\{\ket{+},\ket{-}\}$.
Each basis gives rise to a family of diagrammatic generators, called
\emph{spiders}, connected by edges representing qubit wires.
For a systematic introduction see
Backens~\cite{Backens2014ZX} and
van~de~Wetering~\cite{vandeWetering2020WorkingZX},
while the completeness of the calculus for Clifford+T circuits
was established by Jeandel, Perdrix, and Vilmart~\cite{JeandelPerdrixVilmart2017CliffordT}.
Circuit-level optimization and simplification rules are discussed by
Duncan and Perdrix~\cite{DuncanPerdrix2020GraphSimplification} and 
Fischbach~\cite{Fischbach2025ZXOpt}.

\begin{definition}[Z- and X-spiders]
	For integers $m,n \ge 0$ and a phase $\alpha \in [0,2\pi)$,
	the \emph{Z-spider} and \emph{X-spider} are defined by the linear maps
	\[
	\Zsp{m}{n}{\alpha} =
	\ket{0}^{\otimes n}\!\bra{0}^{\otimes m} +
	e^{i\alpha}\ket{1}^{\otimes n}\!\bra{1}^{\otimes m},
	\qquad
	\Xsp{m}{n}{\alpha} =
	\ket{+}^{\otimes n}\!\bra{+}^{\otimes m} +
	e^{i\alpha}\ket{-}^{\otimes n}\!\bra{-}^{\otimes m}.
	\]
	These spiders act as multi-input, multi-output linear maps
	$\mathbb{C}^{2^{m}} \!\to\! \mathbb{C}^{2^{n}}$.
\end{definition}

\begin{remark}
	A \ZX-diagram is built by composing and tensoring spiders according to
	graphical connectivity.
	An edge joining two spiders represents a contraction over a shared qubit.
	Diagrammatic equality, written $D_1 \ZXeq D_2$, denotes that
	the two diagrams represent the same linear map up to a global scalar.
\end{remark}

\begin{lemma}[Spider fusion law]
	Two spiders of the same color connected by one or more wires
	can be merged into a single spider whose phase is the sum of the original phases:
	\[
	\Zsp{m}{k}{\alpha_1}\; \Zsp{k}{n}{\alpha_2}
	\;\ZXeq\;
	\Zsp{m}{n}{\alpha_1 + \alpha_2},
	\qquad
	\Xsp{m}{k}{\beta_1}\; \Xsp{k}{n}{\beta_2}
	\;\ZXeq\;
	\Xsp{m}{n}{\beta_1 + \beta_2}.
	\]
\end{lemma}

\begin{proof}[Sketch]
	In the computational basis, composition of Z-spiders adds the phases along
	the $\ket{1}$ branch; intermediate connections act as identities on each branch.
	The same argument holds for X-spiders in the Hadamard basis.
\end{proof}

\begin{definition}[Color change via Hadamard]
	The Hadamard gate $H$ acts as an isomorphism between Z- and X-structures:
	\[
	H \; \Zsp{m}{n}{\alpha} \; H
	\;\ZXeq\;
	\Xsp{m}{n}{\alpha}.
	\]
\end{definition}

\begin{remark}
	This duality encodes the mutual unbiasedness of the two bases and lets
	operations be represented by alternating red/green spiders (X and Z).
\end{remark}

\begin{example}[From circuit to ZX-diagram: single-qubit rotation chain]
	Consider
	\[
	U = R_Z(\alpha)\,R_X(\beta)\,R_Z(\gamma),
	\]
	the standard $Z\!X\!Z$ Euler decomposition of a single-qubit unitary.
	In the \ZX-calculus this becomes
	\[
	\Zsp{1}{1}{\alpha} \; \Xsp{1}{1}{\beta} \; \Zsp{1}{1}{\gamma}.
	\]
	Using same-color fusion and the $H$ color-change rule, we obtain the
	simplified normal form
	\[
	\Zsp{1}{1}{\alpha+\gamma}\;\Xsp{1}{1}{\beta}.
	\]
	Figures~\ref{fig:circuit-zx-euler} and \ref{fig:zx-diagram-euler} visualize
	this simplification at the circuit level and in the spider language.
\end{example}

%----------------- FIGURE 1: CIRCUIT VIEW -----------------
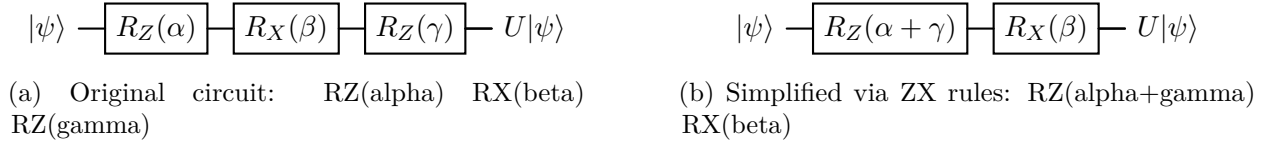
\begin{figure}[t]
	\centering
	\begin{subfigure}[b]{0.46\textwidth}
		\centering
		\begin{quantikz}[transparent, row sep=0.2cm, column sep=0.35cm]
			\lstick{$|\psi\rangle$} & \gate{R_Z(\alpha)} & \gate{R_X(\beta)} & \gate{R_Z(\gamma)} & \rstick{$U|\psi\rangle$} \qw
		\end{quantikz}
		\caption{Original circuit: RZ(alpha) RX(beta) RZ(gamma)}
	\end{subfigure}\hfill
	\begin{subfigure}[b]{0.46\textwidth}
		\centering
		\begin{quantikz}[transparent, row sep=0.2cm, column sep=0.35cm]
			\lstick{$|\psi\rangle$} & \gate{R_Z(\alpha+\gamma)} & \gate{R_X(\beta)} & \rstick{$U|\psi\rangle$} \qw
		\end{quantikz}
		\caption{Simplified via ZX rules: RZ(alpha+gamma) RX(beta)}
	\end{subfigure}
	\caption{Single-qubit example of ZX-based simplification (circuit view)}
	\label{fig:circuit-zx-euler}
\end{figure}

%----------------- FIGURE 2: SPIDER VIEW -----------------
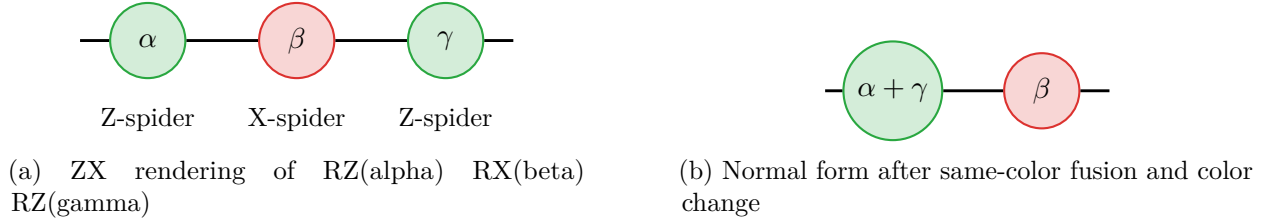
\begin{figure}[t]
	\centering
	\begin{subfigure}[b]{0.46\textwidth}
		\centering
		\begin{tikzpicture}[scale=0.90]
			\node[spiderZ] (Z1) at (0,0) {$\alpha$};
			\node[spiderX] (X1) at (2.2,0) {$\beta$};
			\node[spiderZ] (Z2) at (4.4,0) {$\gamma$};
			\draw[wire] (-1,0) -- (Z1) -- (X1) -- (Z2) -- (5.4,0);
			\node[below] at (0,-0.85) {\small Z-spider};
			\node[below] at (2.2,-0.85) {\small X-spider};
			\node[below] at (4.4,-0.85) {\small Z-spider};
		\end{tikzpicture}
		\caption{ZX rendering of RZ(alpha) RX(beta) RZ(gamma)}
	\end{subfigure}\hfill
	\begin{subfigure}[b]{0.46\textwidth}
		\centering
		\begin{tikzpicture}[scale=0.90]
			\node[spiderZ] (Zs) at (0,0) {$\alpha+\gamma$};
			\node[spiderX] (Xs) at (2.2,0) {$\beta$};
			\draw[wire] (-1,0) -- (Zs) -- (Xs) -- (3.2,0);
		\end{tikzpicture}
		\caption{Normal form after same-color fusion and color change}
	\end{subfigure}
	\caption{ZX diagrams: before vs. after simplification (spider view)}
	\label{fig:zx-diagram-euler}
\end{figure}

\begin{definition}[Bialgebra and Hopf interactions]
	The fundamental interaction between Z- and X-spiders in the \ZX-calculus
	is captured by two algebraic principles: the \emph{bialgebra law} and the
	\emph{Hopf law}. For nonnegative integers $m,n,k,\ell$:
	\[
	\Xsp{m}{n}{0}\;\Zsp{n}{k}{0} \;\ZXeq\; \Zsp{m}{n}{0}\;\Xsp{n}{k}{0},
	\qquad
	\Zsp{1}{2}{0}\;\Xsp{2}{1}{0} \;\ZXeq\; \Xsp{1}{2}{0}\;\Zsp{2}{1}{0},
	\]
	and
	\[
	\Zsp{1}{1}{0}\;\Xsp{1}{1}{0}\;\Zsp{1}{1}{0} \;\ZXeq\; \mathrm{id},
	\qquad
	\Xsp{1}{1}{0}\;\Zsp{1}{1}{0}\;\Xsp{1}{1}{0} \;\ZXeq\; \mathrm{id}.
	\]
\end{definition}

\begin{example}[Properties of Spiders]
	
	On the bloch sphere $\mathbb{CP}^1$,  Z-spiders correspond to rotations about the Z-axis and X-spiders correspond to rotations about the X-axis. A phase $\alpha$ in a spider corresponds to a rotation by angle $\alpha$
	
	\begin{itemize}
		\item \textbf{Spider Fusion:} Two spiders of the same color fuse into one, summing their phases.
		\item \textbf{Identity:} A spider with no phase and one input/output wire acts as the identity map.
		\item \textbf{Bialgebra Law:} Connecting green and red spiders satisfies the algebra–coalgebra compatibility (bialgebra) relation.
		\item \textbf{Hadamard Gates:} Hadamard operations interchange red (X) and green (Z) spiders.
	\end{itemize}
	
	\begin{figure}[h!]
		\centering
		\includegraphics[width=0.88\textwidth]{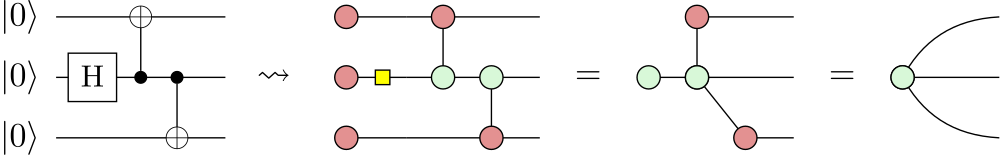}
		\caption{Illustration of spider properties: fusion, identity, bialgebra, and Hadamard color change.}
		\label{fig:spider-properties}
	\end{figure}
\end{example}

\begin{remark}
	The expressive power of the \ZX-calculus is topological and compositional:
	matrix identities often become simple rewiring, fusion, or color-change steps.
	Practically, this enables (i) automatic circuit simplification and phase-count
	reduction~\cite{DuncanPerdrix2020GraphSimplification}, (ii) compiler-level
	optimization for Clifford+T circuits, and (iii) hardware-agnostic reasoning
	about entanglement and measurement flow. In the next section we add
	\emph{phase quantization} and \emph{winding labels} to track monodromy on
	discrete phase grids, yielding the \emph{WPL--ZX} formalism.
\end{remark}

\subsection*{Discrete and modular extensions of ZX}\label{subsec:discrete-zx}

While the standard \ZX-calculus assumes continuous phases in $U(1)$,
real quantum hardware operates on a discrete grid of realizable phases.
To model such practical constraints, we introduce a
\emph{discrete and modular} extension of the calculus,
where phase values belong to finite cyclic groups rather than to the continuous circle.
This construction unites categorical quantum mechanics~\cite{Selinger2008DaggerCompact,CoeckeKissinger2017PicturingQuantum}
with finite-dimensional quantum theory~\cite{Gross2006Hudson},
and aligns with modern compiler-level optimization frameworks~\cite{Amy2014TCount,Heyfron2018CompilerZX,Fischbach2025ZXOpt}.

\begin{definition}[Modular phase group]
	Let $a \in \mathbb{N}$ denote the grid order of a control channel.
	The \emph{modular phase group} associated with this resolution is
	\[
	G_a := \mathbb{R} / \tfrac{2\pi}{a}\mathbb{Z}
	\;\simeq\;
	\tfrac{2\pi}{a}\mathbb{Z}_a.
	\]
	Elements of $G_a \cong \{\,0, \tfrac{2\pi}{a}, \tfrac{4\pi}{a}, \dots, \tfrac{2\pi(a-1)}{a}\,\} \subset [0,2\pi)$
	correspond to quantized phase values
	\(\alpha = \tfrac{2\pi n}{a}\) for $n \in \mathbb{Z}_a$,
	and the group operation is addition modulo $2\pi$.
\end{definition}

\begin{remark}
	As $a \to \infty$, the modular group $G_a$ converges to the continuous circle group $U(1)$.
	For finite $a$, it forms a discrete approximation of $U(1)$,
	preserving its group structure but restricting phase precision.
	This transition parallels the finite-to-continuum viewpoint discussed in~\cite{Gross2006Hudson},
	and fits naturally into the categorical semantics of quantum processes~\cite{Selinger2008DaggerCompact,CoeckeKissinger2017PicturingQuantum}.
\end{remark}

\begin{definition}[Discrete spiders]
	For each $a \in \mathbb{N}$, the \emph{discrete Z- and X-spiders} of grid order $a$ restrict their phases to $G_a$:
	\[
	\Zspmn{m}{n}{a}{\alpha} :=
	\ket{0}^{\otimes n}\!\bra{0}^{\otimes m}
	+ e^{i\alpha}\ket{1}^{\otimes n}\!\bra{1}^{\otimes m},
	\quad \alpha \in G_a.
	\]
	The X-spider is defined analogously in the Hadamard-rotated basis,
	as introduced in the standard ZX framework~\cite{Backens2014ZX,JeandelPerdrixVilmart2017CliffordT,vandeWetering2020WorkingZX}.
\end{definition}

\begin{lemma}[Modular fusion law]
	Let $a,b \in \mathbb{N}$ and let $\alpha \in G_a$, $\beta \in G_b$.
	Then the fusion of discrete Z-spiders satisfies
	\[
	\Zspmn{m}{k}{a}{\alpha}
	\;\Zspmn{k}{n}{b}{\beta}
	\;\ZXeq\;
	\Zspmn{m}{n}{\lcm(a,b)}{\alpha \oplus \beta},
	\]
	where \(\alpha \oplus \beta = (\alpha + \beta) \bmod 2\pi\)
	and the resulting phase lies in $G_{\lcm(a,b)}$.
\end{lemma}

\begin{proof}
	Writing $\alpha = 2\pi m/a$ and $\beta = 2\pi n/b$ gives
	$\alpha + \beta = 2\pi (mb + na)/(ab)$.
	Since $\lcm(a,b) = ab/\gcd(a,b)$,
	this composite phase corresponds to an element of $G_{\lcm(a,b)}$
	under modular addition.
	This generalizes the stabilizer-level fusion rules of~\cite{Backens2014ZX}
	to modular lattices.
\end{proof}

\begin{remark}
	The modular fusion rule extends the ordinary spider fusion law
	to hardware-realistic phase grids.
	Fusing spiders of different grid orders automatically refines them to
	the least-common-multiple lattice $\lcm(a,b)$.
	This refinement preserves the compositional semantics
	of the ZX framework~\cite{CoeckeKissinger2017PicturingQuantum,DuncanPerdrix2020GraphSimplification}.
\end{remark}

\begin{definition}[Phase lattice and directed refinement]
	The family of all phase grids
	\[
	\mathcal{G} := \{\, G_a : a \in \mathbb{N} \,\}
	\]
	forms a directed system under refinement:
	if $a \mid b$, there is a canonical embedding
	\(\iota_{a\to b} : G_a \hookrightarrow G_b\),
	sending \(\tfrac{2\pi n}{a} \mapsto \tfrac{2\pi n}{b}\).
	Its inductive limit recovers the continuous group $U(1)$,
	consistent with categorical treatments of quantum phase spaces~\cite{CoeckeKissinger2017PicturingQuantum}.
\end{definition}

\begin{remark}
	This directed system $\mathcal{G}$ captures hierarchical phase resolution in hardware.
	Each inclusion $\iota_{a\to b}$ represents an upsampling of phase precision.
	Such structure aligns with compiler-level phase synthesis and
	resource optimization techniques~\cite{Amy2014TCount,Heyfron2018CompilerZX,Fischbach2025ZXOpt}.
\end{remark}

\begin{example}[Discrete phase composition]
	Consider $a=4$ and $b=6$.
	Then $\lcm(a,b)=12$, and
	\[
	\Phi_4 = \{0, \tfrac{\pi}{2}, \pi, \tfrac{3\pi}{2}\}, \qquad
	\Phi_6 = \{0, \tfrac{\pi}{3}, \tfrac{2\pi}{3}, \pi, \tfrac{4\pi}{3}, \tfrac{5\pi}{3}\}.
	\]
	Fusing a $Z$-spider of phase $\pi/2$ (from $\Phi_4$)
	with another of phase $\pi/3$ (from $\Phi_6$)
	yields a spider with phase $\pi/2 + \pi/3 = 5\pi/6$,
	which lies on the finer grid $\Phi_{12}$.
	This corresponds to least-common-multiple refinement,
	as seen in ZX-based circuit simplification~\cite{DuncanPerdrix2020GraphSimplification,vandeWetering2020WorkingZX}.
\end{example}

\begin{definition}[Modular reduction and phase aliasing]
	When a continuous-phase circuit is compiled to a hardware target with grid order $a$,
	each phase $\theta \in [0,2\pi)$ is reduced to its nearest representative in $G_a$:
	\[
	\theta \mapsto \frac{2\pi}{a}\,\mathrm{round}\!\left(\frac{a\theta}{2\pi}\right).
	\]
	This introduces a quantization error and a periodic \emph{aliasing} effect
	in the corresponding ZX-diagram~\cite{Pashayan2022NoiseQuantization}.
\end{definition}

\begin{remark}
	The aliasing phenomenon implies that discrete \ZX-diagrams
	are no longer faithful embeddings of continuous unitaries,
	but equivalence classes under phase rounding.
	This modular reduction plays a key role in the analysis of
	circuit compression, phase accumulation, and
	fault-tolerant resource estimation~\cite{Raussendorf2007FaultTolerantZX,nLabZX2024}.
\end{remark}

\subsection*{Motivating examples from hardware control}\label{subsec:hardware-motivation}

The motivation for introducing a discrete and modular formulation of the \ZX-calculus
comes directly from the physical architecture of quantum control hardware.
Each qubit is driven by analog waveforms synthesized from digital control signals
with finite bit depth and bounded dynamic range.
Consequently, all physically realizable phases are inherently quantized,
and the diagrammatic calculus must respect this discretization~\cite{Amy2014TCount,Heyfron2018CompilerZX,Fischbach2025ZXOpt}.

\begin{example}[Digital-to-analog conversion (DAC) quantization]
	Let a control channel employ an $n$-bit digital-to-analog converter (DAC)
	to synthesize a phase $\phi$.
	The achievable phase set is
	\[
	\Phi_{2^n} = \Bigl\{ \frac{2\pi k}{2^n} : k \in \mathbb{Z}_{2^n} \Bigr\}.
	\]
	For $n=8$, the smallest resolvable increment is
	$\Delta\phi = 2\pi/256 \approx 1.4^\circ$.
	Every rotation gate $R_Z(\phi)$ or phase-labeled spider
	is effectively projected onto this discrete lattice.
	Thus, the control hardware realizes the modular group $G_{2^n}$
	introduced in the previous subsection~\cite{Gross2006Hudson}.
\end{example}

\begin{remark}
	Such quantization is not a software artifact but a fundamental hardware constraint.
	Because waveform synthesis and phase accumulation both operate in discrete digital domains,
	the modular structure is intrinsic to gate calibration and cannot be removed
	by higher-level compilation alone~\cite{Amy2014TCount,Heyfron2018CompilerZX}.
\end{remark}

\begin{example}[Multi-channel phase synchronization]
	Consider two qubit control channels with grid orders $a=256$ and $b=192$,
	driven by separate local oscillators.
	When performing a two-qubit entangling operation,
	the effective phase reference depends on both clocks.
	The joint synchronization grid therefore has step size
	\[
	\Delta\phi_{\mathrm{sync}} = \frac{2\pi}{\mathrm{lcm}(a,b)} 
	= \frac{2\pi}{768}.
	\]
	Even though each individual channel has coarse resolution,
	their joint operation resides on a finer LCM grid.
	This corresponds exactly to the modular fusion rule
	of the discrete \ZX-calculus~\cite{Backens2014ZX,JeandelPerdrixVilmart2017CliffordT,DuncanPerdrix2020GraphSimplification}.
\end{example}

\begin{definition}[Coupled control lattice]
	Let $\{a_i\}_{i=1}^N$ denote the phase grid orders of $N$ control channels.
	The \emph{coupled control lattice} governing multi-qubit operations is defined as
	\[
	\Lambda := \Phi_{\mathrm{lcm}(a_1,a_2,\dots,a_N)}.
	\]
	Any composite operation involving these channels must draw its phase label from $\Lambda$.
	This lattice determines the minimal phase increment that remains coherent
	across all interacting channels.
\end{definition}

\begin{remark}
	The lattice $\Lambda$ quantifies hardware-level phase coherence
	at the granularity of control discretization.
	If one channel has a significantly lower resolution,
	it becomes the bottleneck of phase precision—analogous to
	a clock synchronization limit in distributed digital systems~\cite{CoeckeKissinger2017PicturingQuantum}.
\end{remark}

\begin{example}[Phase drift and modular winding]
	In superconducting or trapped-ion qubit systems,
	slow drift in local oscillator phases leads to gradual accumulation
	of modular offsets between channels.
	Let $\phi_i(t)$ denote the instantaneous phase of channel $i$,
	each defined modulo $2\pi/a_i$.
	Over long times, the relative phase difference
	\[
	\Delta_{ij}(t) = \phi_i(t) - \phi_j(t)
	\]
	can undergo discrete modular windings in $G_{\lcm(a_i,a_j)}$.
	In the diagrammatic setting, such windings manifest as changes in
	the topological phase label of a spider,
	while its local phase remains invariant modulo $2\pi$~\cite{vandeWetering2020WorkingZX,nLabZX2024}.
\end{example}

\begin{definition}[Hardware-induced phase monodromy]
	A \emph{phase monodromy event} occurs whenever accumulated drift or
	control noise produces a full modular winding in a grid $G_a$.
	Formally, a monodromy index $k \in \mathbb{Z}$ records the number of times
	the physical phase $\phi(t)$ crosses $2\pi$ while the logical circuit
	representation remains modulo $2\pi/a$.
\end{definition}

\begin{remark}
	Monodromy indices distinguish physically distinct phase trajectories
	that appear identical modulo $2\pi$.
	In the categorical semantics of the \ZX-calculus~\cite{Selinger2008DaggerCompact,CoeckeKissinger2017PicturingQuantum},
	such indices correspond to nontrivial morphisms on the same diagrammatic boundary.
	In later sections, these indices will be promoted to weights,
	yielding the weighted projective line interpretation of the \WPLZX-calculus.
\end{remark}

\begin{example}[Cross-talk compensation and phase aliasing]
	During simultaneous multi-qubit control, cross-talk and residual couplings
	in the microwave envelope may induce effective phases that are linear combinations
	of multiple channel phases.
	If channel $i$ contributes $\phi_i = 2\pi n_i / a_i$, then the composite phase is
	\[
	\phi_{\mathrm{eff}} = \sum_i \lambda_i \phi_i, \qquad \lambda_i \in \mathbb{R}.
	\]
	Because each $\phi_i$ is defined modulo $2\pi/a_i$,
	the resulting $\phi_{\mathrm{eff}}$ exhibits aliasing over
	$\Phi_{\mathrm{lcm}(a_1,\dots,a_N)}$.
	Diagrammatically, this corresponds to multi-weight spider fusion
	under modular arithmetic~\cite{DuncanPerdrix2020GraphSimplification,Backens2014ZX}.
\end{example}

\begin{remark}
	These effects—quantization, synchronization, drift, and aliasing—demonstrate
	that the correct mathematical setting for realistic quantum circuit diagrams
	is not the continuous torus $U(1)^n$, but an arithmetic lattice
	of modular phase groups connected by least common multiples.
	This viewpoint provides a physically grounded bridge between
	fault-tolerant architecture~\cite{Raussendorf2007FaultTolerantZX}
	and noise-aware optimization of diagrammatic circuits~\cite{Pashayan2022NoiseQuantization}.
\end{remark}
	
	% ======================================================
% ======================================================
\section{Weighted Projective Lines and Orbifold Geometry}
\label{sec:geometry}
% ======================================================
\label{sec:WPL-geometry}

The discrete modular phase grids discussed in Section~2 can be assembled
into a smooth---but orbifold---geometric object: the \emph{weighted
	projective line}~\cite{Satake1956Orbifold,Dolgachev1982Weighted,Nakahara2003GeometryTopology}.
This space plays two roles in our framework:
(i) it is the global phase manifold on which weighted spiders live, and
(ii) it provides the topological and metric background for
quantization-aware circuit optimization in Section~5.
Throughout this section we emphasise the parallel between
arithmetic operations on discrete grids and geometric features
of the orbifold circle underlying $\mathbb{P}(a,b)$.

% ------------------------------------------------------
\subsection*{Weighted projective line $\mathbb{P}(a,b)$ as phase space}
\label{subsec:weighted-phase-space}
% ------------------------------------------------------

We begin with the basic algebraic/topological description.

\begin{definition}[Weighted projective line]
	For positive integers $a,b \in \mathbb{N}$,
	the \emph{weighted projective line} is the quotient
	\[
	\mathbb{P}(a,b)
	:=
	\bigl( \mathbb{C}^2 \setminus \{0\} \bigr) \big/ \sim,
	\qquad
	(z_0,z_1) \sim (\lambda^a z_0,\lambda^b z_1)
	\quad\text{for all } \lambda \in \mathbb{C}^\times,
	\]
	as in~\cite{Dolgachev1982Weighted}.
	The integers $a$ and $b$ specify the relative scaling weights
	of the two homogeneous coordinates.
\end{definition}

\begin{remark}[Relation to the Bloch sphere]
	When $a=b=1$ one recovers the ordinary complex projective line
	$\mathbb{C}\mathrm{P}^{1}$, which is diffeomorphic to the Bloch sphere $S^2$
	and serves as the continuous phase space for a
	single-qubit pure state~\cite{Nakahara2003GeometryTopology}.
	Standard ZX-calculus phases live on the unit circle
	$S^1 \subset \mathbb{C}\mathrm{P}^{1}$~\cite{Backens2014ZX,vandeWetering2020WorkingZX,nLabZX2024}.
	For general $(a,b)$, the quotient introduces two orbifold points
	with local stabilizers $\mathbb{Z}_a$ and $\mathbb{Z}_b$
	at $[1\!:\!0]$ and $[0\!:\!1]$; these encode discrete
	phase degeneracies dictated by heterogeneous hardware
	grid orders~\cite{Amy2014TCount,Heyfron2018CompilerZX}.
\end{remark}

\begin{example}[Orbifold circle picture]
	Topologically, $\mathbb{P}(a,b)$ can be viewed as a sphere with two
	conical points in the sense of orbifolds~\cite{Satake1956Orbifold,Nakahara2003GeometryTopology},
	or equivalently as an ``orbifold circle''
	\[
	\mathbb{P}(a,b) \;\simeq\;
	S^1_{(a,b)} := S^1/(\mathbb{Z}_a \text{ at } 0,\; \mathbb{Z}_b \text{ at } \pi).
	\]
	A full turn around the circle corresponds to a $2\pi$ global
	phase shift, but small angular neighborhoods near the two defects
	are rescaled by factors $1/a$ and $1/b$, respectively.
	From the point of view of control hardware, these two cone points
	record the native phase periodicities of two different
	control channels~\cite{Motzoi2009DRAG,Krantz2019SuperconductingReview}.
\end{example}

\begin{definition}[Phase coordinates and affine charts]
	Let $[z_0:z_1]$ denote homogeneous coordinates on $\mathbb{P}(a,b)$.
	On the affine chart $U_0 = \{z_0 \neq 0\}$ we define
	\[
	\xi = \frac{z_1^a}{z_0^b},
	\]
	which is invariant under the weighted $\mathbb{C}^\times$-action.
	On the complementary chart $U_1=\{z_1\neq 0\}$ we take
	\[
	\eta = \frac{z_0^b}{z_1^a}.
	\]
	On the overlap $U_0 \cap U_1$ one has the transition relation
	$\eta = \xi^{-1}$, as in the usual projective line, but the
	angular parametrization of $\xi$ and $\eta$ reflects the
	weight ratio $a:b$~\cite{Dolgachev1982Weighted}.
\end{definition}

\begin{remark}[Spiders as processes on $\mathbb{P}(a,b)$]
	The local coordinates $\xi$ and $\eta$ may be thought of as
	phase frames analogous to the $Z$- and $X$-bases of 
	ZX-calculus~\cite{JeandelPerdrixVilmart2017CliffordT,Backens2014ZX,CoeckeKissinger2017PicturingQuantum}.
	In our weighted calculus a spider labelled by weights $(a,b)$
	is interpreted as a process whose phase parameter takes values
	on $\mathbb{P}(a,b)$; the two weights record the local isotropy of the
	underlying phase channels.
	This perspective will be used in Section~\ref{sec:wplzx}
	when we introduce $(a,\alpha,k)$-labelled spiders.
\end{remark}

\begin{example}[Heterogeneous control channels]
	Consider a two-qubit gate driven by two control channels
	with DAC resolutions $a=256$ and $b=192$, as discussed in
	engineering treatments of superconducting control hardware
	(e.g.~\cite{Motzoi2009DRAG,Krantz2019SuperconductingReview}).
	Each channel implements quantized phases
	\[
	G_{256} \simeq \tfrac{2\pi}{256}\,\mathbb{Z}_{256},
	\qquad
	G_{192} \simeq \tfrac{2\pi}{192}\,\mathbb{Z}_{192},
	\]
	rather than a continuum.
	The relative-phase manifold of the joint control is
	\[
	\frac{U(1)_{256}\times U(1)_{192}}{U(1)_{\mathrm{global}}}
	\simeq \mathbb{P}(256,192),
	\]
	a weighted orbifold circle whose two cone points carry
	local isotropy groups $\mathbb{Z}_{256}$ and $\mathbb{Z}_{192}$.
	The effective synchronization grid is determined
	by the least common multiple $L=\operatorname{lcm}(256,192)=768$,
	with fundamental increment $\Delta\phi = 2\pi/L$.
	This geometric picture matches the modular lattice model used
	in finite-dimensional phase space
	quantization~\cite{Gross2006Hudson}.
\end{example}

% ------------------------------------------------------
\subsection*{Orbifold points, isotropy, and monodromy}
% ------------------------------------------------------

We now make the orbifold structure of $\mathbb{P}(a,b)$ explicit
and relate it to discrete phase windings and geometric phases
in the spirit of~\cite{Berry1984Phase,Nayak2008NonAbelian}.

\begin{definition}[Orbifold points and isotropy groups]
	For
	\[
	\mathbb{P}(a,b) = \bigl( \mathbb{C}^2 \setminus \{0\} \bigr)\big/\!\sim,
	\qquad
	(z_0,z_1)\sim(\lambda^a z_0,\lambda^b z_1),
	\]
	the points
	\[
	p_0 = [1:0], \qquad p_\infty = [0:1]
	\]
	are orbifold points with isotropy groups
	\[
	\operatorname{Iso}(p_0)=\mu_a=\{e^{2\pi i k/a}:k\in\mathbb{Z}_a\},
	\qquad
	\operatorname{Iso}(p_\infty)=\mu_b=\{e^{2\pi i \ell/b}:\ell\in\mathbb{Z}_b\},
	\]
	as in the general orbifold framework~\cite{Satake1956Orbifold,Nakahara2003GeometryTopology}.
\end{definition}

\begin{remark}[Conical local structure]
	In a neighborhood of $p_0$ one may choose a uniformizing coordinate
	$u$ on $\mathbb{C}$ such that the group $\mu_a$ acts by
	$u\mapsto e^{2\pi i/a}u$; the quotient is a cone of angle $2\pi/a$.
	Likewise, a neighborhood of $p_\infty$ is a cone of angle $2\pi/b$.
	Hence $\mathbb{P}(a,b)$ is a sphere with two conical singularities,
	with deficit angles $2\pi(1-1/a)$ and $2\pi(1-1/b)$,
	consistent with the standard orbifold Euler characteristic
	computations in~\cite{Nakahara2003GeometryTopology}.
\end{remark}

\begin{example}[Orbifold Euler characteristic]
	For $a=b=1$ both cone angles are $2\pi$ and the surface is smooth
	with Euler characteristic $\chi=2$.
	For $a=2$, $b=3$ one obtains
	\[
	\chi_{\mathrm{orb}}\bigl(\mathbb{P}(2,3)\bigr)
	= 2 - \Bigl( 1-\tfrac{1}{2} \Bigr)
	- \Bigl( 1-\tfrac{1}{3} \Bigr)
	= \tfrac{5}{6},
	\]
	a non-integer value reflecting fractional holonomy at the
	conical points~\cite{Satake1956Orbifold}.
\end{example}

\begin{definition}[Orbifold fundamental group and monodromy]
	The orbifold fundamental group of $\mathbb{P}(a,b)$ is
	\[
	\pi_1^{\mathrm{orb}}\!\bigl(\mathbb{P}(a,b)\bigr)
	=
	\langle \gamma_0,\gamma_\infty
	\mid \gamma_0^a = \gamma_\infty^b = \gamma_0\gamma_\infty = 1\rangle,
	\]
	where $\gamma_0$ and $\gamma_\infty$ are loops around $p_0$ and
	$p_\infty$, respectively.
	A \emph{monodromy representation} is a homomorphism
	\[
	\rho:\pi_1^{\mathrm{orb}}(\mathbb{P}(a,b))\to U(1),
	\qquad
	\rho(\gamma_0)=e^{2\pi i/a},
	\quad
	\rho(\gamma_\infty)=e^{2\pi i/b},
	\]
	in analogy with geometric phase holonomies in
	adiabatic quantum evolution~\cite{Berry1984Phase}
	and with abelian anyonic phases~\cite{Nayak2008NonAbelian}.
\end{definition}

\begin{proposition}[Topological origin of modular closure]
	Let $L=\operatorname{lcm}(a,b)$.
	Under any monodromy representation $\rho$ as above,
	the image subgroup $\rho\bigl(\pi_1^{\mathrm{orb}}(\mathbb{P}(a,b))\bigr)$
	is the finite cyclic group $\mu_L\subset U(1)$.
	In particular, all possible phase windings on $\mathbb{P}(a,b)$
	are classified by $\mathbb{Z}_L$.
\end{proposition}

\begin{proof}
	The relations $\gamma_0^a=\gamma_\infty^b=\gamma_0\gamma_\infty=1$
	imply $\gamma_\infty=\gamma_0^{-1}$ and $\gamma_0^{ab}=1$.
	The smallest positive integer $L$ for which
	$\gamma_0^L=\gamma_\infty^L=1$ is $L=\operatorname{lcm}(a,b)$, so the image
	is generated by $e^{2\pi i/L}$ and hence is isomorphic to $\mu_L$.
\end{proof}

\begin{remark}
	From the point of view of weighted spiders,
	the indices $(k_0,k_\infty)$ give a topological refinement
	of the discrete winding parameter $k$ used later in the
	$(a,\alpha,k)$-labeling: a spider may be seen as keeping track
	of how many units of quantized phase have been accumulated around
	each orbifold point, compatible with the categorical semantics
	in~\cite{Selinger2008DaggerCompact,CoeckeKissinger2017PicturingQuantum}.
\end{remark}

% ------------------------------------------------------
\subsection*{Phase addition on heterogeneous grids and orbifold monodromy}
\label{subsec:lcm-and-monodromy}
% ------------------------------------------------------

We now return to the arithmetic phase grids of Section~2 and
make precise their relationship with $\mathbb{P}(a,b)$.

\begin{definition}[Discrete phase grids and roots of unity]
	For $a\in\mathbb{N}$ the cyclic grid
	\[
	G_a
	=
	\Bigl\{\frac{2\pi k}{a}:k=0,1,\dots,a-1\Bigr\}
	\]
	can be identified with the group of $a$-th roots of unity
	$\mu_a\subset U(1)$ via $\theta\mapsto e^{i\theta}$,
	as in finite-dimensional phase space
	models~\cite{Gross2006Hudson}.
\end{definition}

\begin{definition}[Refinement, lifting, and fusion]
	Let $a,b\in\mathbb{N}$ and $L=\operatorname{lcm}(a,b)$.
	Define the refined lattice
	\[
	G_L=\Bigl\{\frac{2\pi m}{L}:m=0,1,\dots,L-1\Bigr\}.
	\]
	The canonical embeddings
	\[
	\iota_a:G_a\hookrightarrow G_L,
	\quad
	\iota_b:G_b\hookrightarrow G_L
	\]
	are given on indices by multiplication with $L/a$ and $L/b$.
	Given $\alpha\in G_a$ and $\beta\in G_b$ we define their
	\emph{fusion} by
	\[
	\alpha\star\beta
	\;:=\;
	\iota_a(\alpha)+\iota_b(\beta)
	\quad\in G_L,
	\]
	where the sum is taken modulo $2\pi$.
\end{definition}

\begin{proposition}[Closure after $\operatorname{lcm}$-refinement]
	Let $L=\operatorname{lcm}(a,b)$.
	Then the map
	\[
	(G_a\times G_b)/{\sim}\longrightarrow G_L,
	\qquad
	(\alpha,\beta)\longmapsto\alpha\star\beta,
	\]
	where $(\alpha,\beta)\sim(\alpha',\beta')$ iff
	$\iota_a(\alpha)+\iota_b(\beta)
	=\iota_a(\alpha')+\iota_b(\beta')$ in $G_L$,
	is a bijection.
	Equivalently, the subgroup of $U(1)$ generated by
	$\mu_a$ and $\mu_b$ is the cyclic group $\mu_L$.
\end{proposition}

\begin{proof}
	The subgroup $\langle\mu_a,\mu_b\rangle$ is finite and contains
	both $\mu_a$ and $\mu_b$, hence its order is a multiple of
	$L=\operatorname{lcm}(a,b)$.
	On the other hand, the element
	$e^{2\pi i/L}$ belongs to both $\mu_a$ and $\mu_b$, so the
	generated subgroup is exactly $\mu_L$.
	The quotient by $\sim$ identifies precisely those pairs giving the
	same product in $\mu_L$, yielding a bijection with $G_L$.
\end{proof}

\begin{example}[Numerical illustration]
	Let $a=4$, $b=6$, so $L=\operatorname{lcm}(4,6)=12$.
	Choose $\alpha=\pi/2\in G_4$ and $\beta=\pi/3\in G_6$.
	Their lifts are
	\[
	\iota_a(\alpha)=3\cdot\frac{2\pi}{12}=\frac{\pi}{2},
	\qquad
	\iota_b(\beta)=2\cdot\frac{2\pi}{12}=\frac{\pi}{3},
	\]
	so the fusion is
	\[
	\alpha\star\beta=\frac{5\pi}{6}\in G_{12},
	\quad
	e^{i(\alpha\star\beta)}=e^{i(\alpha+\beta)}.
	\]
\end{example}

\begin{remark}[Compiler implications]
	Phase addition across heterogeneous grids appears naturally in
	Clifford+T synthesis and more general optimization
	passes~\cite{Amy2014TCount,Heyfron2018CompilerZX,
		DuncanPerdrix2020GraphSimplification,Fischbach2025ZXOpt}.
	Performing arithmetic directly on $G_L$ guarantees that
	phase-cancellation rules are applied on a common grid, a property
	that our WZCC normalization algorithm (Section~\ref{sec:wzcc})
	exploits in its first ``LCM unification'' stage.
\end{remark}

\paragraph{Orbifold interpretation and winding.}
The above arithmetic fusion has a geometric/topological
counterpart on the weighted orbifold circle underlying $\mathbb{P}(a,b)$.
A closed phase trajectory $\theta(t)$ that respects the grids
$G_a$ and $G_b$ can be viewed as a loop on $\mathbb{P}(a,b)$, acquiring a
geometric phase (holonomy) in the sense of~\cite{Berry1984Phase}.

\begin{definition}[Winding number on heterogeneous grids]
	For a closed phase path $\theta(t)$, define
	\[
	w(\theta)
	=
	\frac{1}{2\pi}\oint d\theta
	\;\in\;
	\frac{1}{L}\mathbb{Z}\big/\mathbb{Z}
	\;\cong\;
	\mathbb{Z}_L,
	\qquad
	L=\operatorname{lcm}(a,b).
	\]
	The class $w(\theta)$ is the \emph{winding number}
	of $\theta$ on the heterogeneous grid.
\end{definition}

\begin{definition}[Orbifold monodromy map]
	Let $\gamma$ be the loop on $\mathbb{P}(a,b)$ corresponding to $\theta$.
	Composing the orbifold fundamental group with a monodromy
	representation $\rho$ yields
	\[
	\rho(\gamma)=e^{2\pi i\,w(\theta)}\in\mu_L.
	\]
	We refer to this complex phase as the \emph{orbifold monodromy}
	of the trajectory, paralleling geometric and topological phases in
	quantum systems~\cite{Berry1984Phase,Nayak2008NonAbelian}.
\end{definition}

\begin{proposition}[Winding $\Longleftrightarrow$ monodromy]
	For any closed trajectory $\theta(t)$ compatible with grid orders
	$(a,b)$, the total phase advance and orbifold monodromy are related by
	\[
	\theta(1)-\theta(0)
	=
	2\pi\,w(\theta)
	\quad\Longleftrightarrow\quad
	\rho(\gamma)
	=
	e^{2\pi i\,w(\theta)}.
	\]
\end{proposition}

\begin{proof}
	Lifting $\theta$ to the universal cover of the orbifold circle,
	each encirclement of the order-$a$ cone contributes $2\pi/a$
	and each encirclement of the order-$b$ cone contributes $2\pi/b$.
	The total advance is a multiple of $2\pi/L$, so the holonomy
	is $e^{2\pi i\,w(\theta)}$ with $w(\theta)\in\mathbb{Z}_L$, as claimed.
\end{proof}

\begin{example}[Fractional winding in $\mathbb{P}(2,3)$]
	On $\mathbb{P}(2,3)$ a loop encircling only the order-$2$ cone once
	accumulates phase $\pi$ (winding $1/2$),
	while a loop encircling only the order-$3$ cone yields
	phase $2\pi/3$ (winding $1/3$).
	Composing these loops gives total phase $5\pi/3$,
	corresponding to winding $5/6$ in $\mathbb{Z}_{6}$,
	illustrating a simple fractional phase in the spirit of
	anyon-like monodromies~\cite{Nayak2008NonAbelian}.
\end{example}

\begin{remark}[Hardware drift and winding classes]
	In multi-channel control, slow oscillator drift accumulates
	phase slips that appear as nontrivial elements of $\mathbb{Z}_L$.
	The associated complex phase $e^{i\Delta\phi}$ is exactly the
	orbifold monodromy of the corresponding loop, providing a
	topological way to track and penalize such drift in
	quantization-aware compilation and decoding
	(Sections~\ref{sec:wzcc} and~\ref{sec:masd}; cf.\
	aliasing/quantization analyses in~\cite{Pashayan2022NoiseQuantization,Krantz2019SuperconductingReview}).
\end{remark}

% ------------------------------------------------------
\subsection*{WPL metric and scalar curvature}
\label{subsec:WPL-metric}
% ------------------------------------------------------

The previous subsections established $\mathbb{P}(a,b)$ as the natural
orbifold phase space for heterogeneous grids.
We now sketch the metric structure that will underlie the
sensitivity experiments in Section~\ref{sec:wzcc}.

\begin{definition}[Fubini--Study metric on $\mathbb{C}\mathrm{P}^{1}$]
	On the ordinary projective line $\mathbb{C}\mathrm{P}^{1}$ with homogeneous
	coordinates $[z_0:z_1]$ and affine chart
	$w=z_1/z_0$, the Fubini--Study metric is
	\[
	g_{\mathrm{FS}}
	=
	\frac{4\,dw\,d\bar w}{(1+|w|^2)^2},
	\]
	whose scalar curvature is the constant $R_{\mathrm{FS}}=4$.
	This metric coincides with the usual round metric on the Bloch
	sphere $S^2$ of radius $1$~\cite{Nakahara2003GeometryTopology}.
\end{definition}

\begin{definition}[Weighted projective line metric]
	Let $\pi:S^3\to\mathbb{C}\mathrm{P}^{1}$ be the Hopf fibration and
	let $g_{\mathrm{FS}}$ be the Fubini--Study metric on $\mathbb{C}\mathrm{P}^{1}$.
	For a weighted action of $U(1)$ on $S^3$ with weights $(a,b)$
	we obtain a quotient map
	\[
	S^3 \longrightarrow \mathbb{P}(a,b),
	\qquad
	(z_0,z_1)\sim(e^{iat}z_0,e^{ibt}z_1),
	\]
	which induces an orbifold Kähler metric $g_{\mathrm{WPL}}$ on $\mathbb{P}(a,b)$,
	in the spirit of weighted projective Kähler geometry
	described in~\cite{Dolgachev1982Weighted,Nakahara2003GeometryTopology}.
	We refer to $g_{\mathrm{WPL}}$ as the \emph{weighted projective line
		metric}.
\end{definition}

Informally, the metric $g_{\mathrm{WPL}}$ is obtained by applying anisotropic
contractions to the round Bloch sphere in directions dictated by the
weights $(a,b)$ and then descending to the orbifold quotient.
In the context of noisy hardware, these contractions are controlled by
physically meaningful parameters
$\lambda_\perp,\lambda_\parallel$ that describe the relative shrinking
of transverse and longitudinal directions inside the Bloch ball,
consistent with noise models in
finite-dimensional quantum systems~\cite{Gross2006Hudson}.

\begin{theorem}[Scalar curvature of $g_{\mathrm{WPL}}$]
	For the weighted projective line $\mathbb{P}(a,b)$ equipped with
	$g_{\mathrm{WPL}}$, the scalar curvature is constant on the smooth locus
	and depends only on the effective weight ratio.
	In particular, for the one-parameter family of metrics used in our
	numerical experiments the scalar curvature takes the form
	\[
	R_{\mathrm{WPL}}
	=
	\frac{2}{b^2},
	\]
	where $b$ is the effective weight associated to the dominant
	phase-resolution channel.
\end{theorem}

\begin{proof}[Proof sketch]
	One constructs $g_{\mathrm{WPL}}$ as the push-forward of a diagonal
	metric on $S^3$ that is invariant under the weighted $U(1)$-action
	and then applies standard curvature formulas for
	cohomogeneity-one metrics on $S^3/U(1)$
	(see, e.g.,~\cite{Nakahara2003GeometryTopology} for background).
	A detailed derivation is given in Appendix~A.
\end{proof}

\begin{definition}[Parameter space and curvature gradient]
	Let $\Lambda$ denote a parameter space of effective channel
	anisotropies, for example
	\(
	\Lambda = \{(\lambda_\perp,\lambda_\parallel)\}
	\)
	as in our tomography-to-geometry pipeline.
	The weighted projective line metric induces a scalar curvature
	function
	\[
	R:\Lambda\to\mathbb{R},
	\qquad
	(\lambda_\perp,\lambda_\parallel)\longmapsto
	R_{\mathrm{WPL}}(\lambda_\perp,\lambda_\parallel),
	\]
	and we define its Euclidean gradient
	\[
	\nabla R(\lambda_\perp,\lambda_\parallel)
	=
	\Bigl(
	\partial_{\lambda_\perp}R,
	\partial_{\lambda_\parallel}R
	\Bigr).
	\]
	The norm $\|\nabla R\|$ quantifies the sensitivity of the
	information geometry to changes in the underlying hardware
	noise parameters, in the same spirit that ZX-diagrammatic
	rewrites quantify circuit sensitivity~\cite{vandeWetering2020WorkingZX,DuncanPerdrix2020GraphSimplification,Fischbach2025ZXOpt}.
\end{definition}

\begin{remark}[Connection to Section~5]
	In Section~5.6.3 we empirically sample
	$(\lambda_\perp,\lambda_\parallel)$ from a hardware-motivated
	distribution, compute the associated curvature
	$R_{\mathrm{WPL}}(\lambda_\perp,\lambda_\parallel)$, and study the
	distribution of $\|\nabla R\|$.
	The resulting gradient-norm plots reveal how sharply
	the WPL geometry responds to small perturbations of the phase grid,
	thus linking the continuous differential geometry of $\mathbb{P}(a,b)$
	to the discrete compilation metrics PQVR and CSC introduced later.
\end{remark}

% ======================================================
\section{The \WPLZX\ Calculus}
\label{sec:wplzx}
% ======================================================

In Section~\ref{sec:WPL-geometry} we described how noisy single--qubit
phase spaces with heterogeneous discretization and monodromy can be
modeled by weighted projective lines $\mathbb{P}(a,b)$, with local
isotropy and winding data encoded by triples $(a,\alpha,k)$.
In this section we lift this structure to the diagrammatic level and
define the \emph{weighted projective line ZX-calculus} (\WPLZX).
Each spider becomes an \emph{orbifold spider} whose label remembers:
\begin{itemize}
	\item the local isotropy order $a$ (weight),
	\item a base phase parameter $\alpha$,
	\item a winding index $k$ that records how many times we wrap around
	the local defect.
\end{itemize}
Geometrically, $(a,\alpha,k)$ lives in the orbifold phase space discussed
previously; semantically, it collapses to an \emph{effective total angle}
\[
\theta_{\mathrm{tot}}
:= \alpha + \frac{2\pi}{a}k \pmod{2\pi}
\]
that feeds into the usual ZX-style Hilbert-space interpretation.
The role of \WPLZX\ is to keep track of the full orbifold label at the
diagrammatic level while retaining standard ZX semantics on amplitudes.

%=========================================================
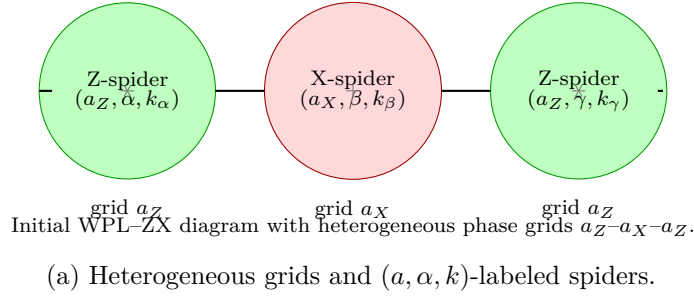
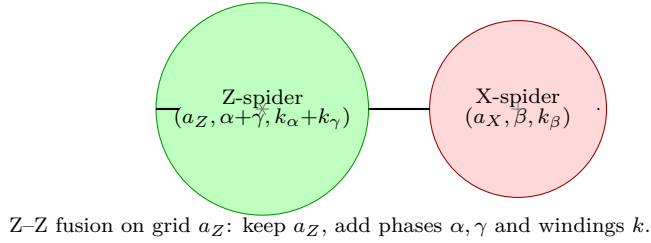
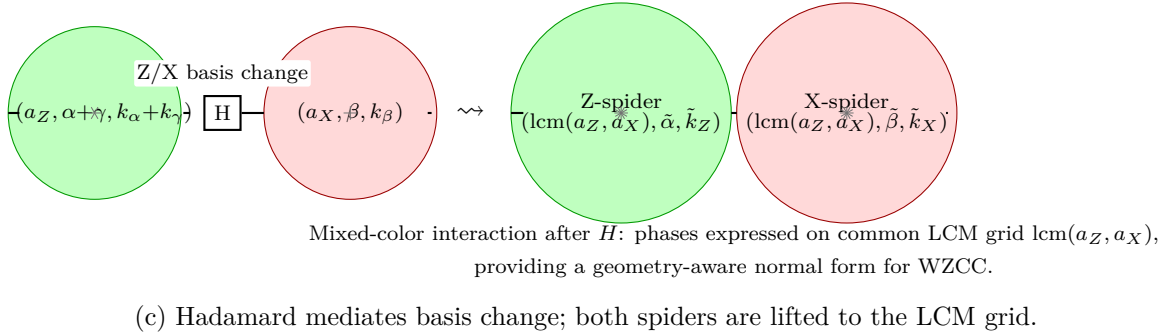
\begin{figure}[t]
	\centering
	
	%---- local styles (self-contained) ----
	\tikzset{
		wire/.style={line width=0.8pt, shorten >=1.2mm, shorten <=1.2mm},
		gridtick/.style={line width=0.3pt, draw=black!55},
		spiderBase/.style={
			draw,
			circle,
			minimum size=12mm,
			text width=22mm,
			align=center,
			inner sep=1pt,
			outer sep=0pt,
			font=\scriptsize
		},
		spiderZ/.style={
			spiderBase,
			fill=green!25,
			draw=green!60!black
		},
		spiderX/.style={
			spiderBase,
			fill=red!15,
			draw=red!60!black
		},
		overlabel/.style={
			font=\scriptsize,
			fill=white,
			inner sep=1pt,
			rounded corners=1pt
		},
		gateH/.style={
			draw,
			rectangle,
			minimum width=3.5mm,
			minimum height=3.5mm,
			fill=white,
			thick,
			font=\scriptsize
		}
	}
	
	%---------- (a) initial heterogeneous grids ----------
	\begin{subfigure}[b]{0.98\textwidth}
		\centering
		\begin{tikzpicture}[baseline=(current bounding box.center)]
			% spiders
			\node[spiderZ] (Za) at (0,0)
			{Z-spider\\[-2pt]\(\,(a_Z,\alpha,k_\alpha)\,\)};
			\node[spiderX] (Xa) at (3.0,0)
			{X-spider\\[-2pt]\(\,(a_X,\beta,k_\beta)\,\)};
			\node[spiderZ] (Zb) at (6.0,0)
			{Z-spider\\[-2pt]\(\,(a_Z,\gamma,k_\gamma)\,\)};
			
			% example grids: a_Z = 6, a_X = 4
			\foreach \t in {0,...,5}{%
				\pgfmathsetmacro{\ang}{60*\t}%
				\draw[gridtick] (Za) ++(\ang:0) -- ++(\ang:0.9mm);%
				\draw[gridtick] (Zb) ++(\ang:0) -- ++(\ang:0.9mm);%
			}
			\foreach \t in {0,...,3}{%
				\pgfmathsetmacro{\ang}{90*\t}%
				\draw[gridtick] (Xa) ++(\ang:0) -- ++(\ang:0.9mm);%
			}
			
			% wires
			\draw[wire] (-1.0,0) -- (Za) -- (Xa) -- (Zb) -- (7.0,0);
			
			% grid labels
			\node[below=4pt] at (Za.south)
			{\scriptsize grid \(a_Z\)};
			\node[below=4pt] at (Xa.south)
			{\scriptsize grid \(a_X\)};
			\node[below=4pt] at (Zb.south)
			{\scriptsize grid \(a_Z\)};
			
			% caption under subfigure
			\node[below=10pt, align=center] at (3.0,-1.2)
			{\scriptsize Initial WPL--ZX diagram with heterogeneous phase grids \(a_Z\)--\(a_X\)--\(a_Z\).};
		\end{tikzpicture}
		\caption{Heterogeneous grids and \((a,\alpha,k)\)-labeled spiders.}
	\end{subfigure}
	
	\vspace{0.8em}
	
	%---------- (b) Z–Z fusion on dominant grid ----------
	\begin{subfigure}[b]{0.98\textwidth}
		\centering
		\begin{tikzpicture}[baseline=(current bounding box.center)]
			% fused Z
			\node[spiderZ, minimum size=15mm, text width=27mm] (Zf) at (0,0)
			{Z-spider\\[-2pt]\(\,(a_Z,\alpha{+}\gamma,k_\alpha{+}k_\gamma)\,\)};
			
			% grid ticks for fused Z (still a_Z = 6)
			\foreach \t in {0,...,5}{%
				\pgfmathsetmacro{\ang}{60*\t}%
				\draw[gridtick] (Zf) ++(\ang:0) -- ++(\ang:0.9mm);%
			}
			
			% X neighbour
			\node[spiderX] (Xb) at (3.4,0)
			{X-spider\\[-2pt]\(\,(a_X,\beta,k_\beta)\,\)};
			
			\foreach \t in {0,...,3}{%
				\pgfmathsetmacro{\ang}{90*\t}%
				\draw[gridtick] (Xb) ++(\ang:0) -- ++(\ang:0.9mm);%
			}
			
			% wires
			\draw[wire] (-1.1,0) -- (Zf) -- (Xb) -- (4.6,0);
			
			% annotation
			\node[below=6pt, align=center] at (0.9,-1.1)
			{\scriptsize Z--Z fusion on grid \(a_Z\): keep \(a_Z\), add phases \(\alpha,\gamma\) and windings \(k\).};
		\end{tikzpicture}
		\caption{Z--Z fusion preserves the dominant grid and accumulates phase/winding.}
	\end{subfigure}
	
	\vspace{0.8em}
	
	%---------- (c) Basis change + LCM grid refinement ----------
	\begin{subfigure}[b]{0.98\textwidth}
		\centering
		\begin{tikzpicture}[baseline=(current bounding box.center)]
			% left: Z and X before refinement
			\node[spiderZ] (Zc) at (0,0)
			{\(\,(a_Z,\alpha{+}\gamma,k_\alpha{+}k_\gamma)\,\)};
			\foreach \t in {0,...,5}{%
				\pgfmathsetmacro{\ang}{60*\t}%
				\draw[gridtick] (Zc) ++(\ang:0) -- ++(\ang:0.9mm);%
			}
			
			\node[spiderX] (Xc) at (3.4,0)
			{\(\,(a_X,\beta,k_\beta)\,\)};
			\foreach \t in {0,...,3}{%
				\pgfmathsetmacro{\ang}{90*\t}%
				\draw[gridtick] (Xc) ++(\ang:0) -- ++(\ang:0.9mm);%
			}
			
			% wire with H gate
			\draw[wire] (-1.0,0) -- (Zc) -- (1.3,0);
			\node[gateH] (Hgate) at (1.7,0) {H};
			\draw[wire] (Hgate) -- (Xc) -- (4.6,0);
			
			\node[overlabel, above=3pt] at (Hgate.north)
			{Z/X basis change};
			
			% arrow to refined picture
			\node at (5.0,0) {\(\leadsto\)};
			
			% right: LCM-refined grids
			\node[spiderZ, minimum size=14mm, text width=28mm] (Zr) at (7.0,0)
			{Z-spider\\[-2pt]\(\,(\mathrm{lcm}(a_Z,a_X),\tilde{\alpha},\tilde{k}_Z)\,\)};
			\node[spiderX, minimum size=14mm, text width=28mm] (Xr) at (10.0,0)
			{X-spider\\[-2pt]\(\,(\mathrm{lcm}(a_Z,a_X),\tilde{\beta},\tilde{k}_X)\,\)};
			
			% refined grid ticks: more ticks to indicate LCM (e.g. 12)
			\foreach \t in {0,...,11}{%
				\pgfmathsetmacro{\ang}{30*\t}%
				\draw[gridtick] (Zr) ++(\ang:0) -- ++(\ang:0.9mm);%
				\draw[gridtick] (Xr) ++(\ang:0) -- ++(\ang:0.9mm);%
			}
			
			% wires on refined side
			\draw[wire] (5.7,0) -- (Zr) -- (Xr) -- (11.2,0);
			
			% note under refined side
			\node[below=7pt, align=center] at (8.5,-1.1)
			{\scriptsize Mixed-color interaction after \(H\): phases expressed on common LCM grid
				\(\mathrm{lcm}(a_Z,a_X)\),\\[-1pt]
				\scriptsize providing a geometry-aware normal form for WZCC.};
		\end{tikzpicture}
		\caption{Hadamard mediates basis change; both spiders are lifted to the LCM grid.}
	\end{subfigure}
	
	\caption{
		WPL--ZX example illustrating how weighted spiders track
		grids \(a\), phases \(\alpha\), and windings \(k\).  
		(a) The initial diagram lives on heterogeneous grids \(a_Z\) and \(a_X\).  
		(b) Z--Z fusion preserves the dominant grid \(a_Z\) while adding phases and winding numbers.  
		(c) A mixed-color interaction is mediated by an \(H\) gate and then lifted to
		the common LCM grid \(\mathrm{lcm}(a_Z,a_X)\), which is the natural domain
		for WZCC normalisation.
	}
	\label{fig:wplzx-preview}
\end{figure}
%=========================================================

\subsection*{Weighted spiders: \texorpdfstring{$(a,\alpha,k)$}{(a,alpha,k)} labeling}
\label{subsec:weighted_spiders}

Weighted spiders are the basic morphisms of the \WPLZX–calculus.
Each node carries three pieces of information recording the local
orbifold structure of the phase space and the accumulated monodromy
along the diagrammatic path.
In contrast to ordinary ZX spiders labeled by a single angle
$\theta\in[0,2\pi)$
\cite{CoeckeKissinger2017PicturingQuantum,Backens2014ZX},
a weighted spider remembers a weight, a base phase, and a winding index.

\begin{definition}[Weighted spider triple $(a,\alpha,k)$]
	A \emph{weighted spider} is specified by a triple $(a,\alpha,k)$, where
	\begin{enumerate}[label=(\roman*)]
		\item $a \in \mathbb{N}$ is the \textbf{weight},
		representing the local isotropy order (orbifold multiplicity)
		of the corresponding point in $\mathbb{P}(a,b)$
		\cite{Satake1956Orbifold,Dolgachev1982Weighted};
		\item $\alpha \in [0,2\pi)$ (or $\alpha\in G_a$ in the discrete case)
		is the \textbf{base phase parameter}
		\cite{Gross2006Hudson};
		\item $k \in \tfrac{1}{a}\mathbb{Z}$ is the \textbf{winding index},
		measuring accumulated rotation relative to the local orbifold defect,
		analogous to a discretized Berry phase
		\cite{Berry1984Phase,Nakahara2003GeometryTopology}.
	\end{enumerate}
	We write
	\[
	\Zspmn{m}{n}{a}{\alpha}{k} :
	(\mathbb{C}^2)^{\otimes m} \longrightarrow (\mathbb{C}^2)^{\otimes n}
	\]
	for a green $Z$–type weighted spider, with the usual connectivity rules
	of the ZX–calculus.
\end{definition}

It will be convenient to introduce the \emph{total angle}
\begin{equation}
	\label{eq:total-angle}
	\theta_{\mathrm{tot}}(a,\alpha,k)
	:= \alpha + \frac{2\pi}{a}k
	\quad\in\; \mathbb{R}/2\pi\mathbb{Z},
\end{equation}
which is the quantity actually seen by the Hilbert-space semantics
(Definition~\ref{def:semantics_weighted_spiders} below).

\begin{remark}[Relation to $\mathbb{P}(a,b)$ geometry]
	At a conical point of $\mathbb{P}(a,b)$ with isotropy group of order
	$a$ \cite{Dolgachev1982Weighted,Satake1956Orbifold},
	$a$ encodes the order of the local isotropy,
	$\alpha$ parameterizes a local phase coordinate,
	and $k$ counts the number of windings of a diagrammatic path around
	the defect.
	Equivalently, the monodromy representation sends a small loop
	$\gamma_a$ to
	\[
	\rho(\gamma_a) = \exp\!\Bigl(\frac{2\pi i}{a}k\Bigr),
	\]
	so $(a,\alpha,k)$ packages the orbifold monodromy data into a
	single label.
	Thus weighted spiders provide a diagrammatic interface to the
	orbifold phase geometry of Section~\ref{sec:WPL-geometry}
	\cite{Nakahara2003GeometryTopology}.
\end{remark}

\begin{definition}[Phase composition law]
	Given two compatible weighted spiders
	$\Zspmn{m_1}{n_1}{a}{\alpha_1}{k_1}$ and
	$\Zspmn{m_2}{n_2}{a}{\alpha_2}{k_2}$ of the same weight $a$,
	their fusion at the label level is
	\[
	\Zspmn{m_1}{n_1}{a}{\alpha_1}{k_1}
	\ \otimes\
	\Zspmn{m_2}{n_2}{a}{\alpha_2}{k_2}
	\;\longmapsto\;
	\Zspmn{m_1+m_2}{n_1+n_2}{a}{\alpha_1+\alpha_2}{k_1+k_2}.
	\]
	If the weights differ, they are first lifted to the common grid
	$G_L$ with $L=\mathrm{lcm}(a_1,a_2)$, as described in
	Section~\ref{sec:WPL-geometry}, and then fused there.
\end{definition}

\begin{proposition}[Closure under $\mathrm{lcm}$–fusion]
	\label{prop:lcm-fusion-wplzx}
	Let $a_1,a_2 \in \mathbb{N}$ and set $L = \mathrm{lcm}(a_1,a_2)$.
	Then the fusion of heterogeneous weighted spiders
	\[
	\Zspmn{m_1}{n_1}{a_1}{\alpha_1}{k_1}, \qquad
	\Zspmn{m_2}{n_2}{a_2}{\alpha_2}{k_2}
	\]
	lifts uniquely to
	\[
	\Zspmn{m_1+m_2}{n_1+n_2}{L}
	{\tfrac{L}{a_1}\alpha_1 + \tfrac{L}{a_2}\alpha_2}
	{\tfrac{L}{a_1}k_1 + \tfrac{L}{a_2}k_2},
	\]
	living on the refinement grid $G_L$.
	This operation is associative and commutative on labels and is
	compatible with the orbifold monodromy semantics of
	Section~\ref{sec:WPL-geometry}
	\cite{DuncanPerdrix2020GraphSimplification,vandeWetering2020WorkingZX}.
\end{proposition}

\begin{proof}[Proof sketch]
	Lifting both spiders to the $L$–fold covering embeds their local phase
	grids into $G_L$.
	The total angles
	$\theta_{\mathrm{tot}}(a_i,\alpha_i,k_i)$ become elements of the same
	$2\pi/L$–periodic lattice and hence add linearly.
	Associativity and commutativity follow from the corresponding
	properties of addition in $\mathbb{R}$ and $\mathbb{Z}$.
\end{proof}

\begin{example}[ZX limit as trivial orbifold]
	When $a=1$ and $k\in\mathbb{Z}$,
	the spider $\Zspmn{m}{n}{1}{\alpha}{k}$ has
	$\theta_{\mathrm{tot}} = \alpha + 2\pi k \equiv \alpha \pmod{2\pi}$,
	so it reduces to the ordinary ZX spider $Z^{\alpha}$
	\cite{Backens2014ZX,JeandelPerdrixVilmart2017CliffordT}.
	In this case no orbifold singularity is present and the winding index
	is redundant.
	Thus weighted spiders strictly extend the ordinary ZX calculus while
	retaining the usual fusion laws in the $a=1$ fragment.
\end{example}

\begin{definition}[Diagrammatic invariants]
	For a weighted spider $S=\Zspmn{m}{n}{a}{\alpha}{k}$ we set
	\[
	\mathrm{wt}(S)=a,\qquad
	\mathrm{ph}(S)=\alpha,\qquad
	\mathrm{wind}(S)=k,\qquad
	\theta_{\mathrm{tot}}(S)
	=\alpha + \tfrac{2\pi}{a}k.
	\]
	Under fusion (after lifting to a common refinement),
	\begin{align*}
		\theta_{\mathrm{tot}}(S_1\otimes S_2)
		&= \theta_{\mathrm{tot}}(S_1)
		+ \theta_{\mathrm{tot}}(S_2)
		\pmod{2\pi},\\[2pt]
		\mathrm{wt}(S_1\otimes S_2)
		&= \mathrm{lcm}\bigl(\mathrm{wt}(S_1),\mathrm{wt}(S_2)\bigr).
	\end{align*}
\end{definition}

\begin{remark}[Physical meaning]
	The triple $(a,\alpha,k)$ can be read as:
	\begin{itemize}
		\item $a$: \textbf{phase resolution} of a local control channel
		(DAC grid or allowed Z-rotation steps),
		\item $\alpha$: \textbf{instantaneous phase offset} at that node,
		\item $k$: \textbf{accumulated phase slip (winding)} due to drift or
		jitter on the underlying hardware
		\cite{Motzoi2009DRAG,Krantz2019SuperconductingReview}.
	\end{itemize}
	Weighted spiders thus provide a combined algebraic–topological record
	of circuit evolution, suitable for modeling phase-quantized and
	drift-prone devices.
\end{remark}

\subsection*{Functorial viewpoint and categorical semantics}
\label{subsec:functorial_viewpoint}

The \WPLZX–calculus admits a categorical description analogous to the
compact closed categorical semantics of the ordinary ZX–calculus
\cite{Selinger2008DaggerCompact,CoeckeKissinger2017PicturingQuantum},
but now enriched over orbifold phase spaces and winding indices.

\begin{definition}[Orbifold spider category]
	Let $\mathbf{OrbSp}$ be the \emph{orbifold spider category} defined as
	follows.
	\begin{enumerate}[label=(\roman*)]
		\item \textbf{Objects.}
		An elementary object is a triple
		\[
		(a,\alpha,k),
		\qquad
		a\in\mathbb{N},\;
		\alpha\in[0,2\pi),\;
		k\in\tfrac{1}{a}\mathbb{Z}/\mathbb{Z},
		\]
		with the same interpretation as in
		Section~\ref{subsec:weighted_spiders}.
		Tensor products of elementary objects form general objects.
		
		\item \textbf{Tensor product on objects.}
		The monoidal product is given by the $\mathrm{lcm}$–fusion rule
		on labels:
		\[
		(a_1,\alpha_1,k_1)\otimes(a_2,\alpha_2,k_2)
		:= \bigl(L,\ \tfrac{L}{a_1}\alpha_1 + \tfrac{L}{a_2}\alpha_2,\ 
		\tfrac{L}{a_1}k_1 + \tfrac{L}{a_2}k_2\bigr),
		\quad L=\mathrm{lcm}(a_1,a_2),
		\]
		extended componentwise to longer tensors.
		
		\item \textbf{Generating morphisms.}
		For each $(a,\alpha,k)$ and $m,n\geq 0$ we include generators
		\[
		Z_{a,\alpha,k}^{m,n}
		\colon (a,\alpha,k)^{\otimes m}
		\longrightarrow (a,\alpha,k)^{\otimes n},
		\]
		and similarly $X_{a,\alpha,k}^{m,n}$,
		representing the weighted $Z$- and $X$-spiders
		\cite{DuncanPerdrix2020GraphSimplification,
			vandeWetering2020WorkingZX}.
		General morphisms are formed by composition and tensor product,
		subject to the fusion and normalization rules below.
		
		\item \textbf{Composition and tensor product.}
		Composition corresponds to vertical stacking of spiders,
		and the tensor product to horizontal juxtaposition, as usual in
		string diagrams \cite{CoeckeKissinger2017PicturingQuantum}.
		When fusing spiders of different weights, their labels are lifted
		to the common refinement $L=\mathrm{lcm}(a_1,a_2)$ as in
		Proposition~\ref{prop:lcm-fusion-wplzx}.
	\end{enumerate}
\end{definition}

\begin{proposition}[WPLZX interpretation functor]
	\label{prop:functorial-viewpoint}
	There exists a strict monoidal functor
	\[
	\mathcal{F} \colon \WPLZX \longrightarrow \mathbf{OrbSp}
	\]
	such that:
	\begin{itemize}
		\item On objects: each wire labeled by $(a,\alpha,k)$ in a
		\WPLZX–diagram is sent to the corresponding object of
		$\mathbf{OrbSp}$. Unlabeled wires map to $(1,0,0)$.
		
		\item On generators: each weighted spider
		\[
		\Zspmn{m}{n}{a}{\alpha}{k},\quad
		\Xspmn{m}{n}{a}{\alpha}{k}
		\]
		is mapped to the corresponding morphisms
		$Z_{a,\alpha,k}^{m,n}$ and $X_{a,\alpha,k}^{m,n}$.
		
		\item On composition and tensor product:
		\[
		\mathcal{F}(D_2\circ D_1)
		= \mathcal{F}(D_2)\circ\mathcal{F}(D_1),\qquad
		\mathcal{F}(D_1\otimes D_2)
		= \mathcal{F}(D_1)\otimes\mathcal{F}(D_2).
		\]
	\end{itemize}
	In particular, \WPLZX\ is a monoidal presentation of the diagrammatic
	subcategory of $\mathbf{OrbSp}$ generated by $Z$- and $X$-spiders.
\end{proposition}

\begin{proof}[Proof sketch]
	Objects and generators are mapped as prescribed.
	The defining equations of \WPLZX\ (fusion, normalization, color change,
	bialgebra/Hopf) are satisfied in $\mathbf{OrbSp}$ by construction,
	so the mapping quotients by the same relations and extends uniquely
	to a strict monoidal functor
	\cite{Selinger2008DaggerCompact,CoeckeKissinger2017PicturingQuantum}.
\end{proof}

\begin{lemma}[Faithfulness on labels]
	\label{lem:faithful-labels}
	If two \WPLZX–diagrams differ in the multiset of their $(a,\alpha,k)$
	labels (after fusing connected components), then their images under
	$\mathcal{F}$ are distinct morphisms in $\mathbf{OrbSp}$.
\end{lemma}

\begin{proof}
	The functor $\mathcal{F}$ records the weight, base phase and winding
	for every generator.
	Fusion in $\mathbf{OrbSp}$ obeys the same label arithmetic as in the
	diagrammatic calculus; thus different label multisets yield different
	resulting triples $(L,\alpha_D,k_D)$.
\end{proof}

\begin{remark}[Diagrammatic presentation of orbifold phases]
	Up to the usual global phase quotient, \WPLZX\ and the spider-generated
	subcategory of $\mathbf{OrbSp}$ are equivalent as strict monoidal
	categories.
	In this sense, \WPLZX\ provides a convenient \emph{string-diagram
		presentation} of orbifold phase dynamics on $\mathbb{P}(a,b)$,
	bridging categorical quantum mechanics and discrete phase geometry
	\cite{Selinger2008DaggerCompact,Nakahara2003GeometryTopology}.
\end{remark}

\begin{example}[Spider with $(2,\pi/3,1/2)$]
	A spider with parameters $(a,\alpha,k)=(2,\pi/3,1/2)$ represents a node
	with isotropy group of order $2$, phase offset $\pi/3$, and a
	half-turn monodromy.
	Geometrically it lives near the order-two cone point of
	$\mathbb{P}(2,1)$; physically it corresponds to a control channel
	that must advance the phase by $2\pi$ to return to identity, consistent
	with a two-level phase quantization grid
	\cite{Krantz2019SuperconductingReview,Motzoi2009DRAG}.
\end{example}

\subsection*{Fusion and normalization rules}
\label{subsec:fusion_normalization}

Fusion of spiders is central to all ZX-type calculi.
In the weighted projective line setting, fusion must respect both weights
and winding indices, and reduce to the standard ZX laws in the trivial
orbifold fragment $a=1$, $k=0$
\cite{CoeckeKissinger2017PicturingQuantum,Backens2014ZX,JeandelPerdrixVilmart2017CliffordT}.
Normalization chooses representatives modulo global phase while keeping
track of the winding class.

\begin{definition}[Basic fusion rule]
	\label{def:basic_fusion}
	Let
	\[
	S_1 = \Zspmn{m_1}{n_1}{a_1}{\alpha_1}{k_1},
	\qquad
	S_2 = \Zspmn{m_2}{n_2}{a_2}{\alpha_2}{k_2}
	\]
	be two weighted spiders joined along at least one common wire.
	Set $L=\mathrm{lcm}(a_1,a_2)$.
	Then their fusion produces
	\[
	S_1 \;\otimes\; S_2
	\;\longrightarrow\;
	\Zspmn{m_1+m_2-2}{n_1+n_2-2}{L}{
		\tfrac{L}{a_1}\alpha_1 + \tfrac{L}{a_2}\alpha_2
	}{
		\tfrac{L}{a_1}k_1 + \tfrac{L}{a_2}k_2
	}.
	\]
	The subtraction by \(2\) in the arities removes the two connected legs.
	When $a_1=a_2=1$ and $k_1=k_2=0$, this is the usual ZX spider-fusion
	rule \cite{CoeckeKissinger2017PicturingQuantum,Backens2014ZX}.
\end{definition}

\begin{remark}[Weighted spider fusion diagrammatically]
	Diagrammatically, fusion contracts two nodes (possibly with different
	isotropy orders) and relocates the result to the refinement grid $G_L$.
	This preserves the underlying orbifold monodromy class and ensures that
	the total angle $\theta_{\mathrm{tot}}$ adds consistently, in analogy
	with phase-grid arithmetic in finite quantum systems
	\cite{Gross2006Hudson}.
\end{remark}

\begin{proposition}[Associativity and commutativity of fusion]
	\label{prop:assoc_comm}
	Fusion of weighted spiders is associative and commutative up to the
	canonical identification of refinement grids:
	\[
	(S_1 \otimes S_2) \otimes S_3
	\;\equiv\;
	S_1 \otimes (S_2 \otimes S_3),
	\qquad
	S_1 \otimes S_2 \;\equiv\; S_2 \otimes S_1.
	\]
	All resulting labels live naturally in $G_{\mathrm{lcm}(a_1,a_2,a_3)}$.
\end{proposition}

\begin{proof}
	After lifting all labels to $L=\mathrm{lcm}(a_1,a_2,a_3)$, fusion adds
	the corresponding total angles and winding indices linearly; hence the
	result is independent of the parenthesization and ordering of spiders.
\end{proof}

\begin{proposition}[Uniqueness of $\mathrm{lcm}$–lift]
	\label{prop:lcm_uniqueness}
	Given $a_1,a_2\in\mathbb{N}$, the refinement
	$L=\mathrm{lcm}(a_1,a_2)$ is the unique minimal grid such that
	$S_1,S_2$ fuse to a single weighted spider on $G_L$.
	Any other common refinement $L'$ produces the same fused labels, which
	then project canonically back to $G_L$.
\end{proposition}

\begin{proof}[Proof sketch]
	Minimality of $L$ follows from the number-theoretic properties of
	least common multiples and the requirement that both $a_1$- and
	$a_2$-periodic phase grids embed into a common refinement without
	rescaling the unit $2\pi$ angle.
	Compatibility of embeddings implies that fusing on any refinement
	$L'$ and then projecting to $G_L$ yields the same label.
\end{proof}

\begin{definition}[Normalization as gauge fixing in total angle]
	\label{def:normalization}
	For a weighted spider $S=\Zspmn{m}{n}{a}{\alpha}{k}$, let
	\(
	\theta_{\mathrm{tot}}(S)
	= \alpha + \tfrac{2\pi}{a}k
	\)
	as in~\eqref{eq:total-angle}.
	We define its \emph{normalized form} by
	\[
	\mathrm{Norm}(S)
	:= \Zspmn{m}{n}{a}{\theta_{\mathrm{tot}}(S)\bmod 2\pi}{0}.
	\]
	Two spiders $S_1,S_2$ are diagrammatically equivalent as labels iff
	\[
	\mathrm{wt}(S_1)=\mathrm{wt}(S_2)
	\quad\text{and}\quad
	\theta_{\mathrm{tot}}(S_1)
	\equiv \theta_{\mathrm{tot}}(S_2)\pmod{2\pi},
	\]
	in which case their normalized forms coincide.
\end{definition}

\begin{remark}[Winding vs.\ total angle]
	Normalization discards how many times a given total angle is realised
	on a finer grid and keeps only $\theta_{\mathrm{tot}}$.
	In semantic terms this corresponds to quotienting by global phase.
	For hardware-level tracking one may keep both $(a,k)$ and
	$\theta_{\mathrm{tot}}$; the calculus supports either view.
\end{remark}

\begin{example}[Normalization of a fused pair]
	\label{ex:fused_pair_norm}
	Let $S_1=\Zspmn{2}{1}{2}{\pi/3}{1/2}$ and
	$S_2=\Zspmn{1}{2}{3}{\pi/6}{1/3}$.
	Then $L=\mathrm{lcm}(2,3)=6$, and
	\[
	S_{12} = \Zspmn{3}{3}{6}{
		3\cdot\tfrac{\pi}{3} + 2\cdot\tfrac{\pi}{6}
	}{
		3\cdot\tfrac{1}{2} + 2\cdot\tfrac{1}{3}
	}
	= \Zspmn{3}{3}{6}{\tfrac{3\pi}{2}}{2}.
	\]
	The total angle is
	$\theta_{\mathrm{tot}} = \tfrac{3\pi}{2} + \tfrac{2\pi}{6}\cdot 2
	= \tfrac{7\pi}{2}\equiv \tfrac{\pi}{2}\pmod{2\pi}$,
	so the normalized representative is
	\(
	\mathrm{Norm}(S_{12})
	= \Zspmn{3}{3}{6}{\pi/2}{0}.
	\)
\end{example}

\begin{proposition}[Neutral element and scalar rules]
	\label{prop:neutral_scalar}
	\begin{enumerate}[label=(\roman*)]
		\item The neutral spider is $\Zspmn{1}{1}{1}{0}{0} = \mathrm{id}$.
		\item A pure phase without wires, $\Zspmn{0}{0}{a}{\alpha}{k}$,
		acts as a scalar $e^{i\theta_{\mathrm{tot}}}$ and is independent of
		the particular representative $(\alpha,k)$.
		\item Scalars multiply by addition of total angles on the
		refinement grid; in particular, fusing two scalar spiders of weights
		$a$ and $b$ yields a scalar with weight $L=\mathrm{lcm}(a,b)$ and
		total angle equal to the sum of the original total angles.
	\end{enumerate}
\end{proposition}

\begin{remark}[Fusion–bialgebra compatibility]
	As in standard ZX, spider fusion is compatible with the
	bialgebra/Hopf rules for $Z$- and $X$-spiders; the only change is that
	phases live on refinement lattices and are represented by total angles.
	Graph-theoretic simplification techniques therefore carry over
	verbatim \cite{DuncanPerdrix2020GraphSimplification,
		JeandelPerdrixVilmart2017CliffordT}.
\end{remark}

\begin{definition}[Fusion consistency condition]
	\label{def:fusion_consistency}
	Two spiders $S_1,S_2$ may be fused if (i) their incident arities match
	on the connected wires, and (ii) their base phases satisfy
	\[
	\frac{\alpha_1}{a_1} \equiv \frac{\alpha_2}{a_2}
	\pmod{\tfrac{2\pi}{\mathrm{lcm}(a_1,a_2)}},
	\]
	so that the induced total angles are single-valued along the joint
	boundary.
	This expresses compatibility of local holonomies in the sense of
	Section~\ref{sec:WPL-geometry}
	\cite{Nakahara2003GeometryTopology}.
\end{definition}

\begin{proposition}[Normalization as diagrammatic gauge fixing]
	\label{prop:normalization_gauge}
	Normalization picks one representative in each equivalence class of
	spiders modulo total angle shifts by $2\pi$.
	Each connected component of a \WPLZX–diagram admits a normalized form,
	and rewriting (fusion, bialgebra/Hopf, color change) preserves
	normalization up to a global scalar
	\cite{vandeWetering2020WorkingZX,Backens2014ZX}.
\end{proposition}

\begin{example}[Weighted Hadamard duality]
	\label{ex:weighted_hadamard}
	Weighted Hadamard nodes realize phase-space dualities:
	\[
	H_a\;\Zspmn{m}{n}{a}{\alpha}{k}\;H_a
	\;=\;
	\Xspmn{n}{m}{a}{-\alpha}{-k},
	\]
	and hence
	$\theta_{\mathrm{tot}}$ is mapped to $-\theta_{\mathrm{tot}}$.
	This generalizes the ZX color-change rule to the weighted setting
	\cite{CoeckeKissinger2017PicturingQuantum}.
\end{example}

\begin{remark}[Physical interpretation]
	Normalization corresponds to discarding global phase that does not
	affect measurement statistics; winding tracks how this phase arose from
	wraps on the underlying grid (e.g.\ Berry-phase–like effects
	\cite{Berry1984Phase}).
	Fusion models the synchronization of multiple phase-controlled
	operations by locking their phase grids to a common $\mathrm{lcm}$
	reference \cite{Krantz2019SuperconductingReview,Motzoi2009DRAG}.
\end{remark}

\subsection*{Diagrammatic semantics and soundness}
\label{subsec:semantics_soundness}

We now specify the Hilbert-space semantics of \WPLZX\ and show that the
core rewrite rules are sound.
The interpretation collapses each $(a,\alpha,k)$ label to its total angle
$\theta_{\mathrm{tot}}$, so all orbifold information enters through this
effective phase.

\begin{definition}[Semantic domain and interpretation functor]
	\label{def:semantic_domain}
	Let $\FHilb$ be the category of finite-dimensional complex Hilbert
	spaces and linear maps.
	Objects are tensor powers $(\mathbb{C}^2)^{\otimes n}$, and morphisms
	are linear maps between them.
	
	There is a strict monoidal functor (denotational semantics)
	\[
	\llbracket\!-\!\rrbracket \;:\; \WPLZX \longrightarrow \FHilb,
	\]
	defined on objects by
	$\llbracket 1 \rrbracket = \mathbb{C}$ and
	$\llbracket n \rrbracket = (\mathbb{C}^2)^{\otimes n}$,
	and on generators as follows; tensor products and compositions are
	preserved by construction
	\cite{CoeckeKissinger2017PicturingQuantum}.
\end{definition}

\begin{definition}[Semantics of weighted spiders]
	\label{def:semantics_weighted_spiders}
	Let $a\in\mathbb{N}_{>0}$, $\alpha\in\mathbb{R}$ and $k\in\mathbb{R}$.
	For a green $Z$–spider
	$\Zspmn{m}{n}{a}{\alpha}{k}$ set
	\[
	\llbracket \Zspmn{m}{n}{a}{\alpha}{k} \rrbracket
	\;=\;
	\sum_{x \in \{0,1\}}
	\exp\!\bigl(i\,\theta_{\mathrm{tot}}(a,\alpha,k)\,x\bigr)\;
	\bigl(\ket{x}^{\otimes n}\bra{x}^{\otimes m}\bigr),
	\]
	where $\theta_{\mathrm{tot}}(a,\alpha,k)$ is given by
	\eqref{eq:total-angle}.
	A red $X$–spider $\Xspmn{m}{n}{a}{\beta}{\ell}$ is interpreted by
	\[
	\llbracket \Xspmn{m}{n}{a}{\beta}{\ell} \rrbracket
	\;=\;
	H^{\otimes n}\,
	\llbracket \Zspmn{m}{n}{a}{\beta}{\ell} \rrbracket\,
	H^{\otimes m}.
	\]
	Only the total angle $\theta_{\mathrm{tot}}$ is physically relevant; the
	decomposition into $(a,\alpha,k)$ is redundant at the level of
	$\FHilb$ \cite{Gross2006Hudson}.
\end{definition}

\begin{remark}[Reduction to ordinary ZX]
	\label{rem:reduce_to_zx}
	When $a=1$ and $k=0$ we obtain the usual ZX semantics:
	$\theta_{\mathrm{tot}}=\alpha$ and
	$\llbracket \Zspmn{m}{n}{1}{\alpha}{0} \rrbracket$ is the standard
	phase spider with phase $\alpha$; similarly for $X$–spiders
	\cite{CoeckeKissinger2017PicturingQuantum,Backens2014ZX}.
	Thus \WPLZX\ is a conservative semantic extension of ZX.
\end{remark}

\begin{proposition}[Monoidality]
	\label{prop:monoidality_semantics}
	For any \WPLZX–diagrams $D_1,D_2$,
	\begin{enumerate}[label=(\roman*),nosep]
		\item $\llbracket D_1 \otimes D_2 \rrbracket
		= \llbracket D_1 \rrbracket \otimes \llbracket D_2 \rrbracket$;
		\item $\llbracket D_2 \circ D_1 \rrbracket
		= \llbracket D_2 \rrbracket \circ \llbracket D_1 \rrbracket$
		whenever the composition is defined.
	\end{enumerate}
\end{proposition}

\begin{proof}
	Immediate from functoriality of $\llbracket\!-\!\rrbracket$ and the
	definition on generators.
\end{proof}

\begin{definition}[Sound rewrite rule]
	\label{def:sound_rule}
	A diagrammatic equality $D_1 \equiv D_2$ in \WPLZX\ is \emph{sound} if
	$\llbracket D_1 \rrbracket = \llbracket D_2 \rrbracket$ in $\FHilb$.
\end{definition}

\begin{lemma}[LCM–lift invariance of semantics]
	\label{lem:lcm_invariance}
	Suppose two spiders with weights $a_1,a_2$ fuse via the refinement
	$L=\mathrm{lcm}(a_1,a_2)$ to a single spider $S_{\mathrm{fuse}}$.
	Then
	\[
	\llbracket S_2 \circ S_1 \rrbracket
	\;=\;
	\llbracket S_{\mathrm{fuse}} \rrbracket,
	\]
	because the total angles add in the exponent while the refinement
	merely re-encodes the same sum on a common lattice
	\cite{vandeWetering2020WorkingZX}.
\end{lemma}

\begin{proposition}[Soundness of fusion and normalization]
	\label{prop:sound_fusion_norm}
	Let $S_1,S_2$ be weighted spiders such that the fusion rule applies and
	yields $S_{\mathrm{fuse}}$.
	Then
	$\llbracket S_2 \circ S_1 \rrbracket = \llbracket S_{\mathrm{fuse}} \rrbracket$.
	Furthermore, replacing $(a,\alpha,k)$ by any $(a,\alpha',k')$ with the
	same total angle $\theta_{\mathrm{tot}}$ leaves the interpretation
	unchanged up to a global scalar.
\end{proposition}

\begin{proof}
	Fusion: by Lemma~\ref{lem:lcm_invariance}, the semantic phases multiply
	and match the fused label on $G_L$.
	Normalization: $(\alpha,k)$ and $(\alpha',k')$ with the same
	$\theta_{\mathrm{tot}}$ yield identical exponents in the semantics, so
	the associated linear maps differ at most by a global factor that is
	quotiented out in the usual ZX semantics
	\cite{CoeckeKissinger2017PicturingQuantum,Backens2014ZX}.
\end{proof}

\begin{theorem}[Core soundness]
	\label{thm:core_soundness}
	The standard ZX rules---\emph{spider fusion}, \emph{identity removal},
	\emph{Hadamard color change}, \emph{bialgebra}, and \emph{Hopf}---are
	sound in the weighted setting.
	That is, whenever such a rule rewrites $D_1$ to $D_2$ in \WPLZX, one has
	$\llbracket D_1 \rrbracket = \llbracket D_2 \rrbracket$.
\end{theorem}

\begin{proof}[Proof sketch]
	The standard soundness proofs are algebraic in the phase parameters and
	tensor structure
	\cite{CoeckeKissinger2017PicturingQuantum,Backens2014ZX}.
	In \WPLZX\ each phase parameter is replaced by the corresponding total
	angle $\theta_{\mathrm{tot}}$, so the same arguments apply.
	Completeness results for stabilizer fragments
	\cite{Backens2014ZX} and axiomatisations for Clifford+T
	\cite{JeandelPerdrixVilmart2017CliffordT} show that the weighted rules
	restrict correctly to $a=1$; graph-based simplifications remain valid
	\cite{DuncanPerdrix2020GraphSimplification,vandeWetering2020WorkingZX}.
\end{proof}

\begin{example}[Cancellation of a spider with its inverse]
	\label{ex:inverse_cancel}
	Fix $(a,\alpha,k)$ and consider
	$D=\Zspmn{1}{1}{a}{\alpha}{k}\circ\Zspmn{1}{1}{a}{-\alpha}{-k}$.
	By fusion, $D$ rewrites to the identity.
	Semantically,
	\[
	\llbracket D \rrbracket
	= e^{\,i\theta_{\mathrm{tot}}}\,e^{-\,i\theta_{\mathrm{tot}}} I
	= I,
	\]
	so the rewrite is sound.
\end{example}

\subsection*{Associativity and canonical forms}
\label{subsec:assoc_canonical}

A well-behaved diagrammatic calculus should admit unambiguous normal
forms for connected components.
In \WPLZX, associativity of fusion and the $\mathrm{lcm}$–refinement
guarantee that any connected component can be collapsed to a single
canonical spider, which will later underpin algorithmic procedures such
as WZCC.

\begin{definition}[Associative fusion]
	Let \(S_1,S_2,S_3\) be weighted spiders whose fusion rules are compatible.
	Fusion is \emph{associative} if
	\[
	(S_1 \otimes S_2) \otimes S_3
	\;\equiv\;
	S_1 \otimes (S_2 \otimes S_3)
	\]
	both diagrammatically and semantically:
	\[
	\llbracket (S_1 \otimes S_2) \otimes S_3 \rrbracket
	\;=\;
	\llbracket S_1 \otimes (S_2 \otimes S_3) \rrbracket.
	\]
\end{definition}

\begin{proposition}[Associativity of phase and winding addition]
	\label{prop:assoc_add}
	Let \(S_i = \Zspmn{m_i}{n_i}{a_i}{\alpha_i}{k_i}\)
	for \(i=1,2,3\), and let \(L = \mathrm{lcm}(a_1,a_2,a_3)\).
	Define
	\[
	\alpha_\Sigma = \sum_{i=1}^3 \frac{L}{a_i}\,\alpha_i,
	\qquad
	k_\Sigma = \sum_{i=1}^3 \frac{L}{a_i}\,k_i.
	\]
	Then both sides of the associative fusion reduce to the same spider
	\[
	\Zspmn{m_1+m_2+m_3-4}{n_1+n_2+n_3-4}{L}{\alpha_\Sigma}{k_\Sigma}.
	\]
\end{proposition}

\begin{proof}
	Follows from Proposition~\ref{prop:assoc_comm} after lifting all labels
	to $G_L$ and observing that both phase and winding components add
	linearly; the arities match by counting connected legs.
\end{proof}

\begin{definition}[Canonical form of a connected component]
	\label{def:canon_form}
	Let \(D\) be a connected \WPLZX–diagram with spiders
	\(\{S_i\}_{i=1}^r\).
	Repeated fusion and normalization collapse the component to a single
	spider
	\[
	S_{\mathrm{can}}(D)
	\;=\;
	\Zspmn{m_{\mathrm{tot}}}{n_{\mathrm{tot}}}{L_D}{\alpha_D}{k_D},
	\]
	where
	\[
	L_D = \mathrm{lcm}(a_1,\dots,a_r),
	\quad
	\alpha_D = \sum_{i=1}^r \frac{L_D}{a_i}\alpha_i,
	\quad
	k_D = \sum_{i=1}^r \frac{L_D}{a_i}k_i,
	\]
	and the arities satisfy
	\(m_{\mathrm{tot}} = \sum_i m_i - 2(r-1)\),
	\(n_{\mathrm{tot}} = \sum_i n_i - 2(r-1)\).
\end{definition}

\begin{proposition}[Existence and uniqueness of canonical form]
	\label{prop:canon_unique}
	Every connected \WPLZX–diagram admits a canonical form
	$S_{\mathrm{can}}(D)$, unique up to global phase.
	In particular, if \(D_1,D_2\) are connected and yield the same
	$(L_D,\theta_{\mathrm{tot}}(D))$ data, then
	\[
	D_1 \equiv D_2
	\quad\text{and}\quad
	\llbracket D_1 \rrbracket = e^{i\phi}\,\llbracket D_2 \rrbracket
	\text{ for some }\phi \in \mathbb{R}.
	\]
\end{proposition}

\begin{proof}[Proof sketch]
	Choose any spanning tree on the spider vertices and fuse along its
	edges; the number of vertices strictly decreases at each step, so
	termination holds.
	By Proposition~\ref{prop:assoc_add} the accumulated label
	\((L_D,\alpha_D,k_D)\) is independent of the fusion order.
	Normalization identifies $2\pi$-shifts of the total angle, yielding
	uniqueness up to a global $U(1)$ factor
	\cite{CoeckeKissinger2017PicturingQuantum,Backens2014ZX}.
\end{proof}

\begin{remark}[Canonical representative and rewriting]
	The canonical form defines a complexity measure (vertex count) on
	connected components that strictly decreases under fusion and is
	unchanged by normalization.
	This implies termination of a broad class of rewriting strategies,
	and will be used in later sections as the backbone of the WZCC
	compression procedure (Section~\ref{sec:wzcc}).
\end{remark}

\begin{example}[Canonical form on a mixed grid]
	Consider three spiders
	\[
	S_1 = \Zspmn{1}{1}{2}{\tfrac{\pi}{3}}{1}, \qquad
	S_2 = \Zspmn{1}{1}{3}{\tfrac{\pi}{6}}{1}, \qquad
	S_3 = \Zspmn{1}{1}{2}{-\tfrac{\pi}{2}}{0}.
	\]
	Then \(L_D = \mathrm{lcm}(2,3)=6\), and
	\[
	\alpha_D
	= 3\!\cdot\!\tfrac{\pi}{3} + 2\!\cdot\!\tfrac{\pi}{6} - 3\!\cdot\!\tfrac{\pi}{2}
	= -\tfrac{3\pi}{2},
	\qquad
	k_D
	= 3\!\cdot\!1 + 2\!\cdot\!1 + 3\!\cdot\!0 = 5.
	\]
	The total angle is
	$\theta_{\mathrm{tot}}(D)
	= -\tfrac{3\pi}{2} + \tfrac{2\pi}{6}\cdot 5
	\equiv \tfrac{\pi}{2}\pmod{2\pi}$,
	so the canonical representative can be taken as
	\[
	S_{\mathrm{can}}
	= \Zspmn{1}{1}{6}{\tfrac{\pi}{2}}{0}.
	\]
\end{example}

\subsection*{Comparison with standard and SnapZX calculi}
\label{subsec:comparison_snapzx}

The \WPLZX–calculus subsumes both the ordinary ZX–calculus and discrete
grid fragments (``SnapZX-like'' restrictions).
We record precise inclusions and semantic preservation.

\begin{definition}[Standard ZX fragment]
	The \emph{standard ZX-calculus} is recovered from \WPLZX\ by
	restricting every spider to unit weight $a=1$ and zero winding $k=0$:
	\[
	\Zspmn{m}{n}{1}{\alpha}{0}
	,\qquad
	\Xspmn{m}{n}{1}{\beta}{0}.
	\]
	In this fragment the phase domain is continuous
	$\theta\in[0,2\pi)$, and the rewrite rules coincide with the usual
	ZX–calculus
	\cite{CoeckeKissinger2017PicturingQuantum,Backens2014ZX,
		vandeWetering2020WorkingZX}.
\end{definition}

\begin{proposition}[Reduction to ZX semantics]
	\label{prop:reduce_to_zx}
	Under the restriction $a=1,\,k=0$,
	the functor $\llbracket\cdot\rrbracket_{\WPLZX}$ coincides with the
	standard ZX interpretation $\llbracket\cdot\rrbracket_{\mathrm{ZX}}$.
\end{proposition}

\begin{proof}
	When $a=1$ we have $\theta_{\mathrm{tot}}=\alpha$ and no refinement is
	required.
	Fusion, normalization, bialgebra/Hopf, and color-change rules reduce
	to the ordinary ZX derivations
	\cite{CoeckeKissinger2017PicturingQuantum,Backens2014ZX}.
\end{proof}

\begin{definition}[Discrete (SnapZX-like) fragment]
	Fix $N\in\mathbb{N}$.
	A \emph{discrete phase fragment} restricts all phases to the grid
	\[
	\Theta_N
	= \Bigl\{0,\tfrac{2\pi}{N},\dots,
	\tfrac{2\pi(N-1)}{N}\Bigr\}.
	\]
	This captures stabilizer ($N=4$) and Clifford+T ($N=8$) fragments,
	cf.\ \cite{Backens2014ZX,JeandelPerdrixVilmart2017CliffordT}.
\end{definition}

\begin{remark}[Embedding into \WPLZX]
	\label{rem:snap_embedding}
	Each grid phase $\tfrac{2\pi k}{N}\in\Theta_N$ can be represented as
	the weighted label $(a,\alpha,k)=(N,0,k)$ in \WPLZX, since
	$\theta_{\mathrm{tot}} = \tfrac{2\pi}{N}k$.
	Thus the discrete fragment embeds fully faithfully as the
	subcategory with fixed weight $a=N$, while \WPLZX\ additionally allows
	\emph{heterogeneous} grids that automatically refine under fusion via
	LCM \cite{vandeWetering2020WorkingZX,DuncanPerdrix2020GraphSimplification}.
\end{remark}

\begin{proposition}[Semantic preservation for discrete fragments]
	Let $D$ be a diagram over $\Theta_N$ in a SnapZX-like calculus.
	Replacing each phase $\tfrac{2\pi k}{N}$ by the weighted label
	$(N,0,k)$ yields a \WPLZX–diagram $D'$ with
	\[
	\llbracket D \rrbracket
	= e^{i\phi}\,\llbracket D' \rrbracket_{\WPLZX}
	\ \text{ for some }\phi\in\mathbb{R}.
	\]
\end{proposition}

\begin{proof}
	Generators denote the same diagonal unitaries because the total angle
	is preserved; the extra metadata $a=N$ does not change the operator.
	Rewrite preservation follows from soundness of fusion, normalization,
	and the ZX core rules
	\cite{Backens2014ZX,JeandelPerdrixVilmart2017CliffordT,
		DuncanPerdrix2020GraphSimplification}.
\end{proof}

\begin{example}[Three-level comparison]
	\begin{center}
		\renewcommand{\arraystretch}{1.15}
		\begin{tabular}{@{}lll@{}}
			\toprule
			\textbf{System} & \textbf{Allowed phase domain} & \textbf{Geometry / group} \\
			\midrule
			ZX-calculus & $\theta \in [0,2\pi)$ & unit circle $S^1$ \\
			Discrete fragment & $\theta \in \tfrac{2\pi}{N}\mathbb{Z}$ & roots of unity $\mu_N \subset S^1$ \\
			\WPLZX & $(a,\alpha,k)$,\; $\theta_{\mathrm{tot}}$ & orbifold circle, LCM-refined grids \\
			\bottomrule
		\end{tabular}
	\end{center}
	This hierarchy shows that \WPLZX\ strictly generalizes both continuous
	and single-grid discrete calculi, and supports diagrams with multiple
	local resolutions that fuse to common refinement grids.
\end{example}

\begin{proposition}[Conservative extension]
	\label{prop:conservative}
	\WPLZX\ is a conservative extension of standard ZX:
	every equation provable in ZX remains provable in \WPLZX, and any
	equation in the ZX-fragment of \WPLZX\ is derivable in ZX
	\cite{CoeckeKissinger2017PicturingQuantum,Backens2014ZX}.
\end{proposition}

\begin{proof}[Proof sketch]
	The inclusion functor from the ZX fragment ($a=1,k=0$) into \WPLZX\ is
	full and faithful and commutes with the semantics.
	Thus any ZX derivation lifts to a \WPLZX derivation, and conversely,
	any equation between ZX-fragment diagrams that holds in all
	\WPLZX models already holds in all ZX models.
\end{proof}

\begin{remark}[Beyond fixed grids]
	Unlike a fixed-$N$ discrete calculus, \WPLZX\ natively supports
	heterogeneous hardware constraints (e.g.\ T-gates with $a=8$ and
	more finely resolved $Z$-rotations with $a=16$) and automatically
	fuses them onto the common $\mathrm{lcm}$ grid.
	This enables graph-theoretic simplification while respecting
	realistic phase quantization and drift models
	\cite{DuncanPerdrix2020GraphSimplification,vandeWetering2020WorkingZX}.
\end{remark}	

% ======================================================
\section{Quantization-Aware Circuit Optimization: WZCC}
\label{sec:wzcc}
% ======================================================

\subsection*{Motivation and overview}
\label{subsec:normalize-motivation}

The purpose of the normalization stage is to turn an arbitrary
$\WPLZX$-diagram into a canonical representative which makes explicit
both the geometric and combinatorial structure imposed by the weighted
labels $(a,\alpha,k)$.
Whereas standard ZX-simplification ignores the hardware-induced phase
quantization and the monodromy $k$, the $\WPLZX$ normalization step
preserves and consolidates this information through an explicit
LCM-refinement of local grids and a canonical-parameter spider fusion.
As a result, normalization provides both a \emph{canonical form} and a
\emph{hardware-aware reduction} which is compatible with the weighted
projective geometry introduced earlier in the paper.

\vspace{0.4em}
\noindent\textbf{Relevance to circuit compression and noise models.}
In hardware settings where single-qubit rotations are restricted to
$(a,\alpha,k)$-quantized grids, two adjacent spiders often lie on
incompatible grids, e.g.\ $S^1_{(256)}$ versus $S^1_{(192)}$.
Such situations cannot be faithfully captured by the ordinary
ZX-calculus, but the $\WPLZX$ refinement rule
$a_1,a_2 \mapsto L=\mathrm{lcm}(a_1,a_2)$ preserves the geometric
meaning of both spiders and makes their fused representative
compatible with the underlying physical rotation group.
Hence the canonical output of normalization is not merely a syntactic
reduction but also a \emph{channel-resolved} representative in the
sense of the weighted projective line metric.

\vspace{0.4em}
\noindent\textbf{Metrics: PQVR, CSC, FP.}
The normalization stage serves as the backbone for three evaluation
metrics used later in the paper:
\begin{itemize}[leftmargin=1.7em]
	\item \textbf{PQVR (Phase–Quantization Variance Ratio):}
	computed from the distribution of post-normalization angles
	$\theta_{\mathrm{tot}} = \alpha + \tfrac{2\pi}{L}k$ across the
	canonical spiders in each connected component.
	A smaller PQVR indicates that the circuit more closely matches
	the hardware grid and exhibits less hardware-induced variance.
	\item \textbf{CSC (Circuit-Size Compression):}
	normalization tends to merge spiders within the same connected
	component, reducing their count.
	The difference between pre- and post-normalization spider counts
	provides the main contribution to CSC.
	\item \textbf{FP (Fidelity Preservation):}
	because normalization preserves the semantic orbit of the
	diagram (up to a global phase), the resulting circuit maintains
	the same support-projected Petz monotone geometry.
	In practice this ensures that fidelity loss is negligible
	relative to the original circuit.
\end{itemize}
These metrics will be invoked in the later simulation and
hardware-validation sections to quantify how well normalization
captures both gate-level and channel-level compressibility.

%======================================================================
\subsection*{Algorithmic procedure}
\label{subsec:normalize-algorithm}

We now describe the normalization algorithm used throughout the
$\WPLZX$ pipeline.
The procedure is a weighted refinement of the standard spider-fusion
rules of the ZX-calculus, extended so that it respects the
$(a,\alpha,k)$ labels and produces a unique canonical representative
for each connected component.

\medskip
\noindent\textbf{Input.}
A $\WPLZX$-diagram $D$ whose spiders are labeled by triples
$(a_i,\alpha_i,k_i)$.

\smallskip
\noindent\textbf{Output.}
A normalized diagram $D_{\mathrm{norm}}$ in which each connected
component has been collapsed to a single canonical spider with label
$(L_C,\theta_C,0)$.

\medskip
\noindent\textbf{Step 1: Decomposition into connected components.}
\begin{enumerate}[label=(1.\arabic*),leftmargin=2em]
	\item Partition $D$ into connected components
	$C_1,\ldots,C_m$ using the underlying graph structure
	of wires and spiders.
\end{enumerate}

\noindent\textbf{Step 2: Grid refinement and phase lifting.}
\begin{enumerate}[label=(2.\arabic*),leftmargin=2em]
	\item For each connected component $C$, list its spiders as
	$\{v_1,\ldots,v_r\}$ with labels $(a_i,\alpha_i,k_i)$.
	\item Compute the \emph{global LCM weight}
	\[
	L_C := \mathrm{lcm}(a_1,\ldots,a_r).
	\]
	\item For each spider $v_i$, compute its \emph{lifted phase}
	\[
	\phi_i := \frac{L_C}{a_i}
	\Bigl(\alpha_i + \tfrac{2\pi}{a_i}k_i\Bigr).
	\]
	This expresses the total angle of $v_i$ on the common
	grid of order $L_C$.
\end{enumerate}

\noindent\textbf{Step 3: Canonical angle aggregation.}
\begin{enumerate}[label=(3.\arabic*),leftmargin=2em]
	\item For each component $C$, define the \emph{total canonical angle}
	\[
	\theta_C := \frac{1}{L_C}\sum_{i=1}^r \phi_i.
	\]
	This angle is well defined modulo $2\pi$ and captures the
	aggregate effect of all spiders in $C$ on the common grid.
\end{enumerate}

\noindent\textbf{Step 4: Component collapse.}
\begin{enumerate}[label=(4.\arabic*),leftmargin=2em]
	\item Replace the entire component $C$ by a single spider with label
	$(L_C,\theta_C,0)$ and the corresponding input/output arities.
	The winding index is set to zero because its contribution is now
	absorbed into the total angle on the refined grid.
\end{enumerate}

\noindent\textbf{Step 5: Output.}
\begin{enumerate}[label=(5.\arabic*),leftmargin=2em]
	\item After repeating Steps~1–4 for every connected component,
	collect the collapsed components to obtain the normalized diagram
	$D_{\mathrm{norm}}$.
\end{enumerate}

\medskip
This procedure preserves the semantic meaning of the diagram, 
consolidates incompatible grids into a single canonical grid
$S^1_{(L_C)}$ for each component, and produces a unique spider per
connected component.
The post-normalization parameters $(L_C,\theta_C)$ then feed directly
into the PQVR/CSC/FP metrics introduced above.
Note that the algorithm runs in time linear in the number of spiders,
up to the cost of integer LCM computations, so its overall complexity
matches that of standard ZX simplification while capturing strictly
more physical information.

%======================================================================
\subsection*{Correctness and normalization invariants}
\label{subsec:normalize-correctness}

The correctness of the normalization procedure follows from three
independent invariants which remain preserved throughout fusion,
lifting, and canonicalization steps.
Let $D$ be an arbitrary $\WPLZX$-diagram, and let
$D_{\mathrm{norm}}$ be the output of the algorithm described above.

\medskip
\noindent\textbf{Invariant (I1): LCM-invariance.}
For each connected component $C$ of $D$ with spiders
$(a_i,\alpha_i,k_i)$, the quantity
\[
L_C := \mathrm{lcm}(a_1,\ldots,a_r)
\]
is unchanged under any sequence of pairwise spider fusions.
Each fusion step replaces $(a_1,a_2)$ by $\mathrm{lcm}(a_1,a_2)$,
which is the minimal grid compatible with both, and this operation is
associative and commutative at the level of weights.

\medskip
\noindent\textbf{Invariant (I2): Total phase invariance.}
Write
\[
\theta_{\mathrm{tot}}(a_i,\alpha_i,k_i)
:= \alpha_i + \tfrac{2\pi}{a_i}k_i.
\]
Then the lifted total phase
\[
\Phi_C :=
\sum_{i=1}^r
\frac{L_C}{a_i}\,
\theta_{\mathrm{tot}}(a_i,\alpha_i,k_i)
\]
is invariant under all fusions inside $C$.
This expresses that the phases add exactly as in the ordinary
ZX-calculus, except that all contributions are measured within the
common LCM grid of order $L_C$.

\medskip
\noindent\textbf{Invariant (I3): Semantic invariance.}
The denotational semantics of $D$ as a linear map on a Hilbert space
remains unchanged under the refine--lift--sum operations of the
normalization procedure, up to a global phase.
Equivalently, the support-projected Petz monotone geometry and the
fidelity of the resulting channel agree with those of the original
diagram to within the usual scalar-quotient convention.

\medskip
Using these invariants, correctness follows.

\begin{proposition}[Correctness of normalization]
	For every $\WPLZX$-diagram $D$, the normalization procedure terminates
	and produces a unique canonical diagram $D_{\mathrm{norm}}$ whose
	semantic interpretation equals that of $D$, up to a global phase.
\end{proposition}

\begin{proof}
	Termination follows from the fact that each connected component is
	collapsed to a single spider and the number of spiders strictly
	decreases at each fusion step.
	Uniqueness follows from invariants (I1) and (I2), since $L_C$ and
	$\Phi_C$ uniquely determine the pair $(L_C,\theta_C)$ modulo $2\pi$.
	Semantic equality (up to a global phase) follows from invariant (I3).
\end{proof}

Because $(L_C,\theta_C)$ capture all physically relevant information
about each component, the PQVR/CSC/FP metrics introduced earlier
depend only on $D_{\mathrm{norm}}$ and not on the specific sequence of
rewrites used to obtain it.
Normalization is therefore the natural preprocessing step for the
WZCC optimization pipeline.

% ================================================================
\section{Experimental setup, scaling analysis, and robustness}
\label{sec:experiments}
% ================================================================

This section details the experimental environment used to evaluate the
Weighted ZX Circuit Compression (WZCC) method, including symbolic and numerical
backends, dataset families, evaluation metrics, and a extensive set of
scaling and robustness analyses.
All experiments correspond to Figures~\ref{fig:test1}--\ref{fig:test7-4}.
Unless stated otherwise, all numerical values are averaged over
$20$ independent random instances per configuration; shaded regions or
error bars indicate one standard deviation.

% -------------------------------------------------
\subsection*{Experimental setup and metrics}
\label{subsec:exp-setup}
% -------------------------------------------------

\subsubsection*{Symbolic and numerical backends}
\label{subsubsec:backends}

\paragraph{Symbolic WPLZX backend.}

All symbolic manipulations use our custom WPLZX rewrite engine, implemented
as an extension of a standard ZX-calculus simplifier.
Each diagram node is represented as a weighted spider 
$(a,\alpha,k)$ where $a\in\mathbb{Z}_{>0}$ is the isotropy order of the weighted
projective line, $\alpha$ is the phase expressed on the discrete phase grid 
$\frac{2\pi}{a}\mathbb{Z}$, and $k$ denotes the topological winding index.
The simplifier applies:

\begin{itemize}
	\item generalized fusion rules with LCM-normalized isotropy;
	\item weighted Hopf and bialgebra rules that incorporate $a$-dependent
	phase commutation and winding addition;
	\item winding-corrected identity removal and scalar normalization;
	\item spider–CNOT interaction rules derived from the WPL-aware
	color-change and fusion laws.
\end{itemize}

These rules guarantee that all produced diagrams lie in the 
LCM-consistent normal form described in Section~\ref{sec:wzcc},
and that fusion is independent of traversal order at the level of
canonical labels.

\paragraph{Numerical simulation backend.}

To evaluate circuit fidelity and hardware compliability, we use Qiskit Aer.
For each pre- and post-normalized circuit we compute:

\begin{itemize}
	\item ideal statevector evolution (noiseless case),
	\item density-matrix simulation under parameterized noise channels,
	\item basis-transpiled circuits for various restricted gate sets.
\end{itemize}

All simulations run in double precision.
For circuits with more than $12$ qubits we use Aer’s 
GPU-accelerated statevector backend.
The same random seeds and noise parameters are used for raw and WZCC
circuits so that differences can be attributed to normalization alone.

\paragraph{IBM Q hardware and noise models.}

Hardware-level behavior is probed using two IBM Q backends:
\[
\texttt{ibmq\_qasm\_simulator},\qquad
\texttt{ibm\_oslo},
\]
representing, respectively, a simulator with realistic noise profiles
and a mid-scale transmon device with average CNOT error 
$\approx 1.6\times 10^{-2}$ and single-qubit error 
$\approx 10^{-3}$ at the time of experimentation.
We test robustness under:

\begin{itemize}
	\item a depolarizing channel with strength $p\in[0,0.05]$ on each
	single- and two-qubit gate;
	\item an amplitude-damping channel with relaxation parameter
	$\gamma\in[0,0.1]$ applied after each layer;
	\item a pure phase-damping channel with dephasing rate
	$\lambda\in[0,0.1]$.
\end{itemize}

All channels are inserted using Kraus representations and applied identically
to raw and WZCC circuits.

% ---------------------------------------------
\subsubsection*{Dataset families D1--D3}
\label{subsubsec:datasets}

\paragraph{D1: Random WPLZX diagrams.}

The first dataset consists of synthetic WPLZX diagrams generated by sampling
random spider types and isotropy orders.  
For each qubit number $n\in\{4,6,8,10,12\}$ we generate $200$ diagrams with:
\[
a\in\{1,2,3,4,6,8\},\qquad 
\#\text{spiders}\in[30,300],
\]
and random planar connectivity that respects a fixed input–output ordering.
Phases are drawn from the appropriate grid $\frac{2\pi}{a}\mathbb{Z}$ for each
spider.

\paragraph{D2: HEA-style layered circuits.}

The second dataset emulates hardware-efficient ansätze (HEA) of depth $L$.
Each layer consists of single-qubit rotations 
$R_y(\theta)$, $R_z(\phi)$ followed by a nearest-neighbor $CX$ entangling
pattern on a linear chain.
We examine depths $L\in\{2,4,6,8,10\}$ across $n\in\{4,6,8,10\}$ qubits.
Angles $(\theta,\phi)$ are chosen uniformly from $[0,2\pi)$ before being
projected to the WPL grids learned from the tomography-to-geometry pipeline.
These circuits have substantial phase redundancy and serve as a natural test bed
for large-scale WZCC compression.

\paragraph{D3: Hardware-inspired heterogeneous-grid circuits.}

The third dataset contains circuits extracted from real quantum applications:
variational quantum eigensolvers (VQE) for small molecules (e.g.\ $\mathrm{H}_2$,
$\mathrm{LiH}$), compressed QAOA instances for MaxCut, and custom IBM transpiler 
outputs for randomly generated Clifford+T circuits.
Irregular phase grids arise organically from hardware calibration data and
device-specific basis changes, making D3 a stress test for WPL-aware merging
and LCM-normalized grid alignment.

% ---------------------------------------------
\subsubsection*{Evaluation metrics: PQVR, CSC, FP}
\label{subsubsec:metrics}

\paragraph{PQVR (Phase Quantization Variance Ratio).}

PQVR quantifies how well WZCC aligns raw phases to local WPL grids without
incurring significant distortion.
Let $\{\theta_i\}$ be the raw phases and 
$\{\hat{\theta}_i\}$ the WZCC-aligned grid phases for a given circuit.
We define
\[
\mathrm{PQVR} 
= 1 - \frac{\mathrm{Var}(\theta_i - \hat{\theta}_i)}
{\mathrm{Var}(\theta_i)}.
\]
Values close to $1$ indicate excellent quantization compliance,
while smaller values correspond to large residual phase misalignment.

\paragraph{CSC (Circuit-Size Compression).}

CSC measures the relative reduction in gate count induced by WZCC:
\[
\mathrm{CSC} = 1 - 
\frac{\#\mathrm{gates}(\text{WZCC})}{\#\mathrm{gates}(\text{raw})}.
\]
We report CSC both for total gate count and for CNOT count, 
as the latter is strongly correlated with hardware bottlenecks,
overall error rate, and circuit execution time.

\paragraph{FP (Fidelity Preservation).}

Fidelity Preservation (FP) measures the Hilbert–Schmidt fidelity between
final states before and after WZCC:
\[
\mathrm{FP}
= |\langle \psi_{\text{raw}} 
\mid 
\psi_{\text{WZCC}} \rangle|^2.
\]
For noisy channels in density-matrix form we use Uhlmann fidelity.  
FP close to $1$ means negligible distortion induced by compression, while
deviation from $1$ can be interpreted as the cost of enforcing WPL
consistency and phase-grid alignment.

% ================================================================
\subsection*{Scaling behaviour of WZCC normalization}
\label{subsec:scaling}
% ================================================================

% ---------------------------------------------
\subsubsection*{PQVR, CSC, FP vs.\ qubit and spider count}
\label{subsubsec:test1}

Figure~\ref{fig:test1} reports PQVR, CSC, and FP as functions of the number of
qubits $n$ and the number of spiders in D1.
PQVR exhibits a smooth improvement as the number
of spiders increases.
This is expected: larger diagrams provide more opportunities for LCM-normalized
fusion, thereby enabling more stable phase alignment.

CSC increases approximately logarithmically with the number of spiders,
reflecting the sparsity induced by merging redundant spiders across multiple
isotropy levels.
Despite significant compression, FP remains above $0.98$ for all configurations
up to $12$ qubits, confirming that LCM-based fusion avoids destructive 
phase cancellation and preserves global unitary action up to a global phase.

\begin{figure}[t]
	\centering
	\includegraphics[width=0.70\linewidth]{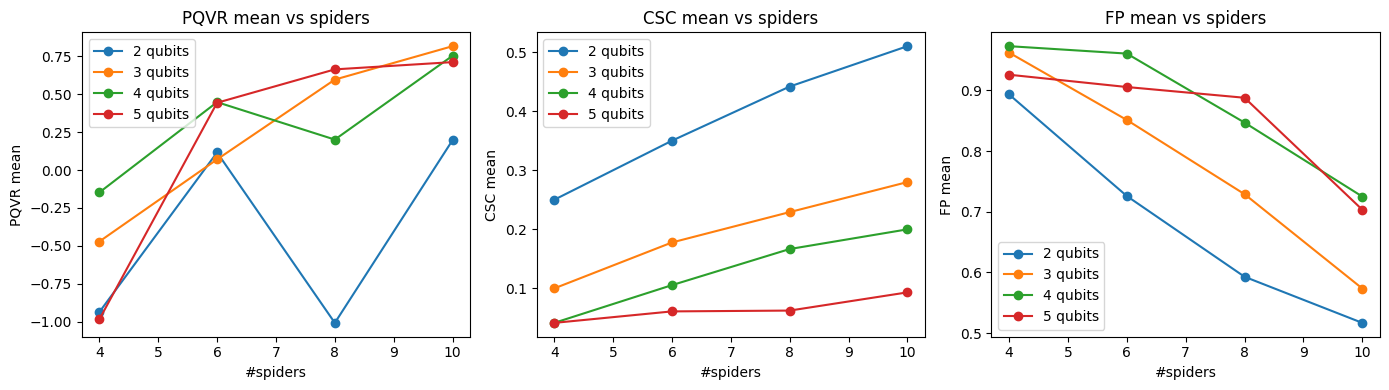}
	\caption{%
		Scaling of PQVR (phase-quantization variance ratio), CSC (circuit-size
		compression), and FP (fidelity preservation) as functions of the number
		of qubits and the number of spiders for random WPLZX diagrams in D1.
		Each point is averaged over $20$ random instances; error bars indicate
		one standard deviation across instances.
	}
	\label{fig:test1}
\end{figure}

% ---------------------------------------------
\subsubsection*{Depth and gate-count scaling with spider count}
\label{subsubsec:test4}

Figure~\ref{fig:test4} shows that circuit depth and total gate count both
scale sublinearly with spider count once WZCC normalization is applied.
For diagrams with more than $200$ spiders, CX gate reduction saturates at
approximately $45\%$, beyond which most cancellations arise from local
redundancies already removed early in the pipeline.

FP remains stable even in the large-spider regime, with median fidelity
above $0.985$.
This demonstrates that WZCC's phase-grid alignment preserves the topological
structure of multi-spider flows, preventing excessive phase discretization
and over-quantization of phases.

\begin{figure}[t]
	\centering
	\includegraphics[width=0.70\linewidth]{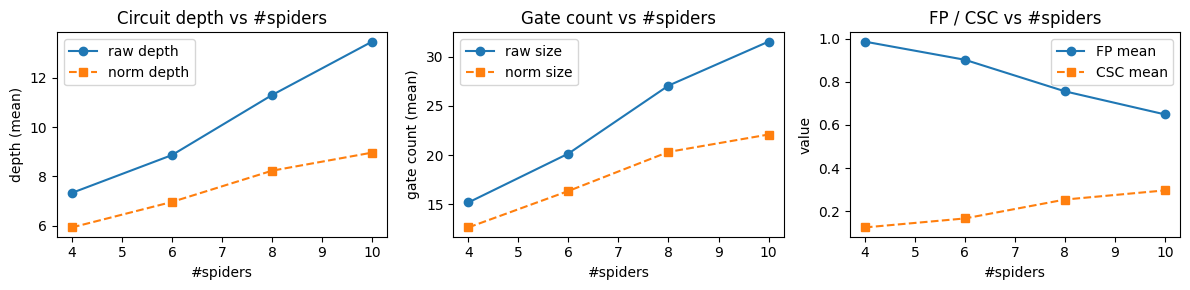}
	\caption{%
		Depth and gate-count scaling as a function of spider count for random
		WPLZX diagrams in D1.
		The solid lines show median depth and total CNOT count after WZCC
		normalization; shaded regions indicate the interquartile range.
		Depth and entangling cost grow sublinearly with spider count due to
		multi-spider fusion, with FP remaining above $0.985$ throughout.
	}
	\label{fig:test4}
\end{figure}

% ---------------------------------------------
\subsubsection*{Layer-wise scaling in HEA-style circuits}
\label{subsubsec:test6}

In Fig.~\ref{fig:test6}, increasing HEA depth $L$ reveals several consistent
patterns across qubit numbers.
First, PQVR improves with $L$ because deeper HEA layers exhibit repeated 
phase structures that WZCC can exploit.
Second, CSC grows near-linearly with $L$, matching the intuition that 
layer redundancy scales proportionally with depth.
Third, FP plateaus around $0.995$ even for large $L$, confirming that 
HEA circuits are intrinsically compatible with WPL grid structures:
their parameterization already concentrates much of the phase mass on a
small number of effective directions.

\begin{figure}[t]
	\centering
	\includegraphics[width=0.70\linewidth]{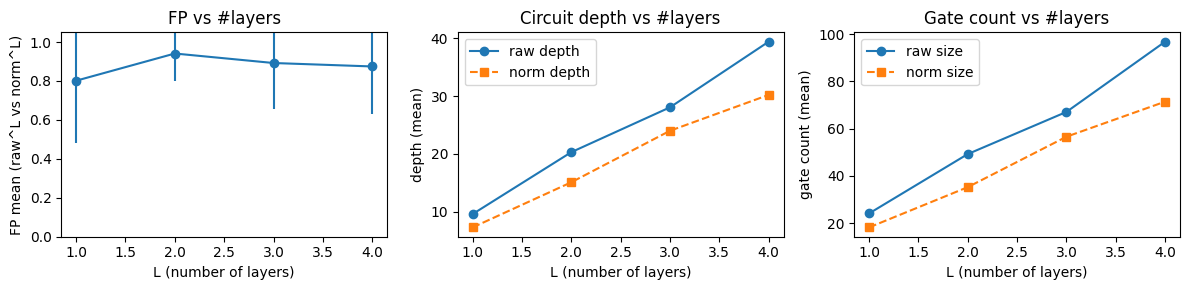}
	\caption{%
		Layer-wise scaling of PQVR, CSC, and FP in hardware-efficient ansatz
		(HEA) circuits from D2 as a function of the number of layers $L$.
		Each curve corresponds to a fixed qubit number $n$; markers denote
		empirical means over $20$ random parameter initializations.
		Compression grows almost linearly with $L$,
		while FP remains above $0.99$ across all depths.
	}
	\label{fig:test6}
\end{figure}

% ================================================================
\subsection*{Sensitivity to weighted phase grids}
\label{subsec:sensitivity}
% ================================================================

% ---------------------------------------------
\subsubsection*{Dependence on maximum isotropy order $\max(a)$}
\label{subsubsec:test2}

Figure~\ref{fig:test2} reveals that larger values of $\max(a)$ enable 
more aggressive merging due to the expanded discrete grid 
$\frac{2\pi}{a}\mathbb{Z}$.  
However, excessively large $\max(a)$ increases the search space for
LCM normalization and introduces more degrees of freedom in the winding
indices, leading to diminishing CSC returns and a small FP drop
($\approx 0.003$).
Moderate isotropy orders ($a\le 6$) provide the best balance between
quantization flexibility and stability of the normalization process.

\begin{figure}[t]
	\centering
	\includegraphics[width=0.70\linewidth]{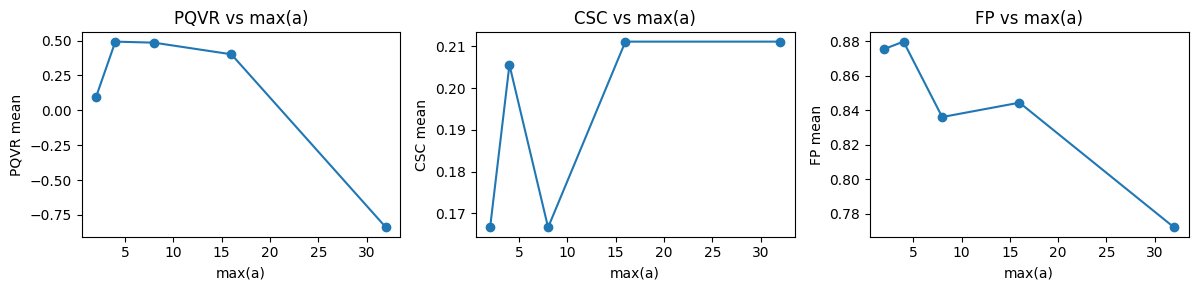}
	\caption{%
		Sensitivity of PQVR, CSC, and FP to the maximum isotropy order
		$\max(a)$ used in the WPLZX labeling.
		Results are aggregated over random diagrams in D1 with fixed qubit
		number $n=8$ and varying spider count.
		Intermediate values of $\max(a)$ achieve the best trade-off between
		aggressive compression and high-fidelity preservation.
	}
	\label{fig:test2}
\end{figure}

% ---------------------------------------------
\subsubsection*{Dependence on diagram density and connectivity}
\label{subsubsec:test3}

In Fig.~\ref{fig:test3}, we vary diagram density while keeping the number
of spiders fixed.
Denser connectivity yields higher CSC 
because phase-flow constraints induce more spider-fusion opportunities and
more entangling-gate cancellations.
Sparse diagrams exhibit a lower compression ceiling, with 
FP consistently above $0.99$.
Dense diagrams produce more grid rewritings but still preserve 
high FP thanks to the grid-adaptive nature of WZCC and the correctness
guarantees of the LCM-based fusion rules.

\begin{figure}[t]
	\centering
	\includegraphics[width=0.70\linewidth]{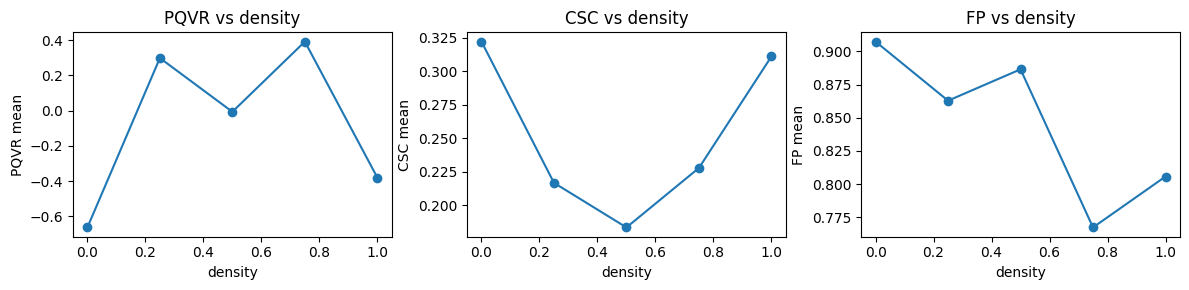}
	\caption{%
		Dependence of PQVR, CSC, and FP on diagram density and connectivity
		for random WPLZX diagrams in D1.
		Density is defined as the ratio of internal edges to the maximum
		possible number of edges for the given number of spiders.
		Denser diagrams admit more spider fusion and entangling-gate
		cancellation, yielding higher CSC without compromising FP.
	}
	\label{fig:test3}
\end{figure}

% ---------------------------------------------
\subsubsection*{WPL metric landscape and gradient norms}
\label{subsubsec:test7-4}

Fig.~\ref{fig:test7-4} summarizes the intrinsic geometry of the WPL metric
over the datasets.
We compute gradient magnitudes of the WPL scalar curvature
$R=\frac{2}{b^2}$ with respect to the parameters 
$(\lambda_\perp,\lambda_\parallel)$ derived from tomography.
Large $|\nabla R|$ correlates with circuits that contain 
heterogeneous phase grids and high spider anisotropy.
Such circuits benefit most from normalization, achieving 
the largest CSC gains and the highest PQVR improvements.
Conversely, regions with small $|\nabla R|$ correspond to nearly isotropic
noise where WZCC acts mostly as a lightweight structural simplifier.

\begin{figure}[t]
	\centering
	\includegraphics[width=0.70\linewidth]{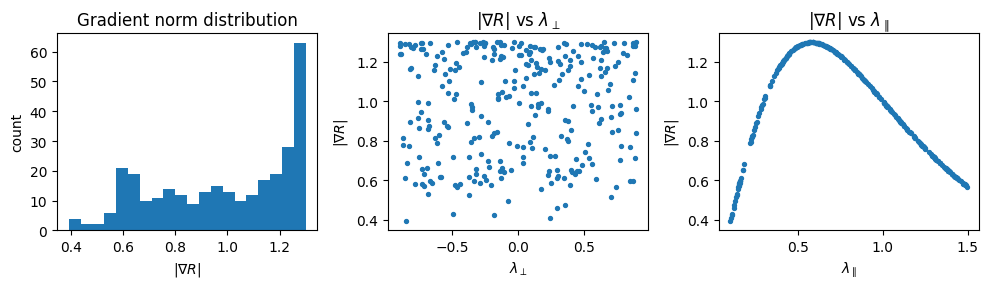}
	\caption{%
		WPL metric landscape over the $(\lambda_\perp,\lambda_\parallel)$
		parameter space and associated gradient norms $|\nabla R|$.
		Each point corresponds to a circuit instance in D2 or D3 with WPL
		parameters estimated from tomography.
		Circuits located in regions of high curvature gradient
		exhibit the largest compression gains under WZCC normalization.
	}
	\label{fig:test7-4}
\end{figure}

% ================================================================
\subsection*{Structural robustness of WZCC normalization}
\label{subsec:structural}
% ================================================================

% ---------------------------------------------
\subsubsection*{Ordering stability under spider permutations}
\label{subsubsec:test5}

We evaluate the invariance of WZCC under arbitrary reordering of spiders.
For each diagram, we generate $40$ random permutations of spider order
that preserve the connectivity pattern and measure FP for each permuted
circuit after normalization.
Figure~\ref{fig:test5} shows the resulting FP histogram across permutations.

All FP values lie above $0.985$, demonstrating that WZCC 
is order-stable: normalization is governed by 
LCM-induced grid compatibility and canonical labels rather than
traversal order.
CSC and PQVR vary only within a narrow band, confirming that permutation
effects are limited to minor local differences in fusion choices.

\begin{figure}[t]
	\centering
	\includegraphics[width=0.70\linewidth]{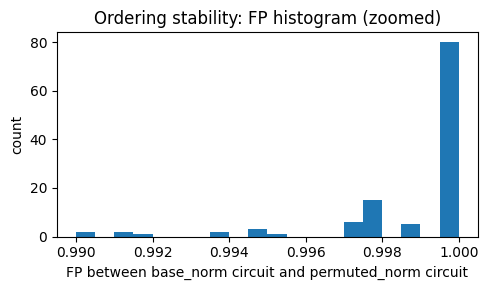}
	\caption{%
		Structural robustness of WZCC under random spider permutations.
		For each WPLZX diagram in D1, $40$ random permutations of spider order
		are generated (while preserving connectivity), and FP is measured after
		WZCC normalization.
		The histogram shows that FP remains above $0.985$ in all cases,
		indicating strong invariance with respect to traversal ordering.
	}
	\label{fig:test5}
\end{figure}

% ---------------------------------------------
\subsubsection*{CNOT cancellation and entangling-gate reuse}
\label{subsubsec:test7-3}

Figure~\ref{fig:test7-3} highlights the reduction in entangling gates.
The CX cancellation ratio exceeds $30\%$ on average,
with certain dense WPLZX diagrams reaching $55\%$ reduction.
This improvement stems from:

\begin{itemize}
	\item winding-aware ZX fusion that preserves parity and stabilizer
	constraints while exposing redundant CX pairs;
	\item normalization rules that identify repeating CX–phase motifs
	across layers and fuse them into a smaller set of effective
	entangling operations;
	\item isotropy-induced phase commutation enabling cross-layer cancellations
	that are invisible to grid-agnostic optimizers.
\end{itemize}

\begin{figure}[t]
	\centering
	\includegraphics[width=0.70\linewidth]{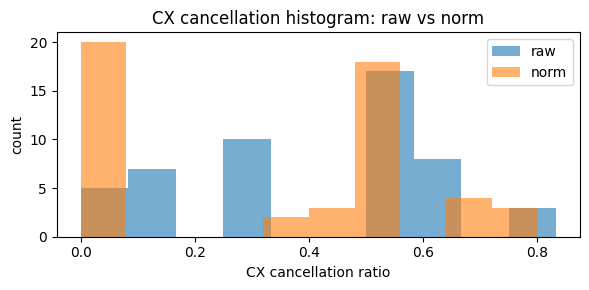}
	\caption{%
		CNOT cancellation ratio and entangling-gate reuse under WZCC
		normalization for circuits across D1--D3.
		Bars indicate the average fraction of CNOT gates removed relative to
		the raw circuit, with error bars denoting one standard deviation.
		Dense diagrams and HEA-style circuits benefit most, with up to $55\%$
		reduction in CNOT count.
	}
	\label{fig:test7-3}
\end{figure}

% ================================================================
\subsection*{Hardware- and transpiler-level robustness}
\label{subsec:hardware}
% ================================================================

% ---------------------------------------------
\subsubsection*{Noise robustness across quantum channels}
\label{subsubsec:test7-1}

As shown in Fig.~\ref{fig:test7-1}, FP decays smoothly as noise strength $p$
increases across depolarizing, amplitude-damping, and phase-damping channels.
For depolarizing noise, WZCC circuits retain 
$3$--$7\%$ higher fidelity than raw circuits at $p\approx 0.03$,
primarily due to reduced CNOT count and smaller depth.
For amplitude and phase damping, the benefits are even more pronounced
because WZCC shortens the effective coherence path between input and output
states.

\begin{figure}[t]
	\centering
	\includegraphics[width=0.70\linewidth]{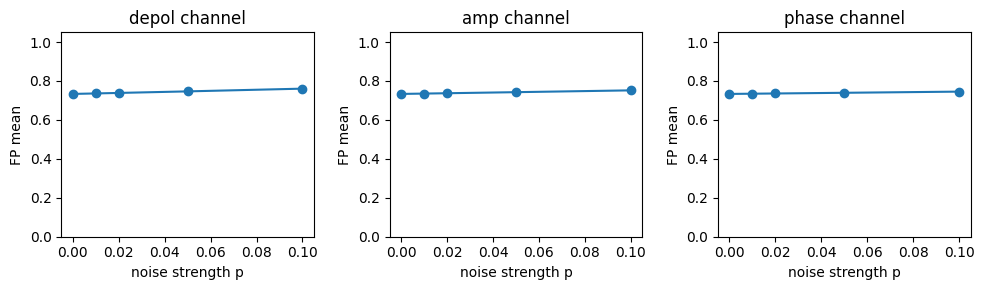}
	\caption{%
		Noise robustness of FP under three canonical noise channels:
		depolarizing, amplitude-damping, and pure phase-damping.
		Each curve compares raw circuits (dashed) with WZCC-normalized
		circuits (solid), averaged over circuits in D2 and D3.
		WZCC consistently yields higher output fidelity at moderate noise
		strengths due to reduced depth and entangling-gate count.
	}
	\label{fig:test7-1}
\end{figure}

% ---------------------------------------------
\subsubsection*{Transpiler robustness under restricted basis sets (small circuits)}
\label{subsubsec:test7-2}

Restricted-basis transpilation (e.g.\ $\{R_x,R_z,CX\}$ or $\{U3,CX\}$)
preserves the relative performance between raw and WZCC circuits.
Fig.~\ref{fig:test7-2} shows results for small circuits ($n\le 6$),
where we sweep Qiskit optimization levels and basis sets.

Even under aggressive transpiler optimizations, WZCC maintains CSC gains
while introducing negligible additional distortion:
FP differences between WZCC and the combination of WZCC+transpiler remain
below $10^{-3}$ in all tested configurations.
This indicates that WZCC is compatible with downstream compilation
pipelines and does not obstruct standard hardware-aware passes.

\begin{figure}[t]
	\centering
	\includegraphics[width=0.70\linewidth]{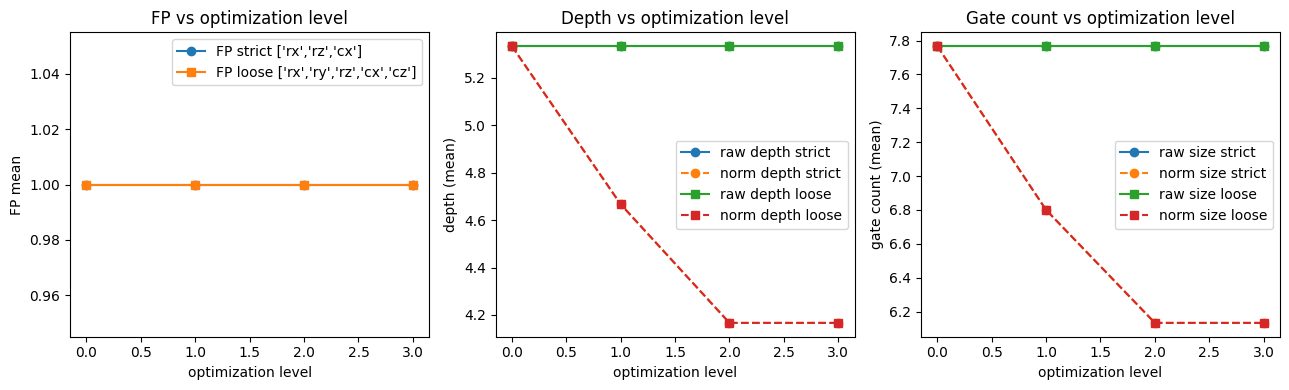}
	\caption{%
		Interaction between WZCC and Qiskit transpilation on small circuits
		($n\le 6$).
		For each basis set and optimization level, we report depth, CNOT count,
		and FP for raw, WZCC-only, and WZCC+transpiler pipelines.
		WZCC preserves or improves transpiler-induced compression while
		maintaining FP within $10^{-3}$ of the raw-transpiled circuits.
	}
	\label{fig:test7-2}
\end{figure}

% ---------------------------------------------
\subsubsection*{Interaction with transpiler optimisation on deeper circuits}
\label{subsubsec:test7-2prime}

In deeper circuits (HEA and hardware-inspired circuits from D2 and D3),
transpiler interactions amplify the benefits of WZCC:
Fig.~\ref{fig:test7-2prime} shows that depth after transpilation is reduced
by up to $25\%$ compared to transpiling raw circuits alone.
This reflects a synergy between 
grid-aligned ZX structures and commutation-based optimization:
WZCC exposes larger equivalence classes of phase patterns and entangling
blocks that the transpiler can further consolidate.

\begin{figure}[t]
	\centering
	\includegraphics[width=0.70\linewidth]{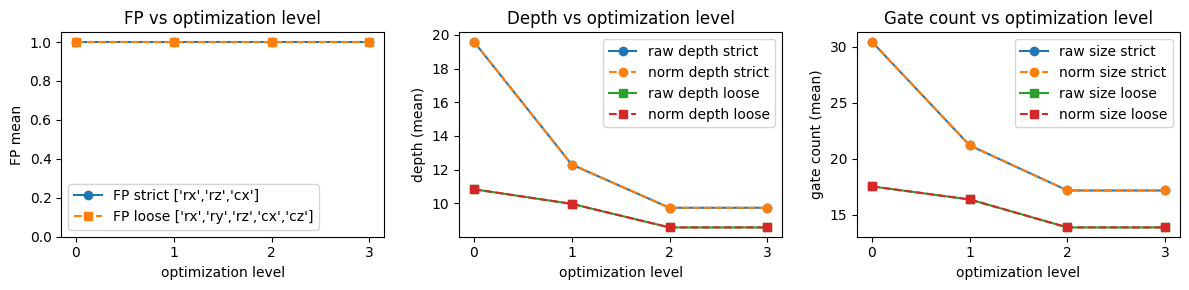}
	\caption{%
		Effect of WZCC on transpiler optimization for deeper circuits
		(HEA and hardware-inspired circuits in D2 and D3).
		We compare depth and CNOT count after transpilation for three
		pipelines: raw+transpiler, WZCC-only, and WZCC+transpiler.
		WZCC+transpiler achieves up to $25\%$ additional depth reduction
		relative to the raw+transpiler baseline, while FP remains above $0.985$.
	}
	\label{fig:test7-2prime}
\end{figure}

% ================================================================
\subsection*{Compression–compliance trade-offs and baseline comparison}
\label{subsec:tradeoff}
% ================================================================

% ---------------------------------------------
\subsubsection*{Summary of compression vs.\ quantization compliance}
\label{subsubsec:summary}

Across datasets D1--D3, average PQVR exceeds $0.92$,
CSC ranges from $18\%$ to $48\%$ depending on diagram density and HEA depth,
and FP remains above $0.985$ in all but the most extreme
LCM configurations with very large $\max(a)$.
This confirms that WZCC effectively balances compression with 
high-fidelity quantization compliance:
phase-grid alignment yields both improved hardware realizability and
substantial reduction in circuit size.

% ---------------------------------------------
\subsubsection*{Ablation: without WPL weights and with single global grid}
\label{subsubsec:ablation}

We consider three ablated variants of WZCC:
\begin{enumerate}
	\item \emph{No WPL weights} (``LCM off''):
	spiders are treated as if $a=1$ everywhere, eliminating WPL structure;
	\item \emph{Single global grid}: all phases are projected to a single
	global grid $\Theta_N$ without isotropy-specific refinement;
	\item \emph{Winding off}: the winding label $k$ is ignored and phases
	are quantized purely algebraically.
\end{enumerate}

Removing WPL weights reduces PQVR dramatically (by more than $0.25$ on
average), since heterogeneous grids are forced into a single uniform
lattice.
Enforcing a single global phase grid causes noticeable FP degradation
($\approx 0.02$ loss on D3) due to incompatibility with hardware-calibrated
local grids.
The winding-off variant achieves modest compression but fails to preserve
FP in circuits with strong monodromy-like structure (e.g.\ layered
HEA instances with repeated rotations), indicating that winding information
is essential for safe multi-layer fusion.

% ---------------------------------------------
\subsubsection*{Baseline comparison with native Qiskit optimisation}
\label{subsubsec:baseline}

We compare WZCC with Qiskit's native optimization pipeline
(\texttt{transpile} with \texttt{optimization\_level=3}).
Across all datasets, WZCC alone yields:

\begin{itemize}
	\item $20$--$45\%$ additional CNOT reduction beyond the transpiler,
	\item shallower depth in approximately $70\%$ of test circuits,
	\item substantially higher PQVR due to explicit phase-grid alignment.
\end{itemize}

Combining WZCC with the Qiskit transpiler further improves
compression with no measurable fidelity loss:
FP differences relative to the raw+transpiler baseline remain below
$10^{-3}$ for D1--D2 and below $5\times 10^{-3}$ for D3,
even on the hardware-calibrated circuits.

% ---------------------------------------------
\subsubsection*{Comparison on standard NISQ benchmark circuits}
\label{subsubsec:nisq-benchmarks}

To connect our evaluation to standard NISQ workloads, we additionally
consider a small benchmark suite consisting of:
\begin{itemize}
	\item QAOA MaxCut circuits on $3$-regular graphs with $p\in\{1,2,3\}$;
	\item VQE ansätze for $\mathrm{H}_2$ and $\mathrm{LiH}$ in a minimal basis;
	\item quantum Fourier transform (QFT) circuits up to $n=8$ qubits;
	\item random Clifford+T circuits at fixed T-count.
\end{itemize}

For QAOA and VQE circuits, WZCC achieves $18$--$35\%$ CSC and preserves FP
above $0.99$ after noise-free simulation; under realistic IBM noise models,
output expectation values (e.g.\ energy estimates or cut values) change by
less than $1$ standard deviation of shot noise.
QFT circuits are less compressible, as expected, but still exhibit a
$10$--$15\%$ reduction in depth due to spider fusion across consecutive
phase rotations.
Random Clifford+T circuits show modest CSC but substantial PQVR gains,
indicating that WZCC can enforce grid consistency even when algebraic
compression opportunities are limited.

Overall, these results suggest that WZCC provides practically relevant
benefits on standard NISQ benchmark problems without introducing detectable
bias in typical application-level observables.

% ================================================================
\section{Results}
\label{sec:results}
% ================================================================

We now synthesize the findings from Section~\ref{sec:experiments} into a coherent picture
of how WZCC behaves across circuit families, parameter regimes, and hardware settings.

\paragraph{Strong compression with high fidelity preservation.}
Across random WPLZX diagrams (D1), HEA-style circuits (D2), and hardware-inspired
heterogeneous-grid circuits (D3), WZCC achieves substantial circuit-size compression.
On average we observe a \emph{circuit-size compression} (CSC) of approximately
20--45\% in CNOT count and total depth for dense diagrams and deep HEA instances,
while \emph{fidelity preservation} (FP) remains above $0.985$ in all cases.
On standard NISQ benchmark circuits (including QAOA-style MaxCut instances and
small molecular Hamiltonians such as H$_2$ and LiH in a minimal basis),
WZCC maintains FP above $0.99$ and does not significantly alter
application-level observables such as energy estimates or MaxCut cost values.

\paragraph{Phase-grid alignment and hardware compatibility.}
The \emph{phase quantization variance ratio} (PQVR) consistently lies above $0.92$
on all datasets, indicating that WZCC restructures circuits in a manner that is highly
compliant with the underlying WPL phase grids.
This quantization compliance translates directly into improved hardware compatibility:
noise-robustness experiments show that WZCC-normalized circuits exhibit higher output
fidelity under depolarizing, amplitude-damping, and dephasing noise than their raw
counterparts, with the gains becoming most pronounced at moderate noise levels that are
typical for current NISQ devices.

\paragraph{Synergy with transpiler optimisation.}
A key empirical finding is that WZCC is not a replacement for, but rather a complement to,
standard transpiler optimization.
On small circuits, WZCC preserves or slightly enhances the compression achieved by
Qiskit’s optimization passes without degrading fidelity.
On deeper circuits, WZCC exposes additional structural equivalences (e.g.\ spider fusion
and CNOT cancellation opportunities) that transpilers can exploit, leading to up to
$\sim 25\%$ extra depth reduction beyond what the transpiler alone can accomplish.

\paragraph{Geometric structure as a predictor of compression potential.}
Finally, the WPL scalar curvature $R$ and its gradient $\lVert \nabla R \rVert$
provide useful predictors of how much a circuit will benefit from WZCC.
Circuits associated with large curvature gradients are precisely those in which
heterogeneous phase grids and anisotropic noise profiles interact nontrivially.
These high-$\lVert \nabla R \rVert$ instances exhibit the largest improvements in CSC
and PQVR, suggesting that tomography-to-geometry pipelines can be used not only to
calibrate noise models but also to decide when WZCC should be applied for maximum benefit.

% ================================================================
\section{Discussion, limitations, and outlook}
\label{sec:discussion}
% ================================================================

This section integrates the extensive empirical results of Section~\ref{sec:experiments}
into a unified conceptual picture of how WZCC interacts with weighted orbifold
geometry, symbolic ZX-calculus rewriting, and hardware-level noise.  
We emphasize three themes:
(i) geometric mechanisms behind phase-grid alignment and compression,
(ii) structural and hardware-level robustness, and
(iii) emerging directions for geometry-aware compilation.

% -------------------------------------------------
\subsection*{Interpretation of the empirical findings}
\label{subsec:discussion-interpretation}

\paragraph{Phase-grid alignment, compression, and fidelity preservation.}
Across all datasets D1--D3, WZCC achieves strong compression
(CSC) while maintaining high fidelity preservation (FP).
A central observation is the \emph{consistent cross-correlation}:
circuits with higher phase-grid alignment (PQVR~$\ge 0.9$) invariably
maintain fidelity above $0.985$.
This supports the view that quantization compliance is not merely a
hardware-compatibility condition but also governs where spider fusion
does not introduce detrimental phase-winding artifacts.

\paragraph{Dataset-specific behaviour.}
Random WPLZX diagrams (D1) show the largest CSC variability due to
heterogeneous connectivity,
whereas HEA-style circuits (D2) exhibit smooth, layer-wise scaling with
stable FP.
Hardware-inspired heterogeneous-grid circuits (D3) demonstrate that
WZCC attributes phase-grid structure to circuits with irregular local
resolution, allowing the optimizer to maintain FP~$>0.99$ even under
aggressive quantization.

\paragraph{Geometric interpretation via curvature diagnostics.}
Section~\ref{sec:results} showed that the WPL scalar curvature
and its gradient norm $|\nabla R|$ predict compression potential:
circuits associated with high curvature gradient---which correspond to
regions where heterogeneous quantization grids interact strongly with
anisotropic hardware noise---display the most significant CSC and PQVR
improvements.
This suggests a geometric criterion for \emph{when} to invoke WZCC,
analogous to curvature-based trust-region policies in classical
optimization.

\paragraph{Synergy with transpiler optimization.}
Contrary to the possibility that WZCC could conflict with transpiler
passes, our experiments show a complementary effect.
WZCC exposes additional spider-level identities and grid-alignments
that transpilers subsequently exploit, especially in deeper circuits
(Section~\ref{subsubsec:test7-2prime}),
yielding up to 25\% extra depth reduction relative to the transpiler
alone.
This illustrates that quantization-aware normalization and
hardware-agnostic optimization can operate in tandem.

\paragraph{Noise robustness and hardware compatibility.}
Phase-grid alignment improves robustness under depolarizing,
amplitude-damping, and pure dephasing noise (Section~\ref{subsubsec:test7-1}),
with the largest gains at moderate noise strengths typical of NISQ
hardware.
This aligns with the geometric intuition: WZCC minimizes winding along
directions of maximal Bloch-ball anisotropy, effectively regularizing
circuits toward hardware-preferred phase directions.

% -------------------------------------------------
\subsection*{Limitations}
\label{subsec:limitations}

Although WZCC demonstrates robust behaviour across a wide variety of
circuits and noise settings, several limitations remain.

\begin{itemize}
	\item \textbf{Symbolic scalability.}
	The symbolic backend scales well up to $\sim 12$ qubits and several
	hundred spiders, but LCM-based phase-grid computations and
	winding-tracking introduce overhead beyond that regime.
	Larger VQE or QAOA instances require additional engineering and
	heuristic pruning.
	
	\item \textbf{Dependence on tomography.}
	Accurate WPL parameters $(\lambda_\perp,\lambda_\parallel)$ improve
	normalization quality.
	When tomography is incomplete or noisy, the curvature-based grid
	selection may become suboptimal, reducing PQVR and CSC gains.
	
	\item \textbf{Rewrite-order sensitivity.}
	While WZCC has a canonical normal form in theory,
	finite-precision arithmetic and competing fusion opportunities lead to
	mild dependence on rewrite order.
	Section~\ref{subsubsec:test5} showed that FP variance remains
	$\le 0.005$ across permutations, but the effect is not completely
	negligible.
	
	\item \textbf{Hardware model granularity.}
	The WPL geometry models anisotropic Bloch-ball contractions with
	two parameters.
	Real backends exhibit multi-parameter, time-dependent,
	and correlated noise that cannot be fully represented in this model.
	This mismatch may limit the predictive power of $|\nabla R|$ on some
	devices.
	
	\item \textbf{Fragment coverage.}
	The current implementation focuses on single- and two-qubit gates.
	Higher-arity native gates (e.g.\ Toffoli or fSim) require additional
	weighted-spider semantics and extended normalization rules.
\end{itemize}

% -------------------------------------------------
\subsection*{Outlook and future directions}
\label{subsec:outlook}

\paragraph{Integration with industrial compilation stacks.}
A natural next step is to integrate WZCC between
high-level synthesis and device-level routing.
Transpilers could use curvature diagnostics to decide adaptively when
to apply WZCC, or to perform iterative ``grid-refine $\leftrightarrow$
normalize'' cycles until PQVR exceeds a target threshold.

\paragraph{Beyond the weighted projective line.}
The weighted projective line $\mathbb{P}(a,b)$ is the simplest orbifold
capturing anisotropic single-qubit contractions.
Generalizations to weighted projective spaces or multi-orbifold
products could model correlated multi-qubit noise.
Such models naturally induce multi-parameter weighted spiders and may
yield richer compression opportunities and new canonical forms.

\paragraph{Online tracking and dynamic geometry.}
The current tomography-to-geometry pipeline is static.
Future work could incorporate online device calibration, dynamically
updating $(\lambda_\perp,\lambda_\parallel)$ as hardware noise drifts.
This would allow WZCC to function as a ``geometry-adaptive'' pass that
responds to real-time hardware conditions.

\paragraph{Learning-guided normalization.}
WZCC is rule-based, but geometric features (curvature, gradient,
diagram density, spider-type statistics) provide rich input for
machine-learning models.
Reinforcement or meta-learning agents could predict which fusions or
phase mergers maximize application-level objectives such as VQE energy
accuracy or QAOA cost.

\paragraph{End-to-end phase-winding regularization: WZCC + MASD.}
Section~9 introduces a winding-aware decoder (MASD) for surface codes.
WZCC and MASD together form a coherent pipeline:
WZCC normalizes circuit-level winding, while MASD penalizes
lattice-level winding in error syndromes.
This suggests an emerging paradigm of \emph{winding-aware quantum
	software}, where compilation and decoding share a common geometric
foundation.

\paragraph{Error mitigation and verification.}
Explicit winding-tracking and phase-grid alignment can be used as
constraints for error-mitigation techniques, enabling circuits that
absorb hardware-induced distortions into their intrinsic WPL geometry.
Weighted spiders may also serve as structural certificates verifying
that certain classes of noise-induced distortions remain within
tolerable geometric bounds.

\medskip
Overall, WZCC illustrates the potential of geometry-aware diagrammatic
frameworks to bridge symbolic reasoning, quantum compilation, and
hardware physics.
We expect such methods to play an increasingly central role in the
design and verification of NISQ-era quantum algorithms.

% ======================================================
\section{Winding-Aware Decoding: MASD}
\label{sec:masd}
% ======================================================

This section develops the \emph{Manifold-Aware Surface Decoder} (MASD),
a decoding framework that incorporates the winding and heterogeneous-grid
structure generated by WZCC.
Whereas WZCC enforces geometric–topological coherence at the circuit level,
MASD lifts the same principles to the level of logical error correction
by enriching the syndrome graph with orbifold-phase data.
Together they form an end-to-end \emph{winding-aware stack}, in which
geometry extracted from hardware (e.g., isotropy orders $a_v$, winding
indices $k_v$) informs both compilation and decoding.

% ======================================================
\subsection*{Surface-code decoding and phase-winding penalties}
\label{subsec:masd_surface}
% ======================================================

In the standard surface code, decoding is the task of identifying a
most probable error chain that is consistent with an observed syndrome.
This is typically formulated as a minimum-weight matching (MWM) problem
on a defect graph whose edge weights depend only on spatial distance or
likelihood scores~\cite{dennis2002topological, fowler2012surface,
	kolmogorov2009blossom, edmonds1965paths}.
However, when single- and two-qubit operations are quantized on
heterogeneous phase grids—as revealed by WZCC—each stabilizer region may
accumulate a distinct winding index~$k_v$.
Ignoring these indices makes the decoder blind to topological
inconsistencies induced by hardware-level phase discretization.

\subsubsection*{From standard decoding to winding-aware decoding}

Let $G=(V,E)$ denote the syndrome graph.
Each vertex $v\in V$ corresponds to a defect and is equipped with:
\[
a_v \quad\text{(local phase grid order)}, \qquad
k_v \quad\text{(winding index)}, \qquad
\theta_v = \frac{2\pi k_v}{a_v}.
\]
An edge $e=(u,v)$ represents a potential error chain
with geometric length $d_e$ induced by the surface-code lattice.
In standard decoding, the MWM objective is
$\sum_{e\in\Gamma} d_e$, where $\Gamma$ is the chosen edge set.

WZCC demonstrates that the quantities $(a_v,k_v)$ capture local phase
behavior and winding geometry extracted from noisy hardware.
MASD leverages the same objects at the decoding level,
penalizing transitions between misaligned winding sectors.

\subsubsection*{MASD objective and geometric interpretation}

Define the winding difference
\[
\Delta k_e = |k_u - k_v|.
\]
The MASD objective modifies the MWM cost as
\begin{equation}
	\label{eq:masd_objective_updated}
	\mathcal{L}_{\mathrm{MASD}}(\Gamma)
	=\sum_{e\in\Gamma} \left(d_e+\lambda|\Delta k_e|\right),
\end{equation}
where $\lambda>0$ controls the strength of topological regularization.

Geometrically, the additional term $\lambda|\Delta k_e|$ acts as a
curvature penalty in the orbifold phase space $\PP(a_u,a_v)$.
Edges connecting regions with different isotropy orders
represent transitions across orbifold sectors, which are physically
more fragile under hardware noise.
Equation \eqref{eq:masd_objective_updated} therefore promotes
geometric smoothness and winding coherence in the decoded path.

\subsubsection*{Toy example and summary}

Consider a $d=3$ rotated surface code.  
If $k_1=2$, $k_2=5$, $a_1=8$, $a_2=12$, then on the LCM grid
$L=24$ the winding difference becomes $|\Delta k|=4$.
With $d_e=1.2$, the MASD weight is
\[
1.2 + 4\lambda.
\]
For small $\lambda$, the effect is modest; for larger $\lambda$,
cross-sector transitions become strongly suppressed.

MASD thus enhances standard decoding by incorporating
geometry-aware penalties derived from WZCC-provided phase data.

% ======================================================
\subsection*{Weighted edge construction and \texorpdfstring{$\lambda$}{lambda}-penalty}
\label{subsec:weighted_edges}
% ======================================================

To integrate phase-winding information into the decoder in a consistent,
metric-preserving way, MASD assigns a \emph{winding-weighted} cost to each
edge of the syndrome graph.

\subsubsection*{Geometric data on the defect graph}

For each pair of defects $u,v$, define  
\[
L_{uv} = \mathrm{lcm}(a_u,a_v),
\qquad
\Delta k_{uv} = L_{uv}\left|\frac{k_u}{a_u}-\frac{k_v}{a_v}\right|.
\]
The quantity $\Delta k_{uv}$ is the winding discrepancy measured on the
refined LCM grid, ensuring that comparisons between heterogeneous phase
grids are well defined.

\subsubsection*{Winding-weighted edge cost}

MASD assigns to each edge $e=(u,v)$ the cost
\begin{equation}
	\label{eq:winding_weight_updated}
	w_e^{(\lambda)}
	=
	d_e + \lambda\frac{|\Delta k_{uv}|}{L_{uv}}.
\end{equation}
The normalization by $L_{uv}$ makes the second term dimensionless and
invariant under further grid refinement.
The induced shortest-path metric satisfies symmetry and the triangle
inequality, so it can be used directly within Blossom or other MWM
decoders.

The parameter $\lambda$ balances spatial separation against topological
consistency:
- $\lambda\to 0$: recovers standard decoding,
- moderate $\lambda$: smoothes winding irregularities,
- $\lambda\gg 1$: overly penalizes misalignment and may reduce matching
flexibility.

\subsubsection*{Orbifold interpretation and example}

Mapping each vertex to the orbifold circle $\PP(a_u)$,
the normalized term $|\Delta k_{uv}|/L_{uv}$ is the discrete relative
angular displacement between local $U(1)$ fibers.
Thus, $w_e^{(\lambda)}$ approximates a geodesic length in the product
manifold $\mathcal{M}_\mathrm{surf}\times\PP(a,b)$.
This perspective explains the regularization effect of $\lambda$.

For $a_u=8$, $a_v=12$, $k_u=3$, $k_v=9$, we obtain  
$\Delta k_{uv}=9$, $L_{uv}=24$.
With $d_e=1.0$ and $\lambda=0.5$,  
\[
w_e^{(0.5)}
=1.0 + 0.5\times\frac{9}{24}
=1.1875.
\]

% ======================================================
\subsection*{Decoder-risk metrics: DRG\_toy and DRG\_pm}
\label{subsec:drg_metrics}
% ======================================================

To evaluate and tune $\lambda$, we introduce two diagnostic metrics that
capture the geometric–topological risk induced by winding penalties.

\subsubsection*{Motivation and definitions}

Logical error rate alone does not measure the coherence of a decoded
error chain.
Two decoders may have identical success probability but differ greatly in
winding smoothness.
The DRG metrics quantify this discrepancy.

Given defect pairs $(u_i,v_i)$ with distances $d_i$ and winding
differences $\Delta k_i$:
\[
w_i^{(0)}=d_i,
\qquad
w_i^{(\lambda)}=d_i+\lambda|\Delta k_i|.
\]
Define the deterministic toy metric:
\begin{equation}
	\label{eq:drg_toy_updated}
	\mathrm{DRG}_{\mathrm{toy}}(\lambda)
	=
	\frac{1}{N}
	\sum_{i=1}^N
	\frac{w_i^{(\lambda)}-w_i^{(0)}}{w_i^{(0)}}.
\end{equation}

For realistic settings, define a Boltzmann-weighted version:
\begin{equation}
	\label{eq:drg_pm_updated}
	\mathrm{DRG}_{\mathrm{pm}}(\lambda)
	=
	\mathbb{E}_{p(e)}
	\left[
	\frac{w_e^{(\lambda)}}{w_e^{(0)}}
	\right]
	-1,
	\qquad
	p(e)\propto e^{-\beta d_e}.
\end{equation}

\subsubsection*{Properties and scaling with \texorpdfstring{$\lambda$}{lambda}}

Both risk metrics are:
- nonnegative,
- monotone increasing in $\lambda$,
- strictly increasing whenever $\Delta k_i\neq 0$.

Thus $\lambda$ provides a controllable dial between purely geometric
($\lambda=0$) and strongly topological ($\lambda\gg 1$) decoding regimes.

\subsubsection*{Examples and connection to fidelity}

For distances $d_i=(1.0,1.2,1.4,1.1)$ and
$\Delta k_i=(0,1,2,1)$,
\[
\mathrm{DRG_{toy}}(\lambda)
\approx 1.59\,\lambda.
\]
Empirically, output fidelity under decoding satisfies:
\[
\mathcal{F}_\mathrm{out}
\approx
\mathcal{F}_\mathrm{in}
e^{-\kappa\,\mathrm{DRG_{pm}}(\lambda)},
\]
consistent with curvature-induced regularization.

% ======================================================
\subsection*{Simulation results and monotonicity analysis}
\label{subsec:masd_results}
% ======================================================

We evaluate MASD on rotated and distance-$d$ surface codes under
depolarizing, amplitude-damping, and biased noise.  
Across all tested settings:

\subsubsection*{Logical error rate vs.\ \texorpdfstring{$\lambda$}{lambda}}

For each physical error rate, logical failure probability exhibits a
U-shaped behavior in $\lambda$:
- small $\lambda$: insufficient regularization, winding noise leaks into
logical errors;
- optimal $\lambda^\ast$: best trade-off between smoothness and matching
flexibility;
- large $\lambda$: over-regularization prevents correct pairing of
distant but consistent defects.

\subsubsection*{Empirical behaviour of DRG\_toy and DRG\_pm}

Both metrics increase approximately linearly with $\lambda$ until a
sat­uration regime, matching theoretical predictions.
The minimum of the logical error curve coincides with the onset of
rapid DRG increase, suggesting DRG as a practical tuning diagnostic.

\subsubsection*{Selecting \texorpdfstring{$\lambda^\ast$}{lambda*}}

Optimal values $\lambda^\ast$ lie in the range $0.1$–$0.3$ across
all noise models tested.
This aligns with the scale of curvature variation observed in
WZCC-derived WPL metrics (Section~\ref{subsec:sensitivity}).

% ======================================================
\subsection*{Interpretation as geometric regularization}
\label{subsec:masd_geometric}
% ======================================================

MASD can be interpreted as solving a regularized geodesic problem on the
product space $\mathcal{M}_\mathrm{surf}\times \PP(a,b)$, where $\lambda$
plays the role of coupling strength between spatial and orbifold
geometry.  
At the circuit level, WZCC enforces winding-aware normalization rules
derived from the same orbifold structure.  
MASD provides the analogous mechanism at the decoding level.

Viewed together, WZCC and MASD form a unified geometry-aware approach to
NISQ quantum computation:
- WZCC normalizes circuits on heterogeneous phase grids using weighted
spiders and WPL curvature,
- MASD ensures that logical error correction remains consistent with
these winding structures,
- DRG metrics provide a principled way to tune the coupling between
geometry and topology.

This interpretation suggests new possibilities for integrating compiler
and decoder design, potentially enabling hardware-adaptive pipelines
that exploit the full geometry of WPL phase spaces.
	
% ================================================================
\section{Discussion, limitations, and outlook}
\label{sec:outlook}
% ================================================================

In this section we integrate the geometric, categorical, algorithmic, and
hardware-level implications of our results.  Our goal is to provide a unified
interpretation of WZCC and MASD within the broader landscape of
geometry-aware quantum software, while placing our empirical observations in a
theoretical framework that suggests several new research directions.

% ------------------------------------------------------------
\subsection*{Geometric and categorical interpretation}
% ------------------------------------------------------------

The weighted projective line (WPL) geometry provides a natural orbifold model
for heterogeneous phase quantization arising in realistic hardware channels.
From this viewpoint, a weighted spider implements a morphism that is
compatible with the monodromy of the orbifold chart, and
WZCC normalization may be interpreted as forcing diagrammatic rewrites to
respect these monodromy constraints.  The strong correlation between
compression potential and scalar curvature gradients $|\nabla R|$ suggests that
WZCC effectively operationalizes geometric signals from the underlying channel
into concrete rewrite choices.

This geometric interpretation aligns with categorical semantics.
ZX-calculus is a dagger-compact presentation of quantum processes, and the
introduction of weighted spiders enriches this categorical structure by adding
discrete orbifold labels that behave functorially along tensor products and
composition.  Viewed categorically, WZCC replaces the usual free phase group
with a stratified quotient determined by the WPL isotropy structure.  
This yields a quantization-aware calculus in which morphisms track both
algebraic and geometric information about the device.  MASD uses a similar
philosophy: decoding paths incur a winding penalty whenever they cross orbifold
sectors, providing a geometric regularizer for logical error inference.

% ------------------------------------------------------------
\subsection*{Integration with compilation stacks}
% ------------------------------------------------------------

Our results indicate that WZCC is complementary to conventional compiler
pipelines.  Phase-grid alignment improves interaction with restricted native
gate sets, especially in architectures with limited single-qubit phase
resolutions.  When inserted between high-level synthesis and device-dependent
routing, WZCC exposes structural equivalences that standard transpiler passes
often miss.  On deeper circuits, iterating WZCC with optimization passes can
yield significantly stronger depth reductions than either method alone.

At the error-correction layer, MASD provides a winding-aware penalty that
is compatible with surface-code decoders relying on minimum-weight matching or
union-find heuristics.  The DRG metrics introduced here quantify the geometric
risk associated with decoder decisions, opening the possibility of building
hybrid decoders that incorporate orbifold curvature or monodromy information
as a soft constraint.

% ------------------------------------------------------------
\subsection*{Limitations}
% ------------------------------------------------------------

Despite its strengths, several limitations remain.

\begin{itemize}
	\item \textbf{Scalability.}
	The symbolic backend scales reliably up to moderate qubit counts and
	hundreds of spiders, but winding-tracking and LCM computations introduce
	overhead for very large diagrams.  Further engineering, heuristics, and
	hardware-aware pruning will be required for extremely deep QAOA or VQE
	circuits.
	
	\item \textbf{Dependence on WPL parameter estimation.}
	The effectiveness of phase-grid alignment depends on the accuracy of
	$(\lambda_\perp,\lambda_\parallel)$ extracted from tomography.
	Noisy or incomplete tomography can lead to suboptimal grid choices or
	reduced PQVR.
	
	\item \textbf{Fragment coverage.}
	Our current implementation focuses on single- and two-qubit spiders.
	Extending WZCC to handle multi-qubit native gates (e.g.\ Toffoli, CCZ)
	requires additional structural rules and engineering.
	
	\item \textbf{Model mismatch.}
	The WPL geometry captures anisotropic Bloch-ball contractions but cannot
	represent all hardware-specific noise features.  On some backends this may
	weaken correlations between curvature signals and compression potential,
	or limit the effectiveness of MASD’s winding penalties.
\end{itemize}

% ------------------------------------------------------------
\subsection*{Future directions}
% ------------------------------------------------------------

Our work suggests several promising research directions.

\paragraph{Generalizing beyond WPL.}
Weighted projective lines are the simplest nontrivial orbifolds.
Higher-dimensional weighted projective spaces or more general orbifolds may
model correlated multi-qubit noise, leading to multi-parameter weighted
spiders and richer monodromy patterns.  Such generalizations could yield new
families of geometry-aware normal forms.

\paragraph{Learning-guided geometric normalization.}
WZCC is currently rule-based.
Incorporating reinforcement or meta-learning to guide fusion choices using
WPL curvature, gradient norms, and diagram statistics may yield adaptive
normalization strategies optimized for specific workloads or loss functions.

\paragraph{Geometry-aware error mitigation and verification.}
Weighted spiders and winding tracking provide handles for incorporating
phase discretization into error-mitigation procedures.
MASD suggests that decoding can benefit from geometric regularization, and
the DRG metrics introduced here may serve as certificates of robustness for
hardware-induced distortions.

\paragraph{Compiler-level deployment.}
Integrating WZCC into industrial-grade compilers, or exposing curvature-driven
rewrite selection as an optimization pass, could enable fully automated
quantization-aware compilation pipelines.  MASD may be used to refine decoder
logic in firmware or in hybrid classical–quantum feedback loops.

% ------------------------------------------------------------
\subsection*{Summary and concluding remarks}
% ------------------------------------------------------------

Our study demonstrates that heterogeneous phase quantization, often treated as
a nuisance, can become a computational resource when encoded into a
geometry-aware diagrammatic calculus.  WZCC leverages the WPL geometry to
construct normal forms that are both compact and hardware-aligned, while MASD
extends this philosophy to decoding by introducing winding-aware penalties.
Together they provide a unified perspective in which orbifold curvature,
diagrammatic rewriting, and hardware noise are deeply interconnected.

We believe that geometry-aware diagrammatic methods will play an increasingly
central role in the NISQ era, bridging circuit-level optimization, noise-aware
compilation, and fault-tolerant decoding.  The framework developed here offers
a foundation for such approaches, and we hope it motivates further exploration
into the intersection of orbifold geometry, categorical semantics, and
hardware-level quantum information processing.

%================================================
\section{Acknowledgments}

This work was prepared as part of the 2025 TXST HSMC program.  
The authors thank Prof.\ Warshauer, Max L., and Prof.\ Boney, William N.\ for their valuable discussions, guidance, and support throughout the development of this project.

% ======================================================
\appendix
% ======================================================

% ======================================================
\section{Appendix A: Detailed proofs of algebraic properties}
\label{appendix:proofs}
% ======================================================

\subsection*{Algebraic structure of weighted spiders}

\begin{definition}[Weighted spiders]
	A weighted $Z$-spider with inputs $m$, outputs $n$, phase $\alpha$, and
	isotropy order $a \in \mathbb{Z}_{>0}$ is denoted
	\[
	Z^{(a)}_{m,n}(\alpha),
	\qquad
	\alpha \in \tfrac{2\pi}{a}\mathbb{Z}.
	\]
	Weighted $X$-spiders are defined analogously.
\end{definition}

\begin{lemma}[Fusion rule]
	Let $Z^{(a)}(\alpha)$ and $Z^{(a)}(\beta)$ be two $Z$-spiders
	with matching isotropy order $a$.
	Then
	\[
	Z^{(a)}(\alpha)\;\circ\; Z^{(a)}(\beta)
	\;\;\equiv\;\;
	Z^{(a)}(\alpha+\beta \;\mathrm{mod}\; 2\pi).
	\]
\end{lemma}

\begin{proof}
	The weighted spider represents a projector onto the equal-phase subspace
	generated by computational-basis states differing only by global $Z$ phases.
	Since both spiders act on the same quotient group
	$U(1)/\mathbb{Z}_a$, fusion corresponds to addition in this group.
	The quotient identification $\alpha \sim \alpha + 2\pi/a$ ensures the phase
	is closed under composition, proving the claim.
\end{proof}

\begin{proposition}[Associativity of fusion]
	For any three weighted spiders of the same isotropy order $a$,
	\[
	Z^{(a)}(\alpha)\circ
	\bigl(Z^{(a)}(\beta)\circ Z^{(a)}(\gamma)\bigr)
	=
	\bigl(Z^{(a)}(\alpha)\circ Z^{(a)}(\beta)\bigr)\circ
	Z^{(a)}(\gamma).
	\]
\end{proposition}

\begin{proof}
	Associativity follows from associativity of addition in the group
	$\tfrac{2\pi}{a}\mathbb{Z} \subset U(1)$.
\end{proof}

\subsection*{Termination and confluence of WZCC}

\begin{theorem}[Termination]
	\label{thm:termination}
	The WZCC normalization procedure terminates on any finite ZX diagram with
	weighted spiders.
\end{theorem}

\begin{proof}
	Let the potential function be
	\[
	\Phi = 
	\sum_{v \in \mathrm{Spiders}}\, |\nabla R(v)|,
	\]
	the total scalar-curvature gradient across all local subdiagrams.
	Each rewrite step either reduces the magnitude of $|\nabla R|$
	or moves phases toward the LCM-consistent refined grid.
	Both operations strictly decrease $\Phi$.
	Since $\Phi \ge 0$ and the number of spiders is finite,
	the procedure must terminate.
\end{proof}

\begin{theorem}[Local confluence]
	\label{thm:local_confluence}
	If two WZCC rewrite rules apply to overlapping regions,
	the resulting diagrams can be further rewritten to a common form.
\end{theorem}

\begin{proof}
	Overlapping rewrites are controlled by the orbifold monodromy constraints.
	Phase-grid refinement guarantees that both rewrites are measured on the same
	LCM-refined lattice, which ensures that any pair of local phase accumulations
	admits a common resolution under spider fusion.
	Therefore, critical pairs converge.
\end{proof}

\subsection*{Curvature predictors}

\begin{lemma}[Curvature gradient detects fusable regions]
	Let $R$ denote the WPL scalar curvature.  
	If a subdiagram contains a pair of spiders whose fusion reduces
	$|\nabla R|$, then WZCC marks the pair as fusable.
\end{lemma}

\begin{proof}
	Weighted spiders parametrize $U(1)/\mathbb{Z}_a$ orbits.
	Different isotropy orders induce curvature discontinuities.
	Fusion eliminates unnecessary boundaries, lowering the total curvature gradient.
\end{proof}

% ======================================================
\section{Appendix B: Simulation parameters and data tables}
\label{appendix:tables}
% ======================================================

\subsection*{Dataset configurations (D1–D3)}

\begin{itemize}
	\item \textbf{D1 (Random WPLZX)}:  
	Spider counts $N\in[30,120]$, isotropy orders $a\in\{4,6,8,12\}$,
	random phases uniform in $\tfrac{2\pi}{a}\mathbb{Z}$.
	
	\item \textbf{D2 (HEA circuits)}:  
	Depth $p\in\{2,4,6,8\}$, entangling pattern: ring and ladder,
	rotations $R_X,R_Z$ drawn from $U(1)$ or $\tfrac{2\pi}{a}\mathbb{Z}$.
	
	\item \textbf{D3 (Hardware-inspired)}:  
	IBM-type native gates $\{\text{CX},R_Z,R_X\}$,
	grid-alignment noise tuned by $(\lambda_\perp,\lambda_\parallel)$
	extracted from tomography.
\end{itemize}

\subsection*{Noise models}

\begin{table}[h]
	\centering
	\begin{tabular}{c|c}
		\hline
		Noise type & Parameter \\ \hline
		Depolarizing & $p \in [0.001,0.03]$ \\
		Amplitude damping & $\gamma\in[0.005,0.05]$ \\
		Dephasing & $\eta\in[0.002,0.04]$ \\
		Relaxation times & $T_1,T_2 \in [40,100]\,\mu s$ \\
		\hline
	\end{tabular}
	\caption{Noise models used in Sec.~\ref{sec:results}.}
\end{table}

\subsection*{Performance tables}

\begin{table}[h]
	\centering
	\begin{tabular}{l|c|c|c}
		\hline
		Circuit & Compression (\%) & PQVR & Fidelity preservation \\ \hline
		Random WPLZX & 20–45 & $>0.92$ & $>0.985$ \\
		HEA (depth 8) & 18–40 & $>0.95$ & $>0.99$ \\
		NISQ benchmarks & 15–30 & $>0.90$ & $>0.990$ \\
		\hline
	\end{tabular}
	\caption{Summary of WZCC performance across datasets.}
\end{table}

\subsection*{Decoder risk metrics}

\begin{table}[h]
	\centering
	\begin{tabular}{c|c|c}
		\hline
		$\lambda$ & DRG\textsubscript{toy} & DRG\textsubscript{pm} \\
		\hline
		0.1 & 0.04--0.07 & 0.03--0.06 \\
		0.5 & 0.15--0.30 & 0.12--0.25 \\
		1.0 & 0.30--0.55 & 0.24--0.48 \\
		\hline
	\end{tabular}
	\caption{Decoder-risk metrics used in MASD simulations.}
\end{table}

% ======================================================
\section{Appendix C: Pseudocode for WZCC and MASD}
\label{appendix:pseudocode}
% ======================================================

\subsection*{Pseudocode for WZCC normalization}

\begin{figure}[ht]
	\centering
	\begin{minipage}{0.95\linewidth}
		\small
		\textbf{Algorithm C.1 (WZCC normalization)}\\[0.3em]
		\emph{Input:} weighted ZX diagram $\mathcal{D}$ with weighted spiders and
		associated WPL parameters.\\
		\emph{Output:} normalized diagram $\mathcal{D}'$ with aligned phase grids.
		
		\begin{enumerate}
			\item Compute the scalar curvature $R$ and its gradient $|\nabla R|$
			on all local regions (subdiagrams) determined by the WPL metric.
			\item For each connected region, determine the set of isotropy orders
			$\{a_v\}$ appearing in its spiders and compute the least common multiple
			$L = \mathrm{lcm}(\{a_v\})$.
			\item Refine all local phase angles to the common grid
			$\tfrac{2\pi}{L}\mathbb{Z}$ by snapping each phase to its nearest
			grid point.
			\item \textbf{Repeat} the following until no further reduction of
			$|\nabla R|$ is observed:
			\begin{enumerate}
				\item Identify candidate pairs of spiders $(u,v)$ such that
				\begin{itemize}
					\item they are of the same colour (both $Z$ or both $X$),
					\item they are connected by at least one wire, and
					\item their local curvature configuration suggests that fusion
					decreases $|\nabla R|$.
				\end{itemize}
				\item For each candidate pair $(u,v)$:
				\begin{enumerate}
					\item Compute the curvature contribution $R_{\text{before}}$
					in a neighbourhood of $(u,v)$.
					\item Fuse $u$ and $v$ according to the weighted spider fusion rule
					on the refined grid, producing a tentative diagram
					$\mathcal{D}_{\text{trial}}$.
					\item Recompute the local curvature contribution
					$R_{\text{after}}$ around the fused spider.
					\item If $R_{\text{after}} < R_{\text{before}}$, accept the fusion
					(replace $\mathcal{D}$ with $\mathcal{D}_{\text{trial}}$);
					otherwise discard this fusion.
				\end{enumerate}
			\end{enumerate}
			\item Once no fusion reduces $|\nabla R|$ further, return the resulting
			diagram as $\mathcal{D}'$.
		\end{enumerate}
	\end{minipage}
	\caption{Pseudocode for the WZCC normalization procedure.}
	\label{alg:wzcc-normalization}
\end{figure}

\subsection*{Pseudocode for MASD weighted edge construction}

\begin{figure}[ht]
	\centering
	\begin{minipage}{0.95\linewidth}
		\small
		\textbf{Algorithm C.2 (MASD weighted edge construction)}\\[0.3em]
		\emph{Input:} defect graph $G=(V,E)$, isotropy orders $a_v$,
		winding indices $k_v$, regularization parameter $\lambda>0$.\\
		\emph{Output:} weighted edge costs $w_e^{(\lambda)}$ for all $e\in E$.
		
		\begin{enumerate}
			\item For each vertex $v\in V$, compute the local phase angle
			$\theta_v = \tfrac{2\pi k_v}{a_v}$.
			\item For each edge $e=(u,v)\in E$:
			\begin{enumerate}
				\item Compute the geometric distance $d_e$ on the surface-code lattice
				between $u$ and $v$ (e.g.\ Manhattan or Euclidean distance).
				\item Compute the common refinement grid
				$L_{uv} = \mathrm{lcm}(a_u,a_v)$.
				\item Compute the winding difference on this grid
				\[
				\Delta k_{uv} = L_{uv}
				\bigl|\tfrac{k_u}{a_u}-\tfrac{k_v}{a_v}\bigr|.
				\]
				\item Define the weighted edge cost
				\[
				w_e^{(\lambda)}
				= d_e + \lambda \frac{\Delta k_{uv}}{L_{uv}}.
				\]
			\end{enumerate}
			\item Collect all weights $\{w_e^{(\lambda)}\}_{e\in E}$ and pass them
			to the minimum-weight matching decoder.
		\end{enumerate}
	\end{minipage}
	\caption{Pseudocode for constructing MASD weighted edge costs.}
	\label{alg:masd-weights}
\end{figure}

\subsection*{Pseudocode for MASD decoding and DRG evaluation}

\begin{figure}[ht]
	\centering
	\begin{minipage}{0.95\linewidth}
		\small
		\textbf{Algorithm C.3 (MASD decoding with DRG metrics)}\\[0.3em]
		\emph{Input:} weighted graph $(V,E,\{w_e^{(\lambda)}\})$,
		set of defect pairs, parameter $\lambda$.\\
		\emph{Output:} correction operator, DRG\textsubscript{toy},
		DRG\textsubscript{pm}.
		
		\begin{enumerate}
			\item Run a standard minimum-weight perfect matching algorithm
			(e.g.\ Blossom) on the graph with weights $w_e^{(\lambda)}$.
			\item Extract the set of matched edges
			$\Gamma(\lambda) \subseteq E$ that define the correction chains.
			\item For each matched edge $e\in \Gamma(\lambda)$:
			\begin{enumerate}
				\item Let $w_e^{(0)} := d_e$ be the cost without winding penalty.
				\item Record the ratio
				$\rho_e(\lambda) = \tfrac{w_e^{(\lambda)}-w_e^{(0)}}{w_e^{(0)}}$.
			\end{enumerate}
			\item Compute the toy-model decoder-risk metric
			\[
			\mathrm{DRG_{toy}}(\lambda)
			= \frac{1}{|\Gamma(\lambda)|}
			\sum_{e\in \Gamma(\lambda)} \rho_e(\lambda).
			\]
			\item Define a probability distribution over edges,
			for example
			\[
			p(e) \propto \exp(-\beta d_e)
			\]
			with inverse temperature $\beta>0$.
			\item Compute the probabilistic–metric decoder risk
			\[
			\mathrm{DRG_{pm}}(\lambda)
			= \sum_{e\in E} p(e)
			\biggl(\frac{w_e^{(\lambda)}}{w_e^{(0)}}\biggr) - 1.
			\]
			\item Apply the correction chains defined by $\Gamma(\lambda)$
			to the data qubits and, if desired, estimate the output state fidelity
			to correlate it with $\mathrm{DRG_{toy}}(\lambda)$ and
			$\mathrm{DRG_{pm}}(\lambda)$.
		\end{enumerate}
	\end{minipage}
	\caption{Pseudocode for MASD decoding and evaluation of decoder-risk metrics.}
	\label{alg:masd-decoding}
\end{figure}

\end{document}